\def\year{2018}\relax
\documentclass[letterpaper]{article} %DO NOT CHANGE THIS
\usepackage{aaai18}  %Required
\usepackage{times}  %Required
\usepackage{helvet}  %Required
\usepackage{courier}  %Required
\usepackage{url}  %Required
\usepackage{graphicx}  %Required
\usepackage{misc}

\let\endproof\relax
\usepackage{amsthm}
\usepackage{amssymb,amsmath}
\usepackage{stmaryrd}
\usepackage{xspace}
\usepackage{enumitem}
\usepackage{calc}
\usepackage{color}
\usepackage{bbm}
\usepackage{tikz}
\usetikzlibrary{fit,calc}
\usepackage{wrapfig}
\usetikzlibrary{arrows,decorations.pathmorphing,backgrounds,positioning,fit,petri}
\allowdisplaybreaks

\renewcommand{\vec}[1]{\ensuremath{\textbf{#1}}\xspace}

\newcommand{\ALCN}{\ensuremath{\mathcal{ALCN}}\xspace}

\newcommand{\SQ}{\ensuremath{\mathcal{SQ}}\xspace}

\newcommand{\SOQ}{\ensuremath{\mathcal{SOQ}}\xspace}

\newcommand{\nb}[1]{\textcolor{red}{$\blacktriangleright$}\footnote{\textcolor{blue}{#1}}}
\newcommand{\nbmute}[1]{\textcolor{red}{(!)}}

\newcommand{\roots}{\ensuremath{F_r}\xspace}
\newcommand{\Rone}{\ensuremath{\mathbf{R_1}}\xspace}
\newcommand{\Rtwo}{\ensuremath{\mathbf{R_2}}\xspace}
\newcommand{\Rthree}{\ensuremath{\mathbf{R_3}}\xspace}

\newcommand{\qnrleq}[3]{\ensuremath{(\leqslant #1 \ #2 . #3)}}
\newcommand{\qnrgeq}[3]{\ensuremath{(\geqslant #1 \ #2 . #3)}}
\newcommand{\qnrsim}[3]{\ensuremath{(\sim#1 \ #2 . #3)}}
\newcommand{\qnreq}[3]{\ensuremath{(= #1 \ #2 . #3)}}

\begin{document}

\title{Answering Regular  Path Queries over \SQ Ontologies}
% \title{On Query Answering in Description Logics with Number Restrictions on Transitive Roles\thanks{Guti\'errez-Basulto was funded by  the EU's Horizon 2020 programme under the Marie Sk\l{}odowska-Curie grant  No 663830}}
% 
\author{Paper ID 2844}
 \author{V\'ictor Guti\'errez-Basulto \\ Cardiff University, UK\\ gutierrezbasultov@cardiff.ac.uk \\
  \And Yazm\'in Ib\'a\~nez-Garc\'ia\\TU Wien, Austria\\ yazmin.garcia@tuwien.ac.at\\
  \And Jean Christoph Jung\\ Universit\"at Bremen, Germany\\ jeanjung@uni-bremen.de}
 \setcounter{secnumdepth}{2}

\newtheorem{example}{Example}
\newtheorem{definition}{Definition}
\newtheorem{theorem}{Theorem}
\newtheorem{lemma}{Lemma}
\newtheorem{proposition}{Proposition}
\newtheorem{claim}{Claim}
\newtheorem{corollary}{Corollary}
\newtheorem{remark}{Remark}
\newtheorem{conjecture}{Conjecture}
\newtheorem{notation}{Notation}

\maketitle

\begin{abstract}
  We study query answering in the description logic \SQ supporting
  qualified number restrictions on both transitive and non-transitive
  roles. Our main contributions are a tree-like model property for \SQ
  knowledge bases and, building upon this, an optimal automata-based
  algorithm  for answering positive existential regular path
  queries in \TwoExpTime.
% 
%   While the problem of answering queries in  expressive description
%   logics (DLs)  has been extensively investigated, (to our
%   knowledge) there have been no attempts to  analyse   DL languages
%   supporting  number restrictions on transitive roles in this
%   context. We provide here  the first result of this kind, by
%   considering the problem of answering regular path queries in the
%   DL \SQ  (with number restrictions on transitive roles). %As our
%   main result, we devise a decision procedure, yielding  a tight
%   2\ExpTime upper bound.   
\end{abstract}

%\section{Introduction}
%\begin{table*}[htbp]
%\begin{center}
%\begin{tabular}{ccccc}
%$\mathcal{SHIQ}$ & $\mathcal{SHOQ}$ & $\Smc$ & $\SQ$ & $\mathcal{SH}$ %$\mathcal{SHOIQ}$ & \textcolor{blue}{$\mathcal{SHQ}$}
%\\[1mm] \hline\\[1mm]
%%
%2\ExpTime-c & 2\ExpTime-c& in 2\ExpTime$^{a,c}$& in 2\ExpTime$^{a,c}$ &  2\ExpTime-c \\%2\NExpTime-dec. & \ExpTime\\
%%
%$\geq$ from $\mathcal{ALCI}$ $^b$, $\leq^a$  & $\geq$ from \ALCO$^d$, $\leq^c$&$\geq$ \textsc{co}\NExpTime $^e$& $\geq$ \textsc{co}\NExpTime$^e$ & $\geq ^e$, $\leq^{a,c}$\\ %&no trans in query&no trans in query\\
%\hline
%\end{tabular}
%\end{center}
%\caption{Overview of results for {\it restricted-counting} case: {\footnotesize $^a$\cite{GlimmLHS08}, $^b$\cite{Lutz08}, $^c$\cite{GlimmHS08}}\\ $^d$~\cite{NgoOS16}, $^e$\cite{EiterLOS09} }
%  \label{tab:overview}
%\end{table*}

\section{Introduction}\label{sec:intro} The use of ontologies to access data has gained
a lot of popularity  in  various research fields such as knowledge
representation and reasoning, and databases. 
%The appeal of the ontology-based data access (OBDA) paradigm comes
%from the fact that, e.g.,  data can be queried through the lens of an
%ontology so that queries can be formulated  solely using the terms in
%the ontology, and thus making  the formulation of queries much
%simpler. %In this setting, the ontology takes care of capturing the
%meaning of such domain relevant terms and their relationship.  
In the ontology-based data access (OBDA)  scenario, ontologies are
often encoded using description logic languages (DLs); as a
consequence, a large amount of research on  the query answering
problem (QA) over DL ontologies  has been conducted. %so that a
%plethora of complexity results and algorithms for both lightweight and
%expressive DLs is currently available.  
In particular,  several
efforts have been put into the study of the query answering problem in
DLs featuring \emph{transitive roles} and \emph{number restrictions}
\cite{GlimmHS08,GlimmLHS08,EiterLOS09,CalvaneseEO09,CalvaneseEO14}.
However, in all these works  %DLs heavily restrict the interaction between these
%two features, or altogether 
the application of number restrictions to transitive
roles is forbidden. 
% (see~\cite{Lutz08} for an orthogonal effort based on
% restricted models). 
This is also reflected in the fact that the W3C ontology language OWL 2
does not allow for this
interaction.\footnote{https://www.w3.org/TR/webont-req/} Unfortunately, this
comes as a shortcoming in crucial DL application areas like medicine
and biology in which many terms   are defined and
classified according to the number of components they \mn{contain} or
 \mn{have} as a  \mn{part}, in a transitive
sense~\cite{Wolstencroft2005,RectorR06,StevensAWSDHR07}. For instance,
 the ontology  $\Tmc$ below  describes that the human heart  has as a part (\mn{hPt}) exactly one mitral valve (\mn{MV}), a left atrium (\mn{LA})
 and a left ventricle (\mn{LV}); and the latter two (enforced to be
 distinct)  also have as a part a  mitral valve. Thus, the left atrium and left ventricle have to share the mitral valve. 
%\textcolor{blue}{Heart}
\begin{align*} \Tmc=\{ \ \mn{Heart} \sqsubseteq (= 1  \ \mn{hPt} .  \mn{MV})    \sqcap \exists \mn{hPt}. \mn{LA}
    \sqcap   \exists \mn{hPt.} \mn{LV},& \\
  \mn{LV}\sqcap \mn{LA} \sqsubseteq \bot, \quad \mn{LV} \sqsubseteq
  \exists \mn{hPt}. \mn{MV}, \quad  \mn{LA} \sqsubseteq \exists
  \mn{hPt}. \mn{MV} \ \}.&
  \end{align*}
The lack of investigations of query answering in DLs of this kind
is partly because  $(i)$~the interaction of these features with other 
traditional constructors often
leads to undecidability of the standard reasoning tasks (e.g.,
satisfiability)~\cite{HorrocksST00};  and
$(ii)$~for those DLs known to be decidable, such as \SQ and \SOQ
~\cite{KazakovSZ07,KaminskiS10}, only recently  tight complexity
bounds were obtained~\cite{GuIbJu-AAAI17}. Moreover, %even if these
%features (with restricted interaction) do not necessarily increase
%the complexity of QA, they do 
these features, even with restricted interaction, pose additional challenges for devising
decision procedures %~\cite{GlimmHS08,GlimmLHS08,EiterLOS09} 
since they
lead to the loss of properties, such as the tree model property,
which make the design of algorithms for QA simpler. %In fact, these
%difficulties are present already  in DLs with transitivity, but
%without number restrictions~\cite{EiterLOS09}. 
Clearly, these issues
are exacerbated if number restrictions are imposed on transitive
roles. 
 
Traditionally, most of the research in OBDA has focused on answering
conjunctive queries. However, \emph{navigational} queries have
recently gained a lot of attention~\cite{StefanoniMKR14,BienvenuOS15,BagetBMT17}
since they are key in various applications. For instance, in
biomedicine they are used to retrieve specific paths from protein,
cellular and disease networks~\cite{Dogrusoz2009,Lysenko2016}. A
prominent class of navigational queries is that of \emph{regular path
queries}~\cite{FlorescuLS98}, where paths are specified by a
regular expression. Indeed, motivated by applications in the semantic
web,  the latest W3C standard SPARQL 1.1   includes property paths, related
to regular expressions.  

The objective of this paper is to start the research on query
answering in DLs supporting qualified number restrictions over
transitive roles. We study the entailment problem of \emph{positive
existential two-way regular path queries}~\cite{CalvaneseGLV00} over
\SQ ontologies, thus generalizing both conjunctive and regular path
queries.  To this end, we pursue an \emph{automata-based approach for
query answering} using two-way  alternating  tree automata
(2ATA)~\cite{Vardi98}. This roughly consists of three
steps~\cite{CalvaneseEO14}: ($i$) show that, if a query $\varphi$ is
not entailed by the knowledge base \Kmc, there is a \emph{tree-like}
interpretation witnessing this, ($ii$) devise an automaton $\Amf_\Kmc$
which accepts precisely the tree-like interpretations of $\Kmc$,
($iii$) devise an automaton $\Amf_\varphi$ which accepts a tree-like
interpretation iff it satisfies $\varphi$.  Query entailment is then
reduced to the question whether $\Amf_{\Kmc}$ accepts a tree that is
not accepted by $\Amf_{\vp}$. In this paper, we significantly adapt
and extend each step to \SQ, resulting in an algorithm running in
\TwoExpTime, even for \emph{binary} coding of numbers. A matching
lower bound follows from positive existential QA in
\ALC~\cite{CalvaneseEO14}.  More precisely, for step~($i$) we develop
the notion of \emph{canonical tree decompositions} which intuitively
are tree decompositions tailored to handle the interaction of
transitivity and number restrictions. We then show via a novel
unraveling operation for \SQ that, if the query is not entailed, there
is a witness interpretation which has a canonical tree decomposition
of width bounded exponentially in the size of $\Kmc$, cf.\ Section~3.
These canonical tree decompositions are crucial in order to construct
a small 2ATA $\Amf_\Kmc$ in step~($ii$), which is done in Section~4.1.
For step~($iii$), we propose in Section~4.2 a novel technique for
answering regular path queries directly using a 2ATA $\Amf_{\vp}$
since a naive application of the techniques from~\cite{CalvaneseEO14}
does not lead to optimal complexity, because of the large width of the
decompositions.

\newcommand{\wit}{\ensuremath{\mn{Wit}}\xspace}
\newcommand{\witgt}{\ensuremath{Q}\xspace}

\section{Preliminaries}\label{sec:preli}

{\bf Syntax.} We consider a vocabulary consisting of
countably infinite disjoint sets of \emph{concept names} $\mn{N_C}$,
\emph{role names} $\mn{N_R}$, and \emph{individual names} $\mn{N_I}$,
and assume that $\mn{N_R}$ is partitioned into two countably infinite
sets of \emph{non-transitive role names} $\mn{N}_\mn{R}^{nt}$ and
\emph{transitive role names} $\mn{N}_\mn{R}^{t}$. The syntax of
\emph{\SQ-concepts $C,D$} is given by the rule 
 $$C,D ::= A\mid \neg C \mid C \sqcap D \mid \qnrleq n r C $$ 
where $A \in \mn{N_C}$,
$r \in \mn{N_R}$, and $n$ is a number given in binary.  We use $\qnrgeq
n r C$ as an abbreviation for $\neg \qnrleq {n-1} r C$, and other
standard abbreviations like $\bot$, $\top$, $C\sqcup D$, $\exists
r.C$, $\forall r.C$. Concepts of the form $\qnrleq n  r  C $ and
$\qnrgeq n  r  C$ are called \emph{at-most restrictions} and \emph{at-least restrictions}, respectively.

An  \emph{\SQ-TBox (ontology) \Tmc} is a finite set of \emph{concept inclusions} $C\sqsubseteq D$  where  $C,D$ are \SQ-concepts.  An
\emph{ABox} is a finite set of \emph{concept} and \emph{role
assertions} of the form $A(a)$, $r(a,b)$ where $A \in \mn{N_C}$, $r
\in \mn{N_R}$ and $\{a,b\} \subseteq \mn{N_I}$;
$\mn{ind}(\Amc)$ denotes the set of individual names occurring in
\Amc. A \emph{knowledge base (KB)} is a pair $\Kmc=(\Tmc, \Amc)$. 

\smallskip \noindent {\bf Semantics.}  An \emph{interpretation} $\Imc
= (\Delta^\Imc, \cdot^\Imc)$ consists of a non-empty \emph{domain}
$\Delta^\Imc$ and an \emph{interpretation function} $\cdot^\Imc$
mapping concept names to subsets of the domain and role names to
binary relations over the domain such that  transitive role names are
mapped to transitive relations.  The interpretation function is
extended to complex concepts by defining $(\neg
C)^\Imc=\Delta^\Imc\setminus C^\Imc$, $(C\sqcap D)^\Imc=C^\Imc\cap
D^\Imc$, and 
\vspace{-1mm} $$ \qnrleq n  r  C^\Imc = \{d\in \Delta^\Imc \mid \
  |\{e\in C^\Imc\mid (d,e)\in r^\Imc\}| \leq n \}.$$
For ABoxes $\Amc$ we adopt the \emph{standard name assumption (SNA)},
that is, $a^\Imc =a$, for all $a \in \mn{ind}(\Amc)$, but we strongly
conjecture that our results hold without it.
The satisfaction relation $\models$ is defined as usual by taking
$\Imc\models C\sqsubseteq D$ iff $C^\Imc\subseteq D^\Imc$, $\Imc
\models A(a)$  iff $a \in A^\Imc$, and $\Imc \models r(a,b)$ iff
$(a,b) \in r^\Imc$.  An interpretation $\Imc$ is a \emph{model} of a
TBox \Tmc, denoted $\Imc\models\Tmc$, if $\Imc\models \alpha$ for all
$\alpha\in\Tmc$; it is a model of an ABox \Amc, written $\Imc \models
\Amc$, if $\Imc\models \alpha$ for all $\alpha\in\Amc$; it is a model
of a KB \Kmc if $\Imc \models \Tmc$ and $\Imc \models \Amc$.

\smallskip \noindent\textbf{Query Language.} A \emph{positive
existential regular path query (PRPQ)} is a formula
$\varphi=\exists\vec{x}\, \psi(\vec{x})$ where $\psi$ is constructed
using $\wedge$ and $\vee$ over atoms of the form $\Emc(t,t')$ where
$t,t'$ are variable or constant names, $\Emc$ is a regular
expression over $\{r,r^-\mid r\in \mn{N}_\mn R \}\cup \{A? \mid A \in
\mn{N_C}\}$, and the tuple $\vec{x}$ denotes precisely the free
variables in $\psi$. Note that atoms $A(t)$ are captured using
$A?(t,t)$. 
% and that the restriction to queries without answer
% variables is without loss of generality. 

We denote with $I_{\varphi}$ the set of constant names in $\varphi$. A
\emph{match for $\vp$ in \Imc} is a function $\pi:\xbf\cup
I_{\varphi}\to \Delta^\Imc$ such that $\pi(a)=a$, for all $a\in
I_{\varphi}$ and $\Imc,\pi\models\psi(\xbf)$ under the standard
semantics of first-order logic extended with the following rule for
atoms of the form $\Emc(t,t')$: $\Imc,\pi\models \Emc(t,t')$ if there
is a word $\nu_1\cdots\nu_n\in L(\Emc)$ and a sequence
$d_0,\ldots,d_n\in\Delta^\Imc$ such that $d_0 = \pi(t), d_n =
\pi(t')$, and for all $i \in [1,n]$ we have that $(i)$ if $\nu_i =
A?$, then $d_{i-1}=d_i \in A^\Imc$, and $(ii)$ if $\nu_i = r$ (resp.,
$\nu_i=r^-$), then $(d_{i-1}, d_i) \in r^\Imc$ (resp.,
$(d_i,d_{i-1})\in r^\Imc$).  A query $\vp$ is \emph{entailed by a KB
\Kmc}, denoted as $\Kmc\models \varphi$, if there is a match for $\vp$
in every model \Imc of \Kmc.  The \emph{query entailment problem} asks
whether a KB \Kmc entails a PRPQ~$\vp$. It is well-known that the
\emph{query answering problem} can be reduced to query entailment, and
that PRPQs are \emph{preserved under homomorphisms}, that is, if
$\Imc\models\vp$ and there is a homomorphism from $\Imc$ to \Jmc, then
also $\Jmc\models\vp$.

\smallskip\noindent \textbf{Additional Notation for Transitive Roles.}
Given some interpretation $\Imc$,  $\Imc|_\Delta$ denotes  the
restriction of $\Imc$ to domain $\Delta\subseteq \Delta^\Imc$. For
$d\in \Delta^\Imc$ and $r \in \mn{N}_\mn{R}^{t}$, the
\emph{$r$-cluster of $d$ in \Imc}, denoted by $Q_{\Imc,r}(d)$, is the
set  containing  $d$ and all elements $e\in\Delta^\Imc$ such that both
$(d,e)\in r^\Imc$ and $(e,d)\in r^\Imc$.  We call a set $\abf\subseteq
\Delta^{\Imc}$ an \emph{$r$-cluster in $\Imc$} if $\abf=Q_{\Imc,r}(d)$
for some $d\in \Delta^\Imc$, and an \emph{$r$-root cluster} if
additionally $(d,e)\in r^\Imc$ for all $d\in\abf$ and
$e\in\Delta^\Imc\setminus\abf$. Note that both a single element
without an $r$-loop and a single element with an $r$-loop are
$r$-clusters of size~1; otherwise $r$-clusters can be viewed as
$r$-cliques.

\newcommand{\bagz}{\ensuremath{\Imf}\xspace}
\newcommand{\bago}{\ensuremath{\mathfrak{r}}\xspace}

\section{Tree Decompositions} \label{sec:tableaux}
Existing algorithms for QA in expressive DLs, e.g., \SHIQ (without
number restrictions on transitive roles), exploit the fact that for
answering queries it suffices to consider \emph{canonical models} that
are forest-like, roughly consisting of an interpretation of the ABox
and a collection of tree-shaped interpretations whose roots are elements of
the ABox. We start with showing that for \SQ this tree-model property
is lost.

\begin{example}\label{ex:nontree} 
  %
 % Let $\Tmc =\{ A \sqsubseteq (\leq 1\ r \ C)   \ \sqcap \ \exists r.B
  %\   \sqcap  \ \exists r.\neg B, \top \sqsubseteq \exists r.C \}$
  %with $r \in \mn N_\mn R ^ t$. 
  The number restrictions in  $\Tmc$, cf.\ Section~\ref{sec:intro}, 
  force that every model of $\Tmc$ satisfying $\mn{Heart}$ contains the
  structure in Fig.~\ref{fig:diamond}(a).
  Moreover, in \SQ clusters can be enforced.  Let $\Tmc'$ be the
  following TBox, where $r \in \mn N_\mn R ^ t$:
  $$\{ A
  \sqsubseteq \qnreq 3 r  B, B \sqsubseteq \qnreq 3  r  B, A\sqsubseteq
\neg B \}.$$ Then, in every model of
$\Tmc'$, an element satisfying $A$ roots the structure depicted in
Fig.~\ref{fig:diamond}(b), where the elements satisfying $B$ form an
$r$-cluster.
\end{example}
Nevertheless, we will establish a \emph{tree-like} model property for
\SQ, showing that it suffices to consider such models for query
entailment. We first introduce  a basic form of tree
decompositions suited for transitive roles. A \emph{tree} is a prefix-closed subset $T\subseteq
(\mathbbm{N}\setminus\{0\})^*$. A node $w\in T$ is a successor of
$v\in T$ and $v$ is a predecessor of $w$ if $w=v\cdot i$ for some
$i\in \mathbbm{N}$. We denote with $w\cdot -1$ the predecessor of
$w$, if it exists.

\begin{definition}\label{def:treedecomp}
  A \emph{tree decomposition of an interpretation~\Imc} is pair
  $(T,\Imf)$ where $T$ is a tree and $\Imf$ is a function that assigns
  an interpretation $\Imf(w)=(\Delta_w,\cdot^{\Imf(w)})$ to every
  $w\in T$, and  the following conditions are satisfied:
  \begin{enumerate}

    \item\label{it:tree1} $\Delta^\Imc=\bigcup_{w\in T}\Delta_w$;

    \item\label{it:tree2} for every $w\in T$, we have
      $\Imf(w)=\Imc|_{\Delta_w}$;

    \item\label{it:tree4} $r^\Imc = \chi_r$ for $r\in \mn{N}_{\mn{R}}^{nt}$ and
      $r^\Imc=\chi_r^+$ for $r\in \mn{N}_\mn{R}^t$,
      where\\\vspace{-3mm}
      $$\chi_r = \textstyle \bigcup_{w\in T} r^{\Imf(w)};$$

    \item\label{it:tree3} for every $d\in \Delta^\Imc$, the set
      $\{w\in T\mid d\in \Delta_w\}$ is connected in $T$.

\end{enumerate}
The \emph{width} of $(T,\Imf)$ is the maximum domain size of
interpretations that occur in the range of $\Imf$ minus 1,
that is, $\sup_{w\in T}|\Delta_w| -1$.  
Its \emph{outdegree} is the outdegree of $T$.  
\end{definition}
\begin{figure}[t]
\centerline{
  \begin{tikzpicture}[>=latex,thin,point/.style={circle,draw=black,minimum
    size=1.4mm,inner sep=0pt},scale=1.0, bend angle=45]\footnotesize 
    \node[point, label=left:{\scriptsize $\mn{Heart}$}] (x) at (0,1) {}; 
    \node[point, label=below:{\scriptsize $\mn{LV}$}] (y) at (1, 0.6) {}; 
    \node[point, label=above:{\scriptsize $\mn{LA}$}] (z) at (1, 1.6) {};
    \node[point, label=right:{\scriptsize $\mn{MV}$}] (d) at (2,1) {};
\draw[->] (x) -- node[midway, below]{} (y); 
\draw[->] (x) -- node[midway, above]{} (d); 
\draw[->] (x) -- node[midway, above]{} (z); 
\draw[->] (y) -- node[midway, below] {} (d); 
\draw[->] (z) -- node[midway, above]{}(d);
% 
% DL display
%    \node[point,label=left:{\scriptsize $A$}] (x) at (0,-1.3) {};
%     \node[point,label=below:{\scriptsize $B\ $}] (y) at (0.8, -1.3) {};
%      \node[point,label=above:{\scriptsize $B$}] (z) at (1.4,-.3) {};
%       \node[point,label=below:{\scriptsize $B$}] (z') at (1.4, -2.3) {};
%
 \node[point,label=left:{\scriptsize $A$}] (x) at (3.5, 1) {};
     \node[point,label=below:{\scriptsize $B\ $}] (y) at (4.3, 1) {};
      \node[point,label=above:{\scriptsize $B$}] (z) at (4.9,2) {};
       \node[point,label=below:{\scriptsize $B$}] (z') at (4.9, 0) {};
\draw[->] (x) -- node[midway, above]{} (y);
 \draw[->](x) -- node[midway, above]{} (z);
 \draw[->] (x) -- node[midway, below]{} (z');
  \path
    (y) edge [loop right] node {} (y)
    edge[<->] node[right] {} (z)
    (z) edge [loop right] node {} (z)
    edge[<->] node[right] {} (z')
    (z') edge [loop right] node {} (z')
    edge[<->] node[right] {} (y)
    ;
    \node at (1,-0.5) {(a)};
    \node at (4.3,-0.5) {(b)};
\end{tikzpicture}}\vspace{-1em} \caption{Example~\ref{ex:nontree}}\label{fig:diamond}
\end{figure}
\noindent Unfortunately, this basic tree decomposition
does not yet enable tree automata to count over transitive roles (with
a small number of states) since the $r$-successors of an element, say
$d \in \Delta^\Imc$, are scattered in the decomposition;  see  Section~\ref{sec:KBA} for further details.
To address this, we extend tree decompositions with a third component
$\bago$ which assigns to every node $w\in T\setminus\{\varepsilon\}$ a
role name $\bago(w)$ and $\bot$ to the root $\varepsilon$.
Intuitively, a node labeled with $r=\bago(w)$ is responsible
for capturing $r$-successors of some element(s) in the predecessor of~$w$.

We need some additional notation. Let $(T,\bagz,\bago)$ be such an
extended tree decomposition, and let $w\in T$ and $r\in\mn{N}_\mn{R}$.
We say that $d\in \Delta_w$ \emph{is fresh in $w$} if $w=\varepsilon$
or $d\notin\Delta_{w\cdot -1}$, and \emph{$r$-fresh in $w$} if 
$r=\bago(w)$ and it is
either fresh or $r\neq \bago(w\cdot -1)$.  We denote with
$F(w)$ and $F_r(w)$ the set of all fresh and $r$-fresh elements in~$w$, respectively. Intuitively, $F_r(w)$ contains all elements
which are allowed to have fresh $r$-successors in the successor nodes
of $w$. Indeed, the following stronger form of tree decompositions
implies (among other things) that, for all $d$ and $r$, there is a
unique $w$ with $d\in F_r(w)$.
\begin{definition}\label{def:candec}
  An extended tree decomposition $\Tmf=(T,\Imf,\bago)$ of an
  interpretation \Imc is \emph{canonical} if the following conditions
  are satisfied for every $w \in T$ with $r=\bago(w)$ and every
  successor $v$ of $w$ with $s=\bago(v)$:
  \begin{enumerate}[label=(C$_\arabic*$),leftmargin=*,align=left]

    \item\label{it:can1} if $(d,e)\in s_1^{\Imf(v)}$, then $s_1=s$, or
      $d=e$ and $s_1\!\in\!\mn{N}_\mn{R}^{t}$;

    \item\label{it:can2} if $s\in \mn{N}_\mn{R}^{nt}$, then $\Delta_{v}=\{d,e\}$, for some
      $d\in F(w)$, $e\in F(v)$, and $s^{\Imf(v)}=\{(d,e)\}$;

    \item \label{it:can3} if $s\in\mn{N}_\mn{R}^t$ and $r\notin\{\bot, s\}$,
      there are $d\in F(w)$ and an $r$-root cluster \abf in $\Imf(v)$
      such that $\Delta_w\cap \Delta_{v}=\{d\}$ and $d\in\abf$;
      moreover, there is no successor $v'\neq v$ of $w$ satisfying
      this for $d$ and $\bago(v') = s$;

    \item \label{it:can4}if $s\in\mn{N}_\mn{R}^{t}$ and $r\in\{\bot,s\}$, then there
      is an $s$-root cluster $\abf$ in $\Imf(v)$ with: 

    \begin{enumerate}

      \item $\abf\subseteq F_{s}(w)$;

      \item \abf is an $s$-cluster in $\Imf(w)$;

      \item for all $d\in \abf$ and $(d,e)\in s^{\Imf(w)}$, we have
	$e\in\Delta_v$; 

      \item for all $(d,e)\in s^{\Imf(v)}$, $d\in \abf\cup F(v)$ or
	$e\notin F(v)$.

    \end{enumerate}

\end{enumerate}

\end{definition}
Definition~\ref{def:candec} imposes restrictions on the structural
relation between interpretations at neighboring nodes.
Condition~\ref{it:can1} expresses that the interpretation at a node
labeled with $\bago(w)=r$ interprets essentially only $r$ non-empty
(among role names).
Condition~\ref{it:can2} is in analogy with standard unravelling over
non-transitive roles~\cite{handbookDL}. Condition~\ref{it:can3} reflects that
interpretations at neighboring nodes with different $\bago$-components do only interact via single elements.  Most
interestingly, Condition~\ref{it:can4} plays the role of~\ref{it:can2}, but for
transitive roles. %It is important to 
Note that~\ref{it:can4} is based on
$r$-clusters  since they can be enforced, see
Example~\ref{ex:nontree} above. %Item~\textit{(a)} restricts via which
%elements $\abf$ neighboring interpretations can interact. Indeed,
%Item~\textit{(b)} requires that $\abf$ is an $s$-cluster in the
%interpretation $\Imf(w)$. Item~\textit{(c)} states that
%everything reachable from \abf via $s$ in $\Imf(w)$ should
%be also included in $\Imf(v)$; finally, Item~\textit{(d)} requires
%that there are no connections $(d,e)\in r^{\Imf(v)}$ between elements
%$d$ from $\Delta_{w}$ and fresh elements $e$ in $v$.

\subsection{Tree-like Model Property  for \SQ}

As our first main result, we show a \emph{tree-like model property},
in particular, that every model can be \emph{unraveled} into a
canonical decomposition of small width.  
% This is of independent
% interest since existing decidability results (for satisfiability) are
% based on the finite model property.  
The proof is via a novel
unraveling operation tailored for the logic \SQ and canonical
decompositions.
\begin{theorem} \label{thm:sq-canmod}
  Let $\Kmc=(\Tmc, \Amc)$ be an \SQ KB and $\vp$ a PRPQ with
  $\Kmc\not\models\vp$. There is a model $\Jmc$ of $\Kmc$ and a
  canonical tree decomposition $(T,\Imf ,\bago)$ of $\Jmc$ with
  (i)~$\Jmc\not\models\vp$, (ii) $\Imf(\varepsilon) \models \Amc$, and
  (iii) width and outdegree of $(T,\Imf)$ are bounded by
  $O(|\Amc|\cdot 2^{p(|\Tmc|)})$, for some polynomial $p$.
% 
%   \begin{enumerate}[label=(\arabic*),leftmargin=*,align=left]
% 
%     \item\label{it:thm1} $\Imf(\varepsilon) \models
%       \Amc$;
% 
%     \item\label{it:thm3} $\Imc\not\models\vp$;
% 
%     \item\label{it:thm4} the width and outdegree of $(T,\Imf)$ are bounded by
%       $O(|\Amc|\cdot 2^{p(|\Tmc|)})$, for some polynomial $p$.
%   
%   \end{enumerate} 
\end{theorem}
%
% \begin{theorem} \label{thm:sq-canmod}
%   %
%   Let $\Kmc=(\Tmc, \Amc)$ be an \SQ KB. There is a polynomial $p$,
%   such that for every $\Imc\models \Kmc$, there is an interpretation
%   $\Jmc$ and a canonical tree decomposition $(T,\Imf ,\bago)$ of
%   $\Jmc$ such that:
% 
%   \begin{enumerate}[label=(\arabic*),leftmargin=*,align=left]
% 
%     \item\label{it:thm1} $\Imf(\varepsilon) \models
%       \Amc$;
% 
%     \item\label{it:thm2} $\Jmc\models\Kmc$;
% 
%     \item\label{it:thm3} there is a homomorphism from $\Jmc$ to \Imc;
% 
%     \item\label{it:thm4} the width and outdegree of $(T,\Imf)$ are bounded by
%       $O(|\Amc|\cdot 2^{p(|\Tmc|)})$.
%   
%   \end{enumerate} 
% \end{theorem}
%
Before outlining  the proof of Theorem~\ref{thm:sq-canmod}, we
introduce some additional notation.  The \emph{width of an interpretation~\Imc} is the
minimum $k$ such that $|Q_{\Imc,r}(d)|\leq k$ for all
$d\in\Delta^\Imc$, $r\in\mn{N}_\mn{R}^t$.  Moreover, for a transitive
role $r$, we say that $e$ is a \emph{direct $r$-successor of $d$} if
$(d,e)\in r^\Imc$ but $e\notin Q_{\Imc,r}(d)$, and for each $f$ with
$(d,f),(f,e)\in r^\Imc$, we have $f\in Q_{\Imc,r}(d)$ or $f\in
Q_{\Imc,r}(e)$; if $r$ is non-transitive, then $e$ is a \emph{direct
$r$-successor of $d$} if $(d,e)\in r^\Imc$.  The \emph{breadth of
\Imc} is the maximum $k$ such that there are $d,d_1,\ldots,d_k$ and a
role name $r$, all $d_i$ are direct $r$-successors of $d$, and
\begin{itemize}

  \item[--] if $r$ is non-transitive, then $d_i\neq d_j$ for all
    $i\neq j$; 

  \item[--] if $r$ is transitive, then $Q_{\Imc,r}(d_i)\neq
    Q_{\Imc,r}(d_j)$, for $i\neq j$.

\end{itemize}

Let now be $\Imc\models \Kmc$ and $\Imc\not\models \varphi$. As PRPQs
are preserved under homomorphisms, the following lemma  implies that we
can assume without loss of generality that $\Imc$ is of bounded width
and breadth. The proof of this lemma adapts a result in~\cite{KazakovP09}.
\begin{lemma}~\label{lem:outdegree-width} 
  For each $\Imc\models \Kmc$, there is a sub-interpretation $\Imc'$
  of \Imc with $\Imc'\models\Kmc$ and width and breadth of $\Imc'$ are
  bounded by $O(|\Amc|+2^{p({|\Tmc|})})$.
\end{lemma}
Let $\mn{cl}(\Tmc)$ be the set of all subconcepts occurring in \Tmc,
closed under single negation.  For each transitive role $r$, define a
binary relation $\rightsquigarrow_{\Imc,r}$ on $\Delta^\Imc$, by
taking $d\rightsquigarrow_{\Imc,r}e$ if there is some $\qnrleq n r 
C\in \Tmc$ such that $d\in \qnrleq n r  C^\Imc$, $e\in
C^\Imc$, and $(d,e)\in r^\Imc$. Based on the transitive, reflexive
closure $\rightsquigarrow_{\Imc,r}^*$ of $\rightsquigarrow_{\Imc,r}$,
we define, for every $d\in \Delta^\Imc$, the set
$\mn{Wit}_{\Imc,r}(d)$ of {\em $r$-witnesses for $d$} by:
$$\mn{Wit}_{\Imc,r}(d) = \textstyle\bigcup_{e\mid
  d\rightsquigarrow_{\Imc,r}^* e} Q_{\Imc,r}(e).$$
Intuitively, $\mn{Wit}_{\Imc,r}(d)$ contains all $r$-witnesses of
at-most restrictions of some element $d$, and due to using
$\rightsquigarrow_{\Imc,r}^*$, also all witnesses of at-most
restrictions of those witnesses and so on. For the stated bounds, it
is important that  the size of $\mn{Wit}_{\Imc,r}(d)$ is bounded as
follows:
\begin{lemma} \label{lem:wit}
  For every $d\in \Delta^\Imc$ and transitive $r$, we have
  $|\mn{Wit}_{\Imc,r}(d)|\leq |\Amc|\cdot 2^{q(|\Tmc|)}$, for some
  polynomial $q$.
\end{lemma}
 
We describe now the construction of the
interpretation $\Jmc$ and its tree decomposition via a possibly
infinite unraveling process. Elements of $\Delta^\Jmc$ will be either
of the form $a$ with $a\in\mn{ind}(\Amc)$ or of the form $d_x$ with
$d\in \Delta^\Imc$ and some index $x$. We usually use $\delta$ to
refer to domain elements in \Jmc (in either form), and define a
function $\tau:\Delta^\Jmc\to\Delta^\Imc$ by setting
$\tau(\delta)=\delta$, for all $\delta\in \mn{ind}(\Amc)$, and
$\tau(\delta)=d$, for all $\delta$ of the form $d_x$.

To start the construction of \Jmc and $(T,\Imf,\bago)$, initialize
the domain $\Delta^{\Jmc}$ with
$\mn{ind}(\Amc)\cup\bigcup_{r\in\mn{N}_\mn{R}^t}\Delta^r$,
where 
% $\Jmc_0=\Imc_0|_{\mn{ind}(\Amc)}$, and add, for every transitive role $r$,
the sets $\Delta^r$ are defined as
\begin{align*}
\Delta^r &= \{d_r\mid d\in
  \textstyle\bigcup_{a\in\mn{ind}(\Amc)}\mn{Wit}_{\Imc,r}(a)\setminus\mn{ind}(\Amc)\}.
%   \quad
%   \text{ and }\\
% \Delta_r' &= \Delta_r\cup\mn{ind}(\Amc).
\end{align*}
Concept and role names are interpreted in a way such that
$\Jmc|_{\mn{ind}(\Amc)}=\Imc|_{\mn{ind}(\Amc)}$, and
for all $r\in\mn{N}_\mn{R}^t$ and all $\delta,\delta'\in
\mn{ind}(\Amc)\cup\Delta^r$, we
have
\begin{align}
  \begin{split}
    \delta\in A^{\Jmc} &\Leftrightarrow \tau(\delta)\in A^{\Imc}\text{, for
    all $A\in\mn{N_C}$,  \quad and  \notag } \\
    (\delta,\delta')\in r^{\Jmc}&
    \Leftrightarrow (\tau(\delta),\tau(\delta'))\in
    r^{\Imc}.
  \end{split}\tag{$\dagger$}\label{eq:ij}
\end{align}
Now, initialize $(T,\Imf,\bago)$ with $T=\{\varepsilon\}$,
$\Delta_\varepsilon = \Delta^{\Jmc}$, and $\bago(\varepsilon)=\bot$.
This first step ensures that all witnesses of ABox
individuals appear in the root. 

\smallskip In the inductive step, we extend \Jmc and $(T,\Imf,\bago)$
by applying the following rules exhaustively in a fair way.
\begin{enumerate}[label=\textbf{R$_\arabic*$},leftmargin=*,align=left]

  \item\label{it:r1} Let $r$ be non-transitive, $w\in T$, $\delta \in F(w)$,
    and $d$ a direct $r$-successor of $\tau(\delta)$ in \Imc with
    $\{\delta,d\}\not\subseteq\mn{ind}(\Amc)$.  Then, add a fresh
    successor $v$ of $w$ to $T$, add the fresh element $d_v$ to
    $\Delta^{\Jmc}$, extend $\Jmc$ by adding $(\delta,d_v)\in
    r^{\Jmc}$ and $d_v\in A^{\Jmc}$ iff $d\in A^\Imc$, for all $A\in \mn{N_C}$, and
    set $\Delta_v= \{\delta,d_v\}$ and $\bago(v)=r$.

  \smallskip \item\label{it:r2} Let $r$ be transitive, $w\in T$, and
   $\delta_0\in F(w)$ such that:

   \begin{enumerate}[label=(\alph*)]

     \item $w=\varepsilon$ and $\delta_0\in\Delta^s$, for some
       transitive $s\neq
       r$, or
%        ($\Delta_s$ defined in the initialization phase), or

     \item $w\neq \varepsilon$ and $\bago(w)\neq r$.

   \end{enumerate}
   Then add a fresh successor $v$ of $w$ to $T$, and define
    %
   %\begin{align*}
      %
%       \vspace*{-1mm}
      $$\Delta = \{e_v\mid e\in
	\mn{Wit}_{\Imc,r}(\tau(\delta_0))\setminus\{\tau(\delta_0)\}\}.$$
% 	\text{and}\ \Delta'=\Delta\cup\{\delta\}.\\[-1mm]$$
      %
    %\end{align*}
    %
    Extend the domain of \Jmc with $\Delta$ and the
    interpretation of concept and role names such that~\eqref{eq:ij}
    is satisfied for all $\delta,\delta'\in\Delta\cup\{\delta_0\}$. Finally, set
    $\Delta_v=\Delta\cup\{\delta_0\}$ and $\bago(v)=r$.

  \smallskip \item\label{it:r3} Let $r$ be transitive, $w\in T$,
    $\abf\subseteq F_r(w)$ an $r$-cluster in $\bagz(w)$ such
    that:
    \begin{enumerate}[label=(\alph*)]

      \item $w=\varepsilon$ and $\abf\subseteq
	\Delta^r\cup\mn{ind}(\Amc)$, or

      \item $w\neq \varepsilon$ and $\bago(w)=r$.

    \end{enumerate}
    If there is a direct $r$-successor $e$ of $\tau(\delta)$ in \Imc
    for some $\delta\in\abf$ such that $(\delta,\delta')\notin r^\Jmc$
    for any $\delta'$ with $\tau(\delta')=e$, then add a fresh
    successor $v$ of $w$ to $T$, and define 
    \begin{align*}
      \Delta &= \{f_v\mid f\in \mn{Wit}_{\Imc,r}(e)\setminus
      \mn{Wit}_{\Imc,r}(\tau(\delta))\} \quad\text{and}\\
      \Delta_v &=\Delta\cup \abf \cup \{ \delta''\mid
	r(\delta',\delta'')\in \bagz(w)\text{ for some
	}\delta'\in\abf\}.
    \end{align*}
   Then extend the domain of \Jmc with $\Delta$ and the interpretation
   of concept names such that~\eqref{eq:ij} is satisfied for all 
   $\delta\in\abf\cup\Delta$ and $\delta'\in\Delta_v$. Finally, set $\bago(v)=r$.

\end{enumerate}
To finish the construction, let \Jmc be the interpretation obtained in
the limit, and set $\Imf(w)=\Jmc|_{\Delta_{w}}$, for all $w\in T$. It
is verified in the appendix that $(T,\bagz,\bago)$ and \Jmc satisfy
the conditions from Theorem~\ref{thm:sq-canmod}. Notably, $\tau$ is a
homomorphism from $\Jmc$ to \Imc, thus $\Jmc\not\models\vp$, due to
preservation under homomorphisms.

Rules~\ref{it:r1}--\ref{it:r3} are, respectively, in one-to-one correspondence
with Conditions \ref{it:can2}--\ref{it:can4} in Definition~\ref{def:candec}. In
particular, \ref{it:r1} implements the well-known unraveling procedure for
non-transitive roles. \ref{it:r2} is used to change the `role
component' for transitive roles by creating a fresh node whose
interpretation contains all
witnesses of the chosen element $\delta$. Finally, \ref{it:r3} describes how
to unravel direct $r$-successors in case of transitive roles~$r$.
In the definition of $\Delta$ it is taken care that
witnesses which are `inherited' from predecessors are not introduced
again, in order to preserve at-most restrictions.

We finish the section
with an illustrating example. 
% 
% Theorem~\ref{thm:sq-canmod} yields a tree-like model property for
% \SQ-knowledge bases, which is interesting on its own since existing
% decidability results (for
% satisfiability)~\cite{KazakovSZ07,GuIbJu-AAAI17} are based on the
  % finite model property.

\begin{example}\label{exam:unra}
 Let $\Kmc$ be the following KB, where $r \in \mn{N}_\mn{R}^{t}$:   $$(\{A_1 \sqsubseteq \qnrleq 1  r  B, A_2 \sqsubseteq \qnrleq
 1  r  C \},  \{A_1(a)\}).$$ Figure
 \ref{fig:candec} shows  a model $\Imc$ of \Kmc and  a canonical
 decomposition \Tmf of its unraveling (transitivity connections are
 omitted).  In the initialization phase, the interpretation $\Imf(\varepsilon)$  is constructed
 starting from individual $a$. Since $a \rightsquigarrow_{\Imc,r} e$ and $e
 \rightsquigarrow_{\Imc,r} f$, we have 
 $\mn{Wit}_{\Imc,r}(a)=\{e,f\}$, thus $e_r$ and $f_r$ are added in
 this phase. The interpretations $\Imf(v_i)$ are introduced using \ref{it:r3}: In all cases 
 $\Delta_\varepsilon$ is the cluster $\abf$ and $\delta=a$; and, e.g., $\Delta =\{ c_{v_1}\}$ for $\Imf(v_1)$. 
 
%  Defining $\mn{Wit}_{\Imc,r}(a)$ in this way  is crucial
%  as otherwise elements $f_{v_1}$ and $f_{v_2}$ are introduced (using
%  the inductive rules)   in the unraveling and  it would then not be a
%  model of  $\Tmc$.

\begin{figure}[t!]
\centering
 \includegraphics[trim = 0.60cm 0 0.6cm 0, scale=0.50]{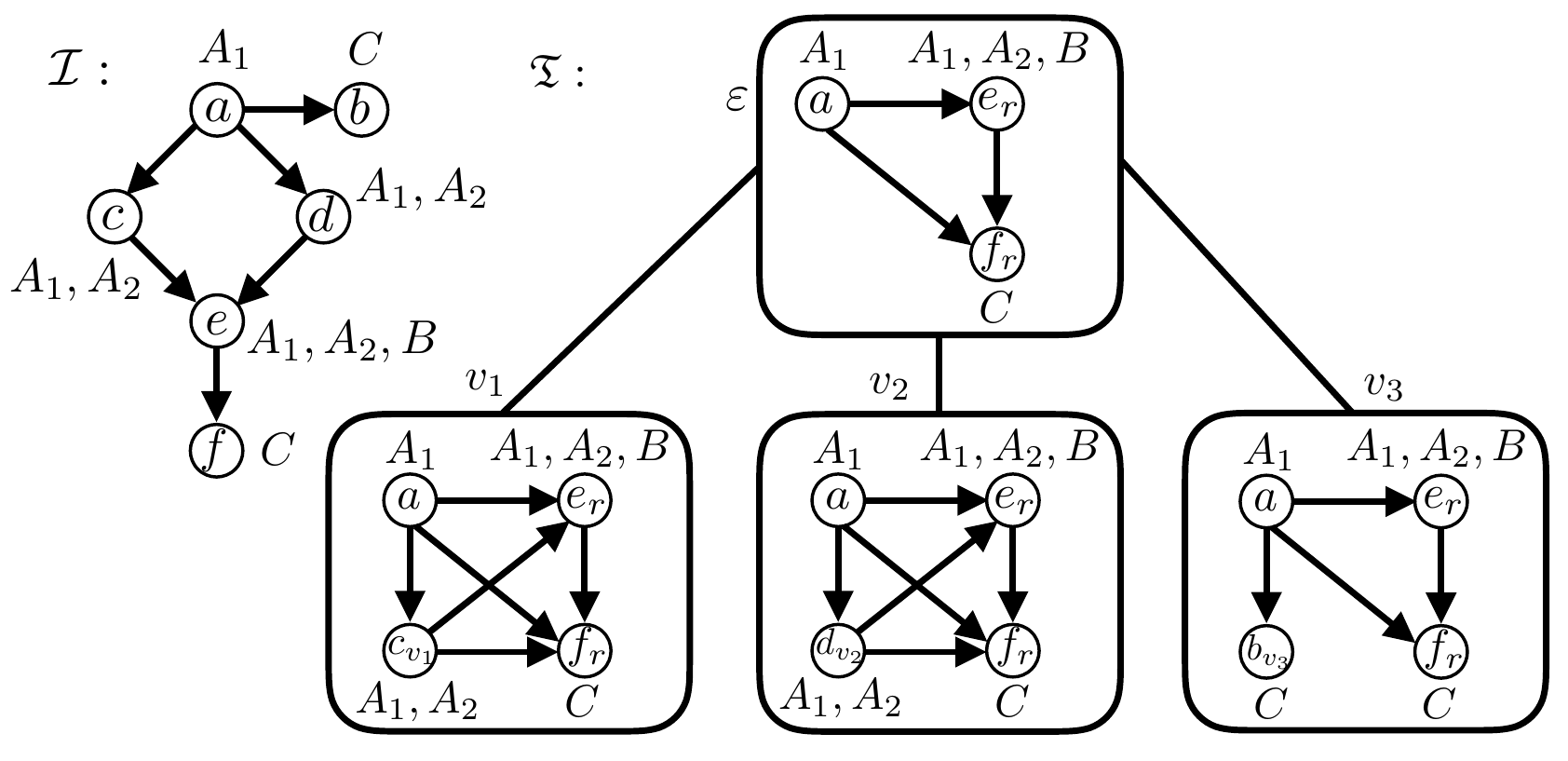}%[trim= 5.5cm 5cm 1cm 7cm, scale=0.43]{figureNEW}
 \vspace{-3mm}\caption{Example~\ref{exam:unra}}\label{fig:candec}
 \end{figure}
\end{example}

\section{Automata-Based Query Entailment}\label{sec:complexityqa}

In this section, we devise an automata-based decision procedure for
query entailment  in \SQ. We start with the necessary background about
the used automata model.

\smallskip\textbf{Alternating Tree Automata.} A tree is {\em $k$-ary}
if each node has {\em exactly} $k$ successors. For brevity, we set
$[k]= \{-1,0, \ldots, k\}$. Let $\Sigma$ be a finite alphabet. A
\emph{$\Sigma$-labeled tree} is a pair  $(T,\tau)$ with  $T$ a tree
and $\tau: T\rightarrow \Sigma$ assigns a letter from $\Sigma$ to each
node.
% For a set $X$, let $\mathcal B^+(X)$ be the set of positive Boolean
% formulas built from elements in $X$ using $\land$, $\lor$,
% $\mn{true}$ and $\mn{false}$. Let $Y \subseteq X$. We say that $Y$
% \emph{satisfies} a formula $\theta \in \mathcal{B}^+(X)$ if
% assigning $ \mn{true}$ to the members of $Y$ and assigning
% $\mn{false}$ to the members of $X \setminus Y$ makes $\theta$
% true.\nb{J: definition of trees as subset of $\mathbb{N}^*$ not here
% anymore, $\Sigma$-labeled trees missing}
% 
A \emph{two-way alternating  tree automaton (2ATA)} over
$\Sigma$-labeled $k$-ary trees is a tuple $\Amf=(Q, \Sigma, q_0,
\delta, F)$ where $Q$ is a finite set of \emph{states}, $q_0 \in Q$ is
an \emph{initial state}, $\delta$ is the \emph{transition function},
and $F$ is the \emph{(parity) acceptance condition}~\cite{Vardi98}.
The transition function maps a state $q$ and an input letter $a\in
\Sigma$ to a positive Boolean formula over the constants $\mn{true}$
and $\mn{false}$, and variables from $[k]\times Q$. The semantics is
given in terms of \emph{runs}, see  appendix. As usual,
$L(\Amf)$ denotes the set of trees accepted by \Amf.  Emptiness of
$L(\Amf)$ can be checked in exponential time in the number of states
of \Amf~\cite{Vardi98}.

\smallskip\textbf{General Picture.} The leading thought is as follows.
If $\Kmc\not\models \vp$, then, by Theorem~\ref{thm:sq-canmod}, there
is a model $\Jmc$ of $\Kmc$ and a canonical tree decomposition thereof
with small width and outdegree such that $\Jmc\not\models \vp$. The
idea is to design 2ATAs $\Amf_{\mn{can}}$, $\Amf_\Kmc$, and $\Amf_\vp$
which accept canonical tree decompositions, (tree-like) models of the KB $\Kmc$,
and (tree-like) models of the query $\vp$, respectively. Query answering is then
reduced to the question whether some tree is accepted by
$\Amf_\mn{can}$ and $\Amf_\Kmc$, but not by $\Amf_\vp$.  As we shall
see, these automata have size exponential in \Kmc and can be
constructed in double exponential time.  Since 2ATAs can be
complemented and intersected in polynomial time, the automaton
$\Amf_\mn{can}\wedge \Amf_\Kmc\wedge \neg \Amf_\vp$ is of exponential
size, and can be constructed in double exponential time. Checking it
for non-emptiness can thus be done in double exponential time. A matching
lower bound is inherited from positive existential query answering
in~\ALC~\cite{CalvaneseEO14}. We thus obtain our main result.
\begin{theorem}\label{thm:main}
  PRPQ entailment over \SQ-knowledge bases is \TwoExpTime-complete. 
\end{theorem}

\smallskip\textbf{Encoding Tree Decompositions.}
As the underlying interpretation might be infinite, 
2ATAs cannot directly work over tree decompositions.
Thus, for the desired approach to work, it is crucial to
\emph{encode} tree decompositions using a finite alphabet. 
% the above representation of tree decompositions needs to be refined
% in a crucial point:  to obtain a finite alphabet, we have to define
% an appropriate encoding of tree decompositions.  \nb{V: new}Indeed,
% interpretations  at the nodes of  a tree decomposition can be
% infinite.  
To this aim, we use an approach similar to~\cite{GradelW99}.
 
Throughout this section, fix a knowledge base $\Kmc$ and
let $K$ and $k$ be the bounds on width and outdegree, respectively, obtained in Theorem~\ref{thm:sq-canmod}. Then, fix a finite set
$\Delta$ having 
$2K$ elements with
$\mn{ind}(\Amc)\subseteq \Delta$, 
and define $\Sigma=\{\bullet\}\cup\Sigma'$, where
$\Sigma'$ is the set of all pairs $(\Imc,x)$ such that $\Imc$ is an
interpretation where only symbols from \Kmc are interpreted non-empty,
$\Delta^\Imc\subseteq \Delta$, $|\Delta^\Imc|\leq K$, and $x$ is
either a role name from \Kmc or $\bot$. The symbol $\bullet\in\Sigma$
is used to encode non-existing branches  (tree decompositions are not necessarily uniformly branching).

Let $(T,\tau)$ be a $\Sigma$-labeled tree with $\Sigma$ as above. For convenience, we use
$\Imc_w$ and $r_w$ to refer to the single components of $\tau$ in a node $w$ with
$\tau(w)\neq \bullet$, that
is, $\tau(w)=(\Imc_w,r_w)$. Given an element $d\in \Delta$, we say that
$v,w\in T$ are \emph{$d$-connected} iff $d\in\Delta^{\Imc_u}$ for all
$u$ on the unique shortest path from $v$ to~$w$. In case $d\in
\Delta^{\Imc_w}$, we use $[w]_d$ to denote the set of all $v$ which
are $d$-connected to~$w$. We call $(T,\tau)$ \emph{consistent} if
$\varepsilon$ is the only node with $r_{\varepsilon}=\bot$ and
$(\Imc_w)|_{D}=(\Imc_v)|_D$ for all neighbors $v,w\in T$ and
$D=\Delta^{\Imc_w}\cap\Delta^{\Imc_v}$.
A consistent $\Sigma$-labeled tree $(T,\tau)$ \emph{represents}
a triple $(T,\bagz,\bago)$ of width at most $K$
% such that $(T,\bagz)$ a the tree decomposition of some \Imc  
as follows. The
domain underlying $(T,\bagz,\bago)$ is the set of all elements  $[w]_d$ with $w\in T$ and $d\in\Delta^{\Imc_w}$, and for every
$w\in T$, the interpretation $\bagz(w)$ is defined as:
\begin{align*}
  \Delta_w = \{[w]_d\mid d\in \Delta^{\Imc_w}\}, \quad
  A^{\Imf(w)} = \{[w]_d\mid d\in A^{\Imc_w}\},\\ 
  r^{\Imf(w)} = \{ ([w]_{d},[w]_{e})\mid (d,e)\in r^{\Imc_w}\},
  \hspace{1cm}
\end{align*} 
for all concept names $A$ and role names $r$ occurring in \Kmc; and
$\bago(w)$ is just $r_w$. 
We denote with $\Imc_{(T,\tau)}$ the interpretation
$\bigcup_{w\in T}\Imc_w$; clearly, $(T,\bagz,\bago)$ is a tree
decomposition of $\Imc_{(T,\tau)}$.  As a convention,
we use $[\varepsilon]_a$ to represent each ABox
individual  $a\in\mn{ind}(\Amc)$ in the encoding.
Based on the size $2K$ of $\Delta$, it is not hard to verify that,
conversely, for every width $K$ tree decomposition of some \Imc,
there is a consistent $(T,\tau)$ such that $\Imc_{(T,\tau)}$ is isomorphic to \Imc.

% We consider only $k$-ary trees where $k$ is the bound on the outdegree
% given in Theorem~\ref{thm:sq-canmod}. 

It is easy to
devise a 2ATA $\Amf_\mn{can}$ which accepts an input $(T,\tau)$
iff it is consistent and  the represented tree decomposition
$(T,\bagz,\bago)$ is canonical. We thus concentrate on the most challenging automata $\Amf_\Kmc$ and
$\Amf_\vp$.

\subsection{	Knowledge Base Automaton $\Amf_\Kmc$}\label{sec:KBA} The
automaton $\Amf_\Kmc$ is the intersection of two automata $\Amf_\Amc$
and $\Amf_\Tmc$ verifying that the input satisfies the ABox and the
TBox, respectively.  Note that, by Point~($ii$) of
Theorem~\ref{thm:sq-canmod}, we can assume that the ABox is satisfied
in the root; thus, an automaton $\Amf_\Amc$ checking whether
$\Imc_{(T,\tau)}\models\Amc$ just has to check the
label~$\tau(\varepsilon)$, see the appendix.

\smallskip For the design of the automaton $\Amf_\Tmc$, assume
w.l.o.g.\ that $\Tmc$ is of the form $\{\top\sqsubseteq C_\Tmc\}$ and
$C_\Tmc$ is in negation normal form. We present the main ideas of the
construction of $\Amf_\Tmc$, see the appendix for further details.  In
its `outer loop', the automaton visits every domain element $d$ in
state $C_\Tmc(d)$. This is realized using the initial state $q_0$, and
states of the form $D(d)$, $D$ a sub-concept of $C_\Tmc$ and $d\in
\Delta$ via the following transitions for every $(\Imc,x)\in \Sigma$:
\begin{align*}
  \delta(q_0,(\Imc,x)) & = \textstyle\bigwedge_{1\leq i\leq K}(i,q_0)
  \wedge \bigwedge_{d\in\Delta^\Imc} (0,C_\Tmc(d)) \\
  \delta(q_0,\bullet) & = \mn{true} 
\end{align*}
If $\Amf_\Tmc$ visits $w$ in a state $D(d)$ this presents the
obligation to verify that, in the represented model, $[w]_d$ satisfies
$D$. The Boolean operations are dealt with using the following
transitions, for every $(\Imc,x)\in \Sigma$:
\begin{align*}
  \delta( A(d),(\Imc,x)) & = \text{if $d\in A^\Imc$, then \mn{true}\
  else \mn{false}} \\
  \delta(\neg A(d),(\Imc,x)) & = \text{if $d\notin A^\Imc$, then
    \mn{true}\ else \mn{false}} \\
  \delta( (C_1\sqcup C_2)(d), (\Imc,x)) & = (0,C_1(d))\vee (0,C_2(d))
  \\
  \delta( (C_1\sqcap C_2)(d), (\Imc,x)) & = (0,C_1(d))\wedge
  (0,C_2(d))
\end{align*}
For states of the form $\qnrsim{n}{r}{D}(d)$ we have to be more careful.
The naive approach for counting the number of $r$-successors of $d$
satisfying $D$ would be to count the number of $r$-successors
satisfying $D$ in the interpretation associated to the current node,
and then move to all other nodes where $d$ appears.  Since
interpretations associated to neighboring nodes might overlap, to
avoid double counting, we have to store (in the states) all elements
that have already been counted in the current node before changing the
node. However, since the domain in each node has size exponential in
$|\Tmc|$, we need doubly exponentially many states for this task.
Since this naive approach does not result in optimal complexity, we pursue an
alternative approach, based on canonicity, leading to only
exponentially many states.

Our approach is based on characterizing how $r$-successors
of an element can be uniquely identified in canonical tree
decompositions.  Assume some $(T,\tau)\in L(\Amf_\mn{can})$ and let $r$
be a role name. In what follows, we assume that the notions of `fresh'
and `$r$-fresh' are lifted to the encoding in the straightforward way. 
An \emph{$r$-path from $[w]_d$ to $[v]_e$ in $(T,\tau)$} is a sequence
$d_0,w_0,d_1,\ldots,w_{n-1},d_n$ such that $d=d_0$, $e=d_n$, $w_0\in
[w]_d$, $w_{n-1}\in [v]_e$, and $(d_i,d_{i+1})\in r^{\Imc_{w_i}}$, for
all $0\leq i<n$. It is \emph{downward} if, for all $0< i<n$, $w_{i}$
is a successor of $w_{i-1}$ and $d_{i}$ is contained in an $r$-root
cluster of $w_{i}$. We then have:
% and \emph{canonical} if 
% %
% \begin{enumerate}[label={\textbf
%     P$_\arabic*$},leftmargin=*,align=left]
% 
%   \item\label{it:p1} it is downward and $d_1\in\roots(w_0)$, or 
% 
%   \item\label{it:p2} $d_0\in\roots(w_0)$, $d_1\notin\roots(w_0)$, and
%     if $n>1$, then $d_1,\ldots,d_n$ is downward and $w_1{\cdot}{-}1 \in [w]_{d_1}$
%     is an ancestor of $w_0$ such that $d_1 \in F_r(w_1{\cdot}{-}1)$. 
%     % and satisfies $d_1 \in \roots(w_1\cdot -1)$.  
% 
% \end{enumerate}
% % 
% Two $r$-paths $d_0,\ldots,d_n$ and $e_0, \ldots,e_m$  from $d$ to $e$
% are \emph{equivalent} if $n= m$ and, for every $0\leq i <n$, we have
% $w_i = w'_i$ and $e_{i}\in Q_{\Imc_{w_i},r}(d_i)$. 
%
\begin{lemma} \label{lem:can-paths-encoding}
  %
%   Let $(T,\tau)\in L(\Amf_{\mn{can}})$. Then,
  For $(T,\tau)\in L(\Amf_{\mn{can}})$,
  we have $([w]_d,[v]_e)\in r^{\Imc_{(T,\tau)}}$ iff one of the
  following is true:
  \begin{itemize}

    \item[--] $r$ is non-transitive and $(d,e)\in
      r^{\Imc_\varepsilon}$ or $(d,e)\in r^{\Imc_v}$,
      $d$ is fresh in $w$, and $v$ is a successor of $w$, or

    \item[--] $r$ is transitive, and there is an $r$-path
      $d_0,w_0,\ldots,d_n$ from $[w]_d$ to $[v]_e$ such that one of
      the following holds: 
      \begin{enumerate}[label=\textbf{\Alph*},leftmargin=*,align=left]

% 	\item\label{it:p1} $d_0\in\roots(w_0)$ and $n=1$, or

	\item\label{it:p1} $d_0\in\roots(w_0)\cup\roots(w_0\cdot
	  -1)$, $d_1\in\roots(w_0)$, and  
	  $d_0,\ldots,d_n$ is downward, or

	\item\label{it:p2} $d_0\in\roots(w_0)$,
	  $d_1\notin\roots(w_0)$, and if $n>1$, then $d_1,\ldots,d_n$
	  is downward and $w_1{\cdot}{-}1 \in [w]_{d_1}$ is an
	  ancestor of $w_0$ such that $d_1 \in F_r(w_1{\cdot}{-}1)$. 

      \end{enumerate}

    \end{itemize}
  
\end{lemma}
This lemma suggests the following approach for verifying the
obligation $\qnrsim{n}{r}{D}(d)$ at some node $w$. If $r$ is
non-transitive, `navigate' with the automaton to the (unique!) $w^*$
such that $d\in F(w^*)$ and count the $r$-successors of $d$ in the
successors $v$ of $w^*$, or in $\varepsilon$. If $r$ is transitive, navigate with the
automaton to the unique $w^*$ such that $d\in F_r(w^*)$ and change to
a state $q^*_{\qnrsim{n}{r}{D},d}$, starting from which $\Amf_{\Tmc}$
systematically scans the $r$-successors according to~\ref{it:p1}
and~\ref{it:p2}. We concentrate on verifying at-least restrictions,
at-most restrictions are completely complementary.

Assume $\tau(w^*)=(\Imc,x)$, and let $\abf_1,\ldots,\abf_\ell$ be all
$r$-clusters in $\Imc$ reachable from $d$ (including $Q_{\Imc,r}(d)$),
and let $a_1,\ldots,a_\ell$ be representatives of these clusters.
Moreover, let $N$ be the set of all tuples $\nbf=(n_1,\ldots,n_\ell)$
such that $\sum_{i}n_i=n$.  Then, the transition
$\delta(q^*_{\qnrgeq{n}{r}{D},d},(\Imc,x))$ is defined as
\begin{align*}
  \bigvee_{\nbf\in N}\bigvee_{X\subseteq [1,\ell]}\bigwedge_{i\in X}
  (0,q^{\ref{it:p1}}_{\qnrgeq{n_i}{\!r}{D},a_i})\wedge\!\!\!\bigwedge_{i\in
  [1,\ell]\setminus X} (0,q^{\ref{it:p2}}_{\qnrgeq{n_i}{\!r}{D},a_i}).
\end{align*}
Thus, $\Amf_{\Tmc}$ guesses a distribution of $n$ to the reachable
clusters. Moreover, it guesses from which clusters it starts paths of
the shape~\ref{it:p1} and~\ref{it:p2}. For both guesses, it verifies
that the chosen $a_i$ is $r$-fresh (for~\ref{it:p1}) or not
(for~\ref{it:p2}), and continues in states $q^\downarrow_{\qnrgeq
n r D}$ and $q^\uparrow_{\qnrgeq n r  D}$, respectively. This is done
using the following transitions:
\begin{align*}
  \delta(q^{\ref{it:p1}}_{\qnrgeq{n}{r}{D}, d}, (\Imc,x)) & =
  (0,F_{r,d})\wedge (0, q_{\qnrgeq{n}{r}{D},d}^\downarrow) \\
  \delta(q^{\ref{it:p2}}_{\qnrgeq{n}{r}{D}, d}, (\Imc,x)) & =
  (0,\overline{F}_{r,d})\wedge (-1,q_{\qnrgeq{n}{r}{D},d}^\uparrow)\\
  \delta(F_{r,d}, (\Imc,\bot)) & = \mn{true} \\ 
  \delta(F_{r,d}, (\Imc,x)) & = \mn{false} \quad\quad\quad \text{if
    $x\notin \{r,\bot\}$} \\
  \delta(F_{r,d}, (\Imc,r)) & = (-1, F'_{r,d}) \\
  \delta(F'_{r,d}, (\Imc,x)) & = \begin{cases} \mn{true } & \text{if
      $x\notin \{r,\bot\}$ or $d\not \in \Delta^\Imc$}, \\ \mn{false}
      & \text{otherwise}, \end{cases}
\end{align*}
and complementary transitions for $\overline F_{r,d}$.  Now, in states
$q^\uparrow_{\qnrgeq{n}{r}{D},d}$, the automaton goes up until it
finds the world where $d$ is $r$-fresh (corresponding to $w_1\cdot -1$
in~\ref{it:p2}) and looks for downward paths starting from there. This
is done by taking setting $\delta(q^\uparrow_{\qnrgeq{n}{r}{D},d},(\Imc,x)) = \mn{false}$ whenever $d\notin\Delta^\Imc$, and
otherwise:
\begin{align*}
  %
%   \delta(q^\uparrow_{(\geq n\,r\,D),d},(\Imc,x)) & = \mn{false}
%   \quad\quad \text{if $d\notin \Delta^\Imc$} \\
  %
  \delta(q^\uparrow_{\qnrgeq{n}{r}{D},d},(\Imc,x)) & =
  (0,q^{\ref{it:p1}}_{\qnrgeq{n}{r}{D},d})\vee
  (0,q^{\ref{it:p2}}_{\qnrgeq{n}{r}{D},d}).
\end{align*}
It thus remains to describe transitions for states of the form
$q^\downarrow_{\qnrgeq{n}{r}{D},d}$ at some node $w$. Such  situations
represent the obligation to find $n$ $r$-successors along downward
paths from $d$. Note that the transitions before ensure that $d\in
F_r(w)$. In this case, the automaton guesses how many of the $n$
successors  it will find locally in the current cluster (using states
$p^{\mn{loc}}_{m,r,D,d}$), and how many are to
be found in successor nodes (using
$p^{\mn{succ}}_{\qnrgeq{m}{r}{D}}$). Formally, let $M$ be the set of all
tuples $\mbf=(m_0,\ldots,m_k)$ with $\sum_{i}m_i=n$, and define the
transition for $\delta(q^\downarrow_{{\qnrgeq n r  D},d},(\Imc,x))$ as:
\begin{align*}
  \bigvee_{\mbf\in M} \Big((0,p^\mn{loc}_{m_0,r,D,d}) \wedge
  \bigwedge_{i\in[1,k]} (i,p^{\mn{succ}}_{\qnrgeq{m_i}{r}{D},d})\Big)
\end{align*}
States of the form $p^\mn{loc}_{n,r,D,d}$ are used  to verify that in
$Q_{\Imc,r}(d)$ there are $n$ elements satisfying $D$:
\begin{align*}
  \delta(p^\mn{loc}_{n,r,D,d}, (\Imc,x)) & = \bigvee_{Y\subseteq
    Q_{\Imc,r}(d),|Y|=n} \bigwedge_{e\in Y} D(e).
\end{align*}
It remains to give the transitions for states
$p^{\mn{succ}}_{\qnrgeq{m}{r}{D}}$. To start, we set
$\delta(p^{\mn{succ}}_{\qnrgeq{n}{r}{D},d},\sigma)=\mn{true}$,
whenever $n=0$; $\delta(p^\mn{succ}_{\qnrgeq{n}{r}{D},d},\bullet)=\mn{false}$;
and $\delta(p^\mn{succ}_{\qnrgeq{n}{r}{D},d},(\Imc,x))=\mn{false}$ whenever
$x\neq r$ or $d$ is not in a root cluster of \Imc. For all other
cases, let $\abf_1,\ldots,\abf_\ell$ be all $r$-clusters reachable
from $d$, except $Q_{\Imc,r}(d)$, let $N$ be again the set of all
$\nbf=(n_1,\ldots,n_\ell)$ such that $\sum_{i}n_i=n$, and include the
transition
\begin{align*}
  \delta(p^\mn{succ}_{\qnrgeq{n}{r}{D},d},(\Imc,x)) & = \bigvee_{\nbf\in N}
  \bigwedge_{i\in[1,\ell]}
  (0,q^{\ref{it:p1}}_{\qnrgeq{n_i}{r}{D},a_i}).
\end{align*}
Using the parity condition, we make sure that states
$q^\downarrow_{\qnrgeq{n}{r}{D},d}$ with $n\geq 1$ are not suspended
forever, that is, eventualities are finally satisfied. 
\begin{lemma} \label{lem:kba-correct}
  For every $(T,\tau)\in L(\Amf_{\mn{can}})$, we have $(T,\tau)\in
  L(\Amf_\Tmc)$ iff $\Imc_{(T,\tau)}\models \Tmc$. It
  can be constructed in time double exponential in $|\Kmc|$, and has
  exponentially many states in $|\Kmc|$.
\end{lemma}
 
\subsection{Query Automaton $\Amf_\varphi$} 

In previous work, we have observed that  the approach for the query
automaton taken in~\cite{CalvaneseEO14} leads to a 2ATA with double
exponentially many states in $\Kmc$, and thus not to optimal
complexity~\cite{GIJ-DL17}.  We thus take an alternative approach by first
giving an intermediate characterization for when a query has a match,
and then show how to exploit this to build a 2ATA with exponentially
many states.

% Our approach relies on a  characterization for when a query has
% a match. 
Fix a P2RPQ $\varphi=\exists \xbf\, \psi(\xbf)$. Note first that since for every regular
expression $\Emc$ over some alphabet $\Gamma$, one can construct in polynomial
time an equivalent non-deterministic finite automaton (NFA)
$\Bmf=(Q_\Bmf,\Gamma,s_{0 \Bmf },\Delta_\Bmf,F_\Bmf)$~\cite{Furer80},
we generally assume an NFA-based representation, that is, 
atoms in $\vp$ take the shape $\Bmf(t,t')$, \Bmf an NFA. For
states $s,s'\in Q_\Bmf$, write $\Bmf_{s,s'}$ for the NFA that
is obtained from $\Bmf$ by taking $s$ as initial state and $\{s'\}$ as
the set of final states.  To give semantics to the automata based
representation, we define $\Imc\models \Bmf(a,b)$
iff $\Imc\models\Emc_\Bmf(a,b)$, where $\Emc_\Bmf$ is a regular
expression equivalent to \Bmf.

A \emph{conjunctive regular path query (CRPQ)} is a PRPQ which does
not use $\vee$. It is well-known that the PRPQ $\varphi$ is equivalent
to a disjunction $q_1\vee\ldots\vee q_n$ of CRPQs, where $n$ is
exponential in $|\varphi|$. Given a CRPQ $p$, we denote with $\hat p$
the equivalent CRPQ obtained from $p$ by replacing every occurrence of
$r$ or $r^-$, $r$ transitive, with $r\cdot r^*$ or $r^-\cdot (r^-)^*$,
respectively. Let $(T,\tau)$ be a consistent $\Sigma$-labeled tree. In
the appendix, we show the following characterization.  
\begin{lemma}\label{lem:witness-sequence}
A function $\pi:\xbf\cup I_\varphi\to \Delta^{\Imc_{(T,\tau)}}$ with
$\pi(a)=[\varepsilon]_a$, for every $a\in I_\varphi$, is a match for
$\varphi$ in $\Imc_{(T,\tau)}$ iff there is a $q_i$ such that for
every $\Bmf(t,t')$ in $\hat q_i$, there is a sequence
\begin{align*}
  (d_0,s_0),w_1,(d_1,s_1),w_2,\ldots,
  w_n,(d_n,s_n),
\end{align*}
where $(d_i,s_i)\in \Delta \times Q_\Bmf$ and $w_i\in T$ and such
that:
\begin{enumerate}[label=\textit{(\alph*)},leftmargin=*,align=left]

  \item $s_0=s_{0\Bmf}$, $s_n\in F_\Bmf$,

  \item $\pi(t)=[w_1]_{d_0}$, $\pi(t')=[w_n]_{d_n}$, and

 \item for every $i\in[1,n]$, we have $d_{i-1},d_{i}\in
   \Delta^{\Imc_{w_i}}$, $w_{i}\in [w_{i-1}]_{d_{i-1}}$ if $i>1$,
   and $\Imc_{w_i}\models \Bmf_{s_{i-1},s_i}(d_{i-1},d_i)$.

\end{enumerate}
\end{lemma}
We will refer to such  sequences as \emph{witness sequences}.
The lemma suggests the following approach. In order to check whether
$\varphi$ has a match in $\Imc_{(T,\tau)}$, the automaton
guesses a $q_i$ and tries to find the
witness sequences characterizing a match. For
this purpose, $\Amf_\vp$ uses as states triples $\langle
p,V_l,V_r\rangle$ such that $p\subseteq \hat q_i$, $I_p=\emptyset$, and:
\begin{itemize}

  \item[--] $V_l$ and $V_r$ are sets of expressions of the form
    $(d,s)\to_\Bmf x$ and $x\to_\Bmf (d,s)$, respectively, where
    $\Bmf$ is the automaton of some atom $\Bmf(t,t')$ in $\hat q_i$, $s\in
    Q_{\Bmf}$, $d\in\Delta$, $x\in \mn{var}(p)$.

\end{itemize}
Intuitively, when the automaton visits a node $w$ in state $\langle
p,V_l,V_r\rangle$, this represents the obligation that 
each atom $\Bmf(x,y)$ in $p$ still has to be processed in the sense
that all variables occuring in $p$ will be instantiated in the subtree rooted
at $w$, and 

\begin{itemize}

  \item[--] for each $(d,s)\to_\Bmf x\in V_l$, $\Amf_\varphi$ tries to
    find a suffix of the witness sequence for $\Bmf(t,t')$ starting with
    $(d,s)$,

  \item[--] for each $x\to_\Bmf (d,s)\in V_r$, $\Amf_\varphi$ tries to
    find a prefix of the witness sequence for $\Bmf(t,t')$ ending with
    $(d,s)$.

\end{itemize}
We describe verbally how the automaton $\Amf_\vp$ acts when visiting a
node $w$ in state $\langle p,V_l,V_r\rangle$; the complete transition
function is given in the appendix. First, $\Amf_\vp$
non-deterministically chooses a partition $S_0,\ldots,S_k$ (with $S_i$
possibly empty, for all $i$) of $\mn{var}(p)$ and values
$d_x\in\Delta^{\Imc_w}$ for all $x\in S_0$. Intuitively, $S_0$
contains the variables that are to be instantiated in $w$, and $S_i$
contains the variables that are to be instantiated in the subtree
rooted at $w\cdot i$. Based on the taken choice, $\Amf_\vp$ determines
states $\langle p^i,V_l^i,V_r^i\rangle$ which are then sent to the
respective successors $i\in [1,k]$ of $w$. Using the parity condition,
we enforce that every variable is instantiated after finitely many of
such steps.

We demonstrate on several examples how to compute the states
$\langle p^i,V_l^i,V_r^i\rangle$ from $S_0,\ldots,S_k$ and $d_x$ for all
$x\in S_0$.
\begin{itemize}

  \item[--] Assume some $\Bmf(x,y)\in p$ with $x,y\in S_0$. In this
    case, $\Amf_\varphi$ guesses some $f\in F_\Bmf$ and verifies
    (using another set of states) that there is a witness sequence for
    $\Bmf(x,y)$ starting with $(d_x,s_{0\Bmf})$ and ending with
    $(d_y,s_{f})$.

  \item[--] Assume $\Bmf(x,y)\in p$ and $x,y\in S_i$ for some $i>0$.
    In this case, just put $\Bmf(x,y)$ into $p^i$.

  \item[--] For an atom $\Bmf(x,y)\in p$ with $x\in S_0$ and $y\in S_i$
    for $i>0$, $\Amf_\vp$ guesses an intermediate tuple $(d,s)$,
    verifies that there is a witness sequence from $(d_x,s_{0\Bmf})$
    to $(d,s)$ and adds $x\to_\Bmf(d,s)$ to $V_r^i$.  

  \item[--] For the treatment of $V_l$ ($V_r$ is similar), assume
    $(d,s)\to_\Bmf x\in V_l$. If $x\in S_0$, $\Amf_\varphi$ verifies
    that the sequence has a suffix from $(d,s)$ to $(d_x,s_f)$, for
    some $s_f\in F_\Bmf$. If $x\in S_i$, $i>0$, $\Amf_\vp$ guesses an
    intermediate pair $(d',s')$, verifies that there is an infix
    between $(d,s)$ and $(d',s')$ and includes $(d',s')\to_{\Bmf} x\in
    V_l^i$. 

\end{itemize}
We show in the appendix how to verify the existence of an infix of a
witness sequence between two pairs $(d,s)$ and $(d',s')$ as required in
the first, third and last item using only exponentially many states.
Regarding number of states, observe that there are only exponentially
many disjuncts (and thus states) $q_i$ and exponentially many states of the form
$\langle p,V_l,V_r\rangle$ as described. 

 We refer the reader to the
appendix for the complete construction and  a proof of the
following lemma.
\begin{lemma} \label{lem:query-automaton}
  There is a 2ATA $\Amf_\vp$ such that for every
  $(T,\tau)\in L(\Amf_{\mn{can}})$, we have $(T,\tau)\in
  L(\Amf_\varphi)$ iff $\Imc_{(T,\tau)}\models q$. It
  can be constructed in exponential time in $|\vp|+|\Kmc|$ and has
  exponentially in $|\vp|+|\Kmc|$ many states.
\end{lemma}

% \smallskip \noindent {\bf Overall Complexity.} Recall how we obtain
% the desired upper bound + complexity theorem

\section{Discussion and Future Work}

The obtained results  are both of practical and theoretical interest.
From the practical point of view, our complexity results and
application demands open up the possibility to include a profile based
on \SQ to OWL 2. Note that there is no increase in the computational
complexity in comparison with that of  \SQ \emph{without} counting over transitive
roles.  From the theoretical perspective, our techniques are useful for
several future lines of research. First, the unraveling 
% (and the corresponding KB automaton) 
lays the groundwork for studying
extensions of \SQ with other DL constructors. Second, the technique underlying the query automaton works
for standard tree decompositions (it does not rely on canonicity) of
bounded outdegree, even if the width is high (exponential in our
case). We thus believe that this technique is useful for query
answering in other DLs. Finally, the gained understanding of the
model-theoretic characteristics of \SQ is an important step towards
the development of more practical decision procedures.

As future work, we will  tackle the following four  interesting
problems: $(i)$ The \emph{data complexity}  of deciding entailment of
PRPQs in \SQ. The present techniques give only exponential bounds, but
we expect \textsc{coNP}-completeness. 
%
%\item
$(ii)$ The complexity of deciding   entailment of  \emph{conjunctive
queries (CQs)} in \SQ. The proposed automata-based approach yields the
same upper bound for PRPQs or CQs, but we expect it to be easier for
CQs.  $(iii)$ The complexity of deciding query entailment in
generalizations of \SQ with \emph{role composition} or  regular
expressions on roles; or with nominals and (controlled) inverses.
$(iv)$ The complexity of query entailment in \SQ over
\emph{finite} models. Indeed, \SQ lacks \emph{finite controlability},
that is, query entailment in the finite does not coincide with
unrestricted query entailment:
\begin{example}
  Consider $\Amc=\emptyset$, $\Tmc =\{\top\sqsubseteq \exists
  r.\top\}$, and $\vp=\exists x\, r(x,x)$ for some $r \in \mn{N}^t_\mn{R}$. Clearly,
  $(\Tmc,\Amc)\not\models \vp$, but for every finite model $\Imc$ of
  $(\Tmc,\Amc)$, we have $\Imc\models\vp$.
\end{example}

%\smallskip\noindent {\bf Acknowledgments.} 
\section*{Acknowledgments} 
The first author was
funded by EU's Horizon 2020 programme under the Marie Sk\l{}odowska-Curie grant 663830, the second one by the FWF project P30360, and  the third
one by the ERC grant 647289 CODA.

\bibliographystyle{aaai}
\bibliography{biblio}

\begin{thebibliography}{}

\bibitem[\protect\citeauthoryear{Baader \bgroup et al\mbox.\egroup
  }{2003}]{handbookDL}
Baader, F.; Calvanese, D.; McGuinness, D.~L.; Nardi, D.; and Patel{-}Schneider,
  P.~F., eds.
\newblock 2003.
\newblock {\em The Description Logic Handbook: Theory, Implementation, and
  Applications}. Cambridge University Press.

\bibitem[\protect\citeauthoryear{Baget \bgroup et al\mbox.\egroup
  }{2017}]{BagetBMT17}
Baget, J.; Bienvenu, M.; Mugnier, M.; and Thomazo, M.
\newblock 2017.
\newblock Answering conjunctive regular path queries over guarded existential
  rules.
\newblock In {\em In Proc.\ of {IJCAI}-17},  793--799.

\bibitem[\protect\citeauthoryear{Bienvenu, Ortiz, and
  Simkus}{2015}]{BienvenuOS15}
Bienvenu, M.; Ortiz, M.; and Simkus, M.
\newblock 2015.
\newblock Regular path queries in lightweight description logics: Complexity
  and algorithms.
\newblock {\em J. Artif. Intell. Res. {(JAIR)}} 53:315--374.

\bibitem[\protect\citeauthoryear{Calvanese \bgroup et al\mbox.\egroup
  }{2000}]{CalvaneseGLV00}
Calvanese, D.; {De Giacomo}, G.; Lenzerini, M.; and Vardi, M.~Y.
\newblock 2000.
\newblock Containment of conjunctive regular path queries with inverse.
\newblock In {\em Proc.\ of {KR}-00},  176--185.

\bibitem[\protect\citeauthoryear{Calvanese, Eiter, and
  Ortiz}{2009}]{CalvaneseEO09}
Calvanese, D.; Eiter, T.; and Ortiz, M.
\newblock 2009.
\newblock Regular path queries in expressive description logics with nominals.
\newblock In {\em Proc.\ of {IJCAI}-09},  714--720.

\bibitem[\protect\citeauthoryear{Calvanese, Eiter, and
  Ortiz}{2014}]{CalvaneseEO14}
Calvanese, D.; Eiter, T.; and Ortiz, M.
\newblock 2014.
\newblock Answering regular path queries in expressive description logics via
  alternating tree-automata.
\newblock {\em Inf. Comput.} 237:12--55.

\bibitem[\protect\citeauthoryear{Dogrusoz \bgroup et al\mbox.\egroup
  }{2009}]{Dogrusoz2009}
Dogrusoz, U.; Cetintas, A.; Demir, E.; and Babur, O.
\newblock 2009.
\newblock Algorithms for effective querying of compound graph-based pathway
  databases.
\newblock {\em BMC Bioinformatics} 10(1):376.

\bibitem[\protect\citeauthoryear{Eiter \bgroup et al\mbox.\egroup
  }{2009}]{EiterLOS09}
Eiter, T.; Lutz, C.; Ortiz, M.; and Simkus, M.
\newblock 2009.
\newblock Query answering in description logics with transitive roles.
\newblock In {\em Proc.\ of {IJCAI}-09},  759--764.

\bibitem[\protect\citeauthoryear{Florescu, Levy, and
  Suciu}{1998}]{FlorescuLS98}
Florescu, D.; Levy, A.~Y.; and Suciu, D.
\newblock 1998.
\newblock Query containment for conjunctive queries with regular expressions.
\newblock In {\em Proc.\ of {PODS}-98},  139--148.

\bibitem[\protect\citeauthoryear{F{\"{u}}rer}{1980}]{Furer80}
F{\"{u}}rer, M.
\newblock 1980.
\newblock The complexity of the inequivalence problem for regular expressions
  with intersection.
\newblock In {\em Proc.\ of {ICALP}-80},  234--245.

\bibitem[\protect\citeauthoryear{Glimm \bgroup et al\mbox.\egroup
  }{2008}]{GlimmLHS08}
Glimm, B.; Lutz, C.; Horrocks, I.; and Sattler, U.
\newblock 2008.
\newblock Conjunctive query answering for the description logic {SHIQ}.
\newblock {\em J. Artif. Intell. Res. {(JAIR)}} 31:157--204.

\bibitem[\protect\citeauthoryear{Glimm, Horrocks, and
  Sattler}{2008}]{GlimmHS08}
Glimm, B.; Horrocks, I.; and Sattler, U.
\newblock 2008.
\newblock Unions of conjunctive queries in {SHOQ}.
\newblock In {\em Proc.\ of {KR}-08},  252--262.

\bibitem[\protect\citeauthoryear{Gr{\"{a}}del and
  Walukiewicz}{1999}]{GradelW99}
Gr{\"{a}}del, E., and Walukiewicz, I.
\newblock 1999.
\newblock Guarded fixed point logic.
\newblock In {\em Proc.\ of {LICS}-99},  45--54.

\bibitem[\protect\citeauthoryear{Guti{\'e}rrez-Basulto,
  Ib{\'a}{\~n}ez-Garc{\'i}a, and Jung}{2017a}]{GuIbJu-AAAI17}
Guti{\'e}rrez-Basulto, V.; Ib{\'a}{\~n}ez-Garc{\'i}a, Y.; and Jung, J.~C.
\newblock 2017a.
\newblock Number restrictions on transitive roles in description logics with
  nominals.
\newblock In {\em Proc.\ of {AAAI-17}}.

\bibitem[\protect\citeauthoryear{Guti{\'e}rrez-Basulto,
  Ib{\'a}{\~n}ez-Garc{\'i}a, and Jung}{2017b}]{GIJ-DL17}
Guti{\'e}rrez-Basulto, V.; Ib{\'a}{\~n}ez-Garc{\'i}a, Y.; and Jung, J.~C.
\newblock 2017b.
\newblock On query answering in description logics with number restrictions on
  transitive roles.
\newblock In {\em Proc.\ of DL-17}.

\bibitem[\protect\citeauthoryear{Horrocks, Sattler, and
  Tobies}{2000}]{HorrocksST00}
Horrocks, I.; Sattler, U.; and Tobies, S.
\newblock 2000.
\newblock Practical reasoning for very expressive description logics.
\newblock {\em Logic Journal of the {IGPL}} 8(3):239--263.

\bibitem[\protect\citeauthoryear{Kaminski and Smolka}{2010}]{KaminskiS10}
Kaminski, M., and Smolka, G.
\newblock 2010.
\newblock Terminating tableaux for $\mathcal{SOQ}$ with number restrictions on
  transitive roles.
\newblock In {\em Proc.\ of the 6th {IFIP} {TC}},  213--228.

\bibitem[\protect\citeauthoryear{Kazakov and
  Pratt{-}Hartmann}{2009}]{KazakovP09}
Kazakov, Y., and Pratt{-}Hartmann, I.
\newblock 2009.
\newblock A note on the complexity of the satisfiability problem for graded
  modal logics.
\newblock In {\em Proc.\ of {LICS}-09},  407--416.

\bibitem[\protect\citeauthoryear{Kazakov, Sattler, and
  Zolin}{2007}]{KazakovSZ07}
Kazakov, Y.; Sattler, U.; and Zolin, E.
\newblock 2007.
\newblock How many legs do {I} have? {N}on-simple roles in number restrictions
  revisited.
\newblock In {\em Proc.\ of {LPAR}-07},  303--317.

\bibitem[\protect\citeauthoryear{Lysenko \bgroup et al\mbox.\egroup
  }{2016}]{Lysenko2016}
Lysenko, A.; Roznov{\u{a}}{\c{t}}, I.~A.; Saqi, M.; Mazein, A.; Rawlings,
  C.~J.; and Auffray, C.
\newblock 2016.
\newblock Representing and querying disease networks using graph databases.
\newblock {\em BioData Mining} 9(1):23.

\bibitem[\protect\citeauthoryear{Rector and Rogers}{2006}]{RectorR06}
Rector, A.~L., and Rogers, J.
\newblock 2006.
\newblock Ontological and practical issues in using a description logic to
  represent medical concept systems: Experience from {GALEN}.
\newblock In {\em Proc.\ of {RW}-06},  197--231.

\bibitem[\protect\citeauthoryear{Stefanoni \bgroup et al\mbox.\egroup
  }{2014}]{StefanoniMKR14}
Stefanoni, G.; Motik, B.; Kr{\"{o}}tzsch, M.; and Rudolph, S.
\newblock 2014.
\newblock The complexity of answering conjunctive and navigational queries over
  {OWL} 2 {EL} knowledge bases.
\newblock {\em J. Artif. Intell. Res.} 51:645--705.

\bibitem[\protect\citeauthoryear{Stevens \bgroup et al\mbox.\egroup
  }{2007}]{StevensAWSDHR07}
Stevens, R.; Aranguren, M.~E.; Wolstencroft, K.; Sattler, U.; Drummond, N.;
  Horridge, M.; and Rector, A.~L.
\newblock 2007.
\newblock Using {OWL} to model biological knowledge.
\newblock {\em International Journal of Man-Machine Studies} 65(7):583--594.

\bibitem[\protect\citeauthoryear{Vardi}{1998}]{Vardi98}
Vardi, M.~Y.
\newblock 1998.
\newblock Reasoning about the past with two-way automata.
\newblock In {\em Proc.\ of {ICALP}-98},  628--641.

\bibitem[\protect\citeauthoryear{Wolstencroft \bgroup et al\mbox.\egroup
  }{2005}]{Wolstencroft2005}
Wolstencroft, K.; Brass, A.; Horrocks, I.; Lord, P.; Sattler, U.; Turi, D.; and
  Stevens, R.
\newblock 2005.
\newblock A little semantic web goes a long way in biology.
\newblock In {\em Proc.\ of {ISWC}-05}.

\end{thebibliography}

\clearpage
% !TEX root = main.tex 
\appendix

\begin{center} {\huge \textbf{APPENDIX}} \end{center}
\section*{Additional Preliminaries}
% The semantics of P2RPQs  is defined as follows. Let $\Imc$ be an
% interpretation and \Emc a regular expression over \Smc.  We write
% $(d,e)\in \Emc^\Imc$, if there is a word $w_1 \dots w_n \in L(\Emc)$ and a
% sequence $d_0, \ldots, d_n \in \Delta^\Imc$ such that $d_0 = d, d_n
% =e$, and for all $i \in [1,n]$ we have that $(i)$ if $w_i = A?$, then
% $d_{i-1}=d_i \in A^\Imc$, and $(ii)$ if $w_i = r \in
% \mn{N}_\mn{R}^\pm$, then $(d_{i-1}, d_i) \in r^\Imc$. 
% 
% Let $q(\vec{x}) = \exists \vec{y}.\varphi$ be a P2RPQ with $\vec{x}=
% (x_1, \ldots, x_k)$, \Imc an interpretation. A map $\pi$ from the set
% of terms in $q$ to $\Delta^\Imc$ is a \emph{match for $q$ in \Imc} if
% $(i)$ $\pi(c)=c$ for all individual names $c$ in $q$, $(ii)$ $\pi(t)
% \in A^\Imc$ for every $A(t) \in q$, and $(iii)$ $(\pi(t),\pi(t'))\in
% \Emc^\Imc$, for every $\Emc (t,t') \in q$.  A tuple $\vec c = (c_1,
% \ldots, c_k) \in \mn{ind}(\Amc)^k$ is an \emph{answer to $q$ in \Imc},
% written $\Imc \models q(\vec c)$, if there is a match $\pi$ for $q$ in
% \Imc with $\pi(x_i)=c_i$, for all $i\in[1,k]$. Moreover, $\vec c$ is a
% \emph{certain answer to $q$ over a KB \Kmc}, written $\Kmc \models
% q(\vec c)$, if $\Imc \models q(\vec c)$ for  every model $\Imc$ of
% \Kmc. The set of all certain answers to $q$ over \Kmc is denoted 
% $\mn{cert}(q,\Kmc)$. 
% 

\textbf{Homomorphisms.} Let $\Imc_1$ and $\Imc_2$ be two
interpretations. A {\em homomorphism from $\Imc_1$ to $\Imc_2$} is a
function $h:\Delta^{\Imc_1}\to \Delta^{\Imc_2}$ such that {\it
(i)}~$h(a)=a$ for all $a\in \mn{N_I}$, {\it (ii)}~if $d\in
A^{\Imc_1}$, then $h(d)\in A^{\Imc_2}$, for all $A\in\mn{N_C}$, and
{\it (iii)} if $(d,e)\in r^{\Imc_1}$, then $(h(d),h(e))\in
r^{\Imc_2}$, for all $r\in\mn{N_R}$. It is folklore that PRQPs are
\emph{preserved under homomorphisms}, that is, if $\Imc_1\models
\varphi$ and there is a homomorphism from $\Imc_1$ to $\Imc_2$, then
$\Imc_2\models \varphi$.

\smallskip \noindent \textbf{Semantics of 2ATAs.} A \emph{run} of $\Amf$ on a labelled
tree $(T,\tau)$ is a $T\times Q$-labelled tree $(T_r, r)$ such
that $r(\varepsilon)=(\varepsilon,q_0)$  and whenever $x \in T_r$, $r(x)
=(w,q)$, and $\delta(q,\tau(w))= \theta$, then there is a set
$\mathcal S = \{(m_1,q_1), \ldots, (m_n,q_n)\} \subseteq [k] \times Q$
such that $\mathcal S$ satisfies~$\theta$ and for $1 \leq i \leq n$,
we have $x\cdot i \in T_r$, $w \cdot m_i$ is defined, and
$\tau_r(x\cdot i)= (w\cdot m_i, q_i)$. A run is
\emph{accepting} if
every infinite path $\pi$ satisfies the \emph{parity condition}.  A
\emph{parity condition $F$ over $Q$} is a finite sequence $G_1,
\ldots, G_m$ with $G_1 \subseteq G_2 \subseteq \ldots \subseteq G_m
=Q$. An infinite path $\pi$ satisfies $F$ if there is an even $i$ such
that $\mn{inf}(\pi) \cap G_i \neq \emptyset$ and $\mn{inf}(\pi) \cap
G_{i-1} = \emptyset$, where $\mn{inf}(\pi) \subseteq Q$ denotes the
set of states that occur infinitely often in $\pi$. The automaton
accepts an input tree if there is an accepting run for it. We use 
$L(\Amf)$ to denote the set of trees accepted by \Amf. The
\emph{nonemptiness problem} is to decide, given a 2ATA $\Amf$,
whether $L(\Amf)$ is nonempty.

\section{Proof of Lemma~\ref{lem:outdegree-width}}

\noindent\textbf{Lemma~\ref{lem:outdegree-width}.}
\textit{
%   For each $\Imc\models \Kmc$, there is some $\Imc'\models\Kmc$ such
%   that there is a homomorphism from $\Imc'$ to \Imc and width and
%   breadth of $\Imc'$ are bounded by $O(2^{\mn{poly}(|\Tmc|)})$ and
%   $O(|\Amc|+2^{\mn{poly}({|\Tmc|})})$, respectively.
  %
  For each $\Imc\models \Kmc$, there is a sub-interpretation $\Imc'$
  of \Imc with $\Imc'\models\Kmc$ and width and
  breadth of $\Imc'$ are bounded by $O(|\Amc|+2^{\mn{poly}({|\Tmc|})})$.
}

\begin{proof} We show the lemma in two stages, adapting a technique
  from~\cite{KazakovP09,GuIbJu-AAAI17}. 
  
  Let $\hat m$
  be the maximal number appearing in \Tmc.

  \smallskip \textit{Stage 1 (Bounded breadth)}. As it is standard to
  achieve bounded breadth for non-transitive roles~\cite{GlimmLHS08},
  we only deal with transitive roles here. 

  An element $e$ is an \emph{strict $r$-successor of $d$} if $(d,e)\in
  r^\Imc$, but $e\notin Q_{\Imc,r}(d)$. 
  Let $W_r(d)$ be the set of
  strict $r$-successors of $d$ and $W_r(d,C)\subseteq W_r(d)$ be the set of
  all strict $r$-successors of $d$ satisfying $C$. 
  Then, fix a subset $W_r'(d)\subseteq W_r(d)$ by adding, for each $C\in
  \mn{cl}(\Tmc)$, $\min(\hat m, |W_r(d,C)|)$ elements from $W_r(d,C)$.

  Assume without loss of generality that $W_r'(d_1)=W_r'(d_2)$ if
  $d_1\in Q_{\Imc,r}(d_2)$, and define relations $S_r^1$, $S_r^2$, and
  $S_r^3$, for each $r\in \mn{Rol}_t(\Kmc)$, as follows: 
   \begin{align*}
     S_r^1 & = \{(d,d')\in r^\Imc\mid d'\in Q_{\Imc,r}(d)\}; \\
     S_r^2 & = \{(d,d')\in r^\Imc\mid r(d,d')\in\Amc\}; \\
     S_r^3 & = \{(d,d')\in r^\Imc\mid d'\in W'_r(d)\}.
   \end{align*}
   Intuitively, $S_r^1$ is the restriction of $r^\Imc$ to the
   clusters, $S_r^2$ takes care of the ABox, and $S_r^3$ keeps a
   sufficient set of successors to witness all at-least restrictions.

   Finally, obtain $\Imc'$ from $\Imc$ by taking $\Delta^{\Imc'}=\Delta^{\Imc}$,
   $A^{\Imc'}=A^{\Imc}$ for all concept names $A$, $r^{\Imc'}=r^\Imc$,
   for all non-transitive roles $r$, and, for all transitive roles
   $r$,
   $$r^{\Imc'}= (S_r^1\cup S_r^2\cup S_r^3)^+.$$

   \smallskip\noindent{\it Claim 1.} $C^\Imc=C^{\Imc'}$, for all
   $C\in\mn{cl}(\Tmc)$.

   \smallskip\noindent{\it Proof of Claim 1.} This is shown by
   induction on the structure of concepts. The only non-trivial case
   are concepts $C=\qnrleq n r D$, $r$ transitive. Clearly, $d\in C^{\Imc}$ implies $d\in C^{\Imc'}$ since
   $r^{\Imc'}\subseteq r^{\Imc}$. The converse is a direct consequence
   of the definition of $W_r'(d)$ and $S_r^3$. In particular, only
   $r$-successors that ``cannot be seen'' by at-most restrictions (due
   to the choice of $\hat m$) are removed. 

   \smallskip From Claim~1, we conclude that $\Imc'\models\Tmc$; by
   Claim~1 and the definition of $r^{\Imc'}$, particularly $S_r^2$,
   we also have $\Imc'\models\Amc$, thus $\Imc\models\Kmc$. Since
   $r^{\Imc'}\subseteq r^\Imc$ and $A^\Imc=A^{\Imc'}$, for all $A\in
   \mn{N_C}$, the identity is an homomorphism from $\Imc'$ to \Imc.
   Finally note that, by construction, the breadth of $\Imc'$
   is at most $|\Amc|+|\mn{cl}(\Tmc)|\cdot \hat m$ and thus
   $O(|\Amc|+2^{\mn{poly}(|\Tmc|)})$.

   \smallskip \textit{Stage 2 (Bounded Width).}
   For every transitive role $r$, and every $d\in \Delta^\Imc$, fix a set
   $W_{r}(d)\subseteq Q_{\Imc,r}(d)$ as follows.  For each $C\in
   \mn{cl}(\Tmc)$, $W_{r}(d)$ contains the set $Q_{\Imc,r}(d)\cap
   C^\Imc$ if this set has size at most $\hat m$, and otherwise a
   subset thereof having size $\hat m$. Without loss of generality, we
   assume that $W_{r}(d)=W_{r}(e)$ for all $e\in
   Q_{\Imc,r}(d)$.  
   Now, define a set $\Delta_r$, for each transitive $r$, by
   taking
   $$\Delta_r=\mn{ind}(\Amc)\cup \bigcup_{d\in \Delta^\Imc}
   W_{r}(d), $$
   and define an interpretation $\Imc'=(\Delta^{\Imc'},\cdot^{\Imc'})$
   by setting $\Delta^{\Imc'} = \Delta^\Imc$, $A^{\Imc'}= A^\Imc,
   \text{ for all }A\in\mn{N_C}$, $r^{\Imc'} = r^\Imc, \text{ for all
   non-transitive roles $r$}$, and
   $$r^{\Imc'} = r^\Imc \cap (\Delta^\Imc\times\Delta_r), \text{ for all
   transitive roles $r$.} $$
   It is not hard to verify that $r^{\Imc'}$ is indeed transitive.

   \smallskip\noindent{\it Claim 2.} $C^\Imc=C^{\Imc'}$, for all
   $C\in\mn{cl}(\Tmc)$. 

   \smallskip\noindent{\it Proof of Claim 2.} This is again shown by
   induction on the structure of concepts. The only non-trivial case
   are concepts $C=\qnrleq n r  D$, $r$ transitive. Clearly, $d\in
   C^{\Imc}$ implies $d\in C^{\Imc'}$ since $r^{\Imc'}\subseteq
   r^{\Imc}$. The converse is a direct consequence of the definition
   of $W_{r}(d)$, in particular the choice of $\hat m$, and the
   definition of $r^{\Imc'}$. In particular, we remove only
   $r$-successors that cannot contribute to at-least restrictions. 

   \smallskip Based on Claim~2, it is easy to see that $\Imc'\models
   \Tmc$ and $\Imc'\models \Amc$. Moreover, the identity is a
   homomorphism from $\Imc'$ to \Imc. Finally, by definition of $W_{r}(d)$, particularly the choice of $\hat
   m$, it should be clear that the width of $\Imc'$ is bounded by
   $|\Amc|+2^{\mn{poly}(|\Tmc|)}$.
%   \qed
\end{proof}

\section{Properties of $\mn{Wit}_{\Imc,r}(d)$}

We next verify two properties of the witness set
$\mn{Wit}_{\Imc,r}(d)$, which are needed later on.  Throughout the
following Lemmas, we denote with $W_{\Imc,r}^\rightsquigarrow(d)$ the
set $\{e\mid d\rightsquigarrow_{\Imc,r}^*e\}$.

\medskip\noindent\textbf{Lemma~\ref{lem:wit}.} \textit{For every $d\in \Delta^\Imc$ and transitive $r$, we have
  $|\mn{Wit}_{\Imc,r}(d)|\leq |\Amc|\cdot 2^{p(|\Tmc|)}$, for some
polynomial $p$.}

\begin{proof} 
  We construct a tree $T$ labeled with elements from $\Delta^\Imc$. We
  start with the single node tree $d$. Then, we exhaustively performing the following operation:
  \begin{itemize}

    \item[$(\ast)$] Choose a leaf labeled with $e$ and add, for all $f\in
      \Delta^\Imc\setminus T$ with $e\rightsquigarrow_{\Imc,r}f$,
      $f$ as a successor of $e$ in $T$.

   \end{itemize}
   By definition of $\rightsquigarrow_{\Imc,r}$ and~$(\ast)$, the
   obtained graph is indeed a tree which additionally satisfies
   $W_{\Imc,r}^\rightsquigarrow(d)\subseteq T$. Now, consider the labelling
   $\ell:T\to 2^{\mn{cl}(\Tmc)}$ given by
   $$\ell(e)=\{C\mid e\in {\qnrleq n r C}^{\Imc}, \qnrleq n r C \in\mn{cl}(\Tmc)\}.$$ 
   Let $f$ be a successor of $e$ in $T$. By construction of $T$, this implies
   \begin{itemize}

     \item[--] $\ell(e)\subseteq \ell(f)$ if $f$ is a leaf in $T$;

     \item[--] $\ell(e)\subsetneq \ell(f)$ if $f$ is not a leaf in $T$.

   \end{itemize}
   Thus, the depth of $T$ is bounded by $|\Tmc|$. Since, for any $e$,
   there are at most exponentially (in \Tmc) many $f$ such that
   $e\rightsquigarrow_{\Imc,r} f$, we know that the outdegree of
   $(T,E)$ is bounded exponentially in \Tmc. Overall, we get that the
   size of $T$, and thus of the set $W_{\Imc,r}^\rightsquigarrow(d)$
   is bounded by an exponential in \Tmc. Note next that, by
   Lemma~\ref{lem:outdegree-width}, for every $f\in
   W_{\Imc,r}^\rightsquigarrow(d)$, we have $Q_{\Imc,r}(d)\subseteq
   \mn{ind}(\Amc)\cup X_d$, for some set $X_d$ of size bounded by
   $2^{p(|\Tmc|)}$, $p$ a polynomial.  As $\mn{Wit}_{\Imc,r}(d)=\bigcup_{e\in
     W^{\rightsquigarrow}_{\Imc,r}(d)} Q_{\Imc,r}(e)$, this implies
     the statement in the lemma.  \end{proof}

\begin{lemma} \label{lem:wit2}
  Let $d\in\Delta^\Imc$ and $r$ transitive. Then for all
  $e\in\mn{Wit}_{\Imc,r}(d)$, we have that
  $\mn{Wit}_{\Imc,r}(e)\subseteq \mn{Wit}_{\Imc,r}(d)$.
\end{lemma}

\begin{proof} 
  Let $e\in \mn{Wit}_{\Imc,r}(d)$.  By definition of
  $\mn{Wit}_{\Imc,r}$, it suffices to show that
  $W_{\Imc,r}^\rightsquigarrow(e)\subseteq \mn{Wit}_{\Imc,r}(d)$.  To
  this end, suppose $f\in W^\rightsquigarrow_{\Imc,r}(e)$. By
  definition of $W_{\Imc,r}^\rightsquigarrow$, there is a sequence
  $e_1\rightsquigarrow_{\Imc,r} \cdots \rightsquigarrow_{\Imc,r}e_n$
  with $e=e_1$ and $f=e_n$ (possibly $n=1$).  As $e\in
  \mn{Wit}_{\Imc,r}(d)$, we have either~\textit{(i)} $e\in
  W^\rightsquigarrow_{\Imc,r}(d)$ or~\textit{(ii)} there is some
  $e'\in W_{\Imc,r}^\rightsquigarrow(d)$ such that $e\in
  Q_{\Imc,r}(e')$.  We distinguish cases. 
  \begin{itemize}

     \item[\textit{(i)}] $e\in W^\rightsquigarrow_{\Imc,r}(d)$ implies that
       there is a sequence
       $d_1\rightsquigarrow_{\Imc,r}\cdots\rightsquigarrow_{\Imc,r}d_m$
       with $d_1=d$ and $d_m=e$. Thus, there is a sequence
       $d_1\rightsquigarrow_{\Imc,r}\cdots \rightsquigarrow_{\Imc,r}
       d_m=e=e_1 \rightsquigarrow_{\Imc,r} \cdots
       \rightsquigarrow_{\Imc,r}e_n=f$. Hence, $f\in
       W^\rightsquigarrow_{\Imc,r}(d)\subseteq\mn{Wit}_{\Imc,r}(d)$.

     \item[\textit{(ii)}] Similar to Case~\textit{(i)}, there is a
       sequence
       $d_1\rightsquigarrow_{\Imc,r}\cdots\rightsquigarrow_{\Imc,r}d_m$
       with $d_1=d$ and $d_m=e'$. If $e=f$, that is $n=1$ in the
       sequence above, we know that $f\in Q_{\Imc,r}(e')$ and thus
       $f\in \mn{Wit}_{\Imc,r}(d)$. Otherwise, observe that we can
       assume that $|Q_{\Imc,r}(d)|\geq 2$ (otherwise $e'=e$ and we
       are in Case~\textit{(i)}). Thus, we have $e'\in{\qnrsim \ell  r
       C}^\Imc$ iff $e\in {\qnrsim \ell  r  C}^\Imc$, for all $\sim$, $\ell$, and
       $C$, and hence also $e'\rightsquigarrow_{\Imc,r}e_2$ implying
       that $d_1\rightsquigarrow_{\Imc,r}\cdots
       \rightsquigarrow_{\Imc,r} d_m=e'\rightsquigarrow_{\Imc,r}
       e_2\rightsquigarrow_{\Imc,r} \cdots
       \rightsquigarrow_{\Imc,r}e_n=f$. Hence, $f\in
       W^\rightsquigarrow_{\Imc,r}(d)\subseteq\mn{Wit}_{\Imc,r}(d)$.

   \end{itemize}
   %
%    \qed
\end{proof}

\section{Proof of Theorem~\ref{thm:sq-canmod}}\label{app:thm1}

Before we establish Theorem~\ref{thm:sq-canmod}, we prove two auxiliary
lemmas, which establish how to address in a unique way
$r$-successors in canonical decompositions.  For the first auxiliary
lemma, observe that as a consequence of Definition~\ref{def:candec},
particularly, Condition~\ref{it:can3}, for every $d\in \Delta^\Jmc$,
$r\in \mn{N}_\mn{R}^t$, there is a unique node $w \in T$ with
$\bago(w)=r$ and $d \in \roots(w)$.  We denote this node with
$w_{d,r}$. 

\begin{lemma} \label{lem:roots}
  Let $r\in \mn{N}_\mn{R}^t$. For every $u\in T$ with $\bago(u)=r$ and
  $(d,e) \in r^{\bagz(u)}$, exactly one of the following holds: 
\begin{itemize} 

  \item $w_{d,r}=w_{e,r}$ and $(d,e) \in r^{\bagz(w_{d,r})}$; 

  \item $w_{e,r}$ is a successor of $w_{d,r}$, $(d,e) \in
   r^ {\bagz(w_{e,r})}$ and $d$ belongs to an $r$-root cluster in $w_{e,r}$; 

  \item $w_{e,r}$ is an ancestor of $w_{d,r}$ and $(d,e) \in
    r^{\bagz(w_{d,r})}$.

\end{itemize}
\end{lemma} 

\begin{proof}

  Since $d,e\in\Delta_u$, we know that $w_{d,r}$ and
  $w_{e,r}$ are either equal to $u$ or ancestors of $u$. We
  distinguish three cases:
  \begin{itemize}

    \item  If $w_{d,r}=w_{e,r}$, then, by
      Definition~\ref{def:treedecomp},  $d,e\in \Delta_{v'}$ for every $v'$
      on the path from $w_{d,r}$ to $u$,  and 
      $(d,e)\in r^{\bagz(v')}$ for every such $v'$.  Therefore,  $(d,e)\in
      r^{\bagz(w_{d,r})}$.
%     If $w_{d,r}=w_{e,r}$, then $d,e\in \Delta_{v'}$ for every $v'$
%       on the path from $w_{d,r}$ to $u$, by
%       Definition~\ref{def:treedecomp}~\ref{it:tree3}.  By
%       Definition~\ref{def:treedecomp}, we obtain that
%       $(d,e)\in r^{\bagz(v')}$ for every such $v'$, hence $(d,e)\in
%       r^{\bagz(w_{d,r})}$.

    \item If $w_{d,r}$ is an ancestor of $w_{e,r}$, then we know by
      the same reasoning as in the previous point that $(d,e)\in
      r^{\bagz(w_{e,r})}$.  Let $w' = w_{e,r}\cdot -1$  (the predecessor of $w_{e,r}$).
      Since $\Tmf$ is a canonical
      decomposition, either~\ref{it:can3} or~\ref{it:can4} applies to $w'$ and
      $w_{e,r}$. Assume first that $w'=w_{d,r}$. 
      \begin{itemize}

	\item In case of \ref{it:can3}, since $d \in \Delta_{w_{d,r}} \cap
	  \Delta_{w_{e,r}}$ we know that there is a $r$-root
	  cluster $\abf \subseteq \Delta_{w_{e,r}}$ such that $d
	  \in \abf$. 

	\item  In case of \ref{it:can4}, let $\abf \subseteq \roots(w')$ be
	  the cluster witnessing this.  By Item~\textit{(b)}, $\abf$
	  is an $r$-root cluster in $\bagz(w_{e,r})$. By definition,
	  we know $e \in F_r(w_{e,r})$ and thus $e\in F(w_{e,r})$.
	  From this and Item~\textit{(d)}, we obtain that $d \in \abf
	  \cup F(w_{e,r})$, and since $d \not \in F(w_{e,r})$ we know
	  $d \in \abf$. 

      \end{itemize}
      Thus, in both cases, we are in the second case of the lemma. 
      Assume now that $w'\neq w_{d,r}$. We show that it leads to 
      a contradiction in both cases:
      \begin{itemize}

	\item In case of~\ref{it:can3}, since $d \in \Delta_{\bagz(w_{e,r})} \cap
	  \Delta_{\bagz(w')}$ we know $d\in F(w')$. On the other hand, 
	  $d\in \roots(w_{d,r})$ implies that either $d \in
	  F(w_{d,r})$ or $w_{d,r}$ has a predecessor $w''$ such that $d
	  \in F(w'')$. This is a contradiction since $w'\neq w''$
	  since $w_{d,r}$ is an ancestor of $w_{e,r}$.

	\item In case of~\ref{it:can4}, let $\abf$ be the $r$-cluster
	  witnessing this.  By definition,  $e \in \roots(w_{e,r})$, implies $e\in F(w_{e,r})$.
	  Since $w_{d,r}\neq w'$ but $w_{d,r}$ is an ancestor of
	  $w_{e,r}$, we know that $w' \neq \varepsilon$ and
	  $\bago(w') = r$. From Item~\textit{(d)} we obtain that
	  $d\in \abf\cup F(w_{e,r})$, and  since $d\notin F(w_{e,r})$,
	  we know $d\in \abf$.  By Item~\textit{(a)}, we know that
	  $\abf\subseteq F_r(w')$, but then $w'=w_{d,r}$,
	  contradiction. 
	  
      \end{itemize}

    \item If $w_{e,r}$ is an ancestor of $w_{d,r}$, then we know by the reasoning
      in the first point that $(d,e)\in r^{\bagz(w_{d,r})}$; thus, we are in the
      last case of the lemma.
      
  \end{itemize}
%   
%   \qed
\end{proof}

The second auxiliary lemma now provides a way to address
$r$-successors in canonical tree decompositions.  For this purpose, we
introduce the notion of $r$-paths. Let $(T,\bagz,\bago)$ be a
canonical decomposition of an interpretation~\Imc. An \emph{$r$-path
from $d$ to $e$ in $(T,\bagz,\bago)$} is a sequence
$d_0,w_0,d_1,\ldots,w_{n-1},d_n$ such that $d=d_0$, $e=d_n$, and
$(d_i,d_{i+1})\in r^{\bagz(w_i)}$, for all $0\leq i<n$. It is
\emph{downward} if, for all $0< i<n$, $w_{i}$ is a successor of
$w_{i-1}$ and $d_{i}$ is contained in an $r$-root cluster of $w_{i}$.

We then have: 

\begin{lemma} \label{lem:can-paths}
  Let $(T,\bagz,\bago)$ be a canonical decomposition of an
  interpretation~\Imc. We have that $(d,e)\in r^{\Imc}$ iff one of the
  following is true:
  \begin{itemize}

    \item[--] $r$ is non-transitive and $(d,e)\in
      r^{\bagz(\varepsilon)}$ or
      $(d,e)\in r^{\bagz(v)}$ for some successor $v$ of the unique
      $w$ where $d$ is fresh;

    \item[--] $r$ is transitive and there is an $r$-path
      $d_0,w_0,\ldots,d_n$ from $d$ to $e$ in $(T,\bagz,\bago)$ such that one of
      the following holds: 
      \begin{enumerate}[label=\textbf{\Alph*},leftmargin=*,align=left]

	\item\label{it:p1} $d_0\in\roots(w_0)\cup\roots(w_0\cdot
	  -1)$, $d_1\in\roots(w_0)$, and  
	  $d_0,\ldots,d_n$ is downward, or

	\item\label{it:p2} $d_0\in\roots(w_0)$,
	  $d_1\notin\roots(w_0)$, and if $n>1$, then $d_1,\ldots,d_n$
	  is downward and $w_1{\cdot}{-}1$ is an ancestor of $w_0$
	  with $d_1 \in F_r(w_1{\cdot}{-}1)$. 

      \end{enumerate}

    \end{itemize}
  
\end{lemma}

\begin{proof}
  Let first be $(d,e)\in r^\Imc$ for some non-transitive role $r$. The
  direction $(\Leftarrow)$ is immediate. For $(\Rightarrow)$, assume
  that $(d,e)\notin r^{\bagz(\varepsilon)}$. By
  Condition~\ref{it:can1}, $(d,e)\notin r^{\bagz(w)}$, for all
  $w\neq \varepsilon$ such that $\bago(w)\neq r$. The statement then
  follows from Condition~\ref{it:can2}.

  Let now be $r$ transitive. Again, the direction $(\Leftarrow)$ is
  trivial. For $(\Rightarrow)$, $(d,e)\in r^\Imc$
  implies, by definition of 
  tree decomposition, that there is an $r$-path
  $d_0,w_0,\ldots,w_{n-1},d_n$ from $d$ to $e$ in $(T,\bagz,\bago)$.
  We show first that it is without loss of generality to assume that
  for all $0\leq j<n-1$, we have:

  \begin{enumerate}[label=(\alph*)]

    \item $w_{j+1}=w_j$, 

    \item $w_{j+1}$ is a successor of $w_{j}$, $d_{j+1} \in
      \roots(w_{j})$ and  $d_{j+1}$ belongs to an $r$-root cluster in $\bagz(w_{j+1})$

    \item $w_{j+1}$ is an ancestor of $w_{j}$ and $d_{j+1}\in
      \roots(w_{j+1})$, or

    \item The predecessor of $w_{j+1}$ is an ancestor of $w_j$,
      $d_j\in\roots(w_j)$, and $d_{j+1}\in\roots(w_{j+1})$.

  \end{enumerate}
 Observe that, by Lemma~\ref{lem:roots}, we can assume that $w_i \in
  \{w_{d_i,r},w_{d_{i+1},r}\}$. Moreover, if $w_i=w_{d_i,r}$, then
  either $w_{d_{i},r}=w_{d_{i+1},r}$ or $w_{d_{i+1},r}$ is an ancestor
  of $w_{d_i,r}$ and $d_{i+1}\in\roots(w_{j+1})$; if
  $w_i=w_{d_{i+1},r}$, then $w_{d_{i+1},r}$ is a successor of
  $w_{d_{i},r}$ and $d_{i+1}\in\roots(w_i)$. Let now be $0\leq j<
  n-1$. We distinguish four cases:
  \begin{itemize}

    \item If $w_j=w_{d_j,r}$ and $w_{j+1}=w_{d_{j+1},r}$, then
      Case~(a) or Case~(c) applies.

    \item If $w_j=w_{d_{j+1},r}$ and $w_{j+1}=w_{d_{j+1},r}$, then
      Case~(a) applies.

    \item If $w_j=w_{d_j,r}$ and $w_{j+1}=w_{d_{j+2},r}$, then
      Case~(b) or~(d) applies.
    \item If $w_j=w_{d_{j+1},r}$ and $w_{j+1}=w_{d_{j+2},r}$, then
      Case~(b) applies.

  \end{itemize}
  Note then, that in case~(a) is satisfied for some $j$, we can safely
  drop $d_{j+1}$ and $w_{j}$ and the remaining sequence is still an
  $r$-path, due to Definition~\ref{def:treedecomp} (item~\ref{it:tree2}). %item  (iv) of the definition of tree decomposition.
  So from now on, we assume that for all $0\leq j<n-1$, one of~(b)--(d)
  is the case.

  If Condition~(b) applies for all $j$ then, by the second item in Lemma~\ref{lem:roots}, the $r$-path
  is downward and it satisfies \ref{it:p1}. 
  Otherwise, we modify the sequence by performing the following operation
  exhaustively. Let $0\leq k < n-1$ be some index
  satisfying~(c), that is, $w_{k+1}$ is an ancestor of $w_k$, and let
  $k'$ be minimal such that all $i$ with $k'\leq i<k$ satisfy~(b). If
  $k'=k$, then do nothing, otherwise we distinguish the following
  cases: 

\begin{description}

\item[Case 1:] $w_{k+1}=w_j$ for some $k'\leq j<k$. We show
  inductively that then $(d_i,d_{k+1})\in r^{\bagz(w_i)}$, for all $j\leq
  i\leq k$. For $i=k$ it is clear by assumption. For the inductive
  step, assume $j\leq i<k$.  Clearly, we have
  $(d_i,d_{i+1})\in r^{\bagz(w_i)}$ and, by the choice of $k$ and the
  assumption $w_j=w_k$, also
  $(d_{k+1},d_{k+2})\in r^{\bagz(w_{j})}$. Moreover, by induction, we can
  assume that $(d_{i+1},d_{k+1})\in r^{\bagz(w_{i+1})}$. By
  the definition of tree decomposition (item \ref{it:tree3}), we know that
  $d_{k+1}\in \Delta_{w_i}$; and that $r(d_{i+1},d_{k+1})\in r^{\bagz(w_i)}$. 
  Further,  the definition of tree decomposition yields also
  $(d_{i},d_{k+1})\in r^{\bagz(w_i)}$, thus finishing the inductive step.
%   \nb{Y: I am not sure about this argument. Before we used, ``overlapping isomorphic'' as argument}

  This implies $(d_j,d_{k+1})\in r^{\bagz(w_j)}$. Since also
  $(d_{k+1},d_{k+2})\in r^{\bagz(w_j)}$, we know
  $(d_{j},d_{k+2})\in r^{\bagz(w_j)}$.
  Thus, dropping the subsequence $$d_{j+1},w_{j+1},\ldots,w_{k+1}$$ 
  yields an $r$-path satisfying~(b)--(d) for all $j$.

\item[Case 2:] $w_{k+1}$ is an ancestor of $w_{k'}$. We can argue as in
  Case~1 that $(d_{k'},d_{k+1})\in r^{\bagz(w_{k'})}$. Thus, we can drop the
  subsequence $d_{k'+1},\ldots,d_k,w_k$ obtaining an $r$-path which
  satisfies~(b)--(d), for all $j$.

\end{description}
We can deal similarly with an index satisfying~(d).  After performing
this step exhaustively, we obtain an $r$-path
$e_0,v_0,\ldots,v_{m-1},e_m$ from $d$ to $e$ which is downward, and satisfies \ref{it:p1} or \ref{it:p2}, or
\begin{itemize}

  \item[($\ast$)] there is some $0\leq j<m$ such
    that~(c) holds for all $0\leq i<j$, and if $j<m-1$, then (d) holds
    for $j$, and~(b) holds for all $j< i <m$.

\end{itemize}
In case of~$(\ast)$, we show how to obtain an $r$-path
satisfying~$(\ast)$ with $j=0$. 

\smallskip\noindent
\textit{Claim.} If $j\geq 1$, then $(e_0,e_2)\in r^{\bagz(v_0)}$.

\smallskip\noindent\textit{Proof of the Claim.} 
We show inductively that $(e_1,e_2)\in r^{\bagz(u)}$ for all $u$ on the
path between $v_1$ and $v_0$. It is obviously true for $u=v_1$.

Let now $u$ be the successor of some $u_0$ on the path from $v_{1}$ to
$v_0$, and assume by induction that $u_0$ satisfies $(e_1,e_{2})\in
\Delta_{u_0}$. Suppose that \ref{it:can3} holds for $u$. Then
$\bago(u)\neq r$. But since $\bago(w_0)=r$, we know that \ref{it:can3} holds
again for some node between $u$ and $v_0$. The only possible witness
for this is $e=e_{2}$.  However, this leads to a contradiction as
well, because $e_{2}\notin F(w)$ for any $w$ on the path between $u$
and $v_{0}$. Hence, we know that \ref{it:can4} holds for $u$. Let \abf be
the $r$-cluster in $\bagz(u)$ witnessing this.
\begin{itemize}

  \item If $e_{2}\in \abf$, \ref{it:can4} (b) implies that $e_{2}\in
    \Delta_u$, since 
    $(e_{1},e_{2})\in r^{\bagz(u_0)}$. 
  \item If $e_{2}\notin \abf$, then we know by~\ref{it:can4} that
    $(e,e_{2})\in r^{\bagz(u)}$, for some $e\in\abf$. Thus,
    $(e,e_{2})\in r^{\bagz(u_0)}$. Again, \ref{it:can4}(b) implies that
      $e_{2}\in\Delta_u$.

  \end{itemize}
  By the definition of tree decomposition, we obtain in both cases $(e_1,e_2)\in
  r^{\bagz(u)}$, thus finishing the induction.  Since also
  $(e_0,e_1)\in r^{\bagz(v_0)}$, we obtain $(e_0,e_2)\in r^{\bagz(w_0)}$.
  This finishes the proof of the Claim.

  \smallskip
  It is now easy to verify that dropping $e_1,v_1$ from the sequence
  preserves~$(\ast)$, but with $j$ and $m$ decreased by one. By the
  Claim, we can perform this operation repeatedly until $j=0$.

  We argue that the remaining $r$-path satisfies either \ref{it:p1} or \ref{it:p2}.  \begin{itemize}
    \item If $m=1$, we distinguish cases according to
      Lemma~\ref{lem:roots}:

      \begin{itemize}

	\item if $w_{e_0,r}=w_{e_1,r}$, then $e_0,w_{e_1,r},e_1$ is a
	  downward path from $d$ to $e$ satisfying \ref{it:p1};

	\item if $w_{e_1,r}$ is a successor of $w_{e_0,r}$, then
	  $e_0,w_{e_1,r},e_1$ is a downward path from $d$ to $e$ satisfying \ref{it:p1};

	\item if $w_{e_1,r}$ is an ancestor of $w_{e_0,r}$, then
	  $e_0,w_{e_0,r},e_1$ is an $r$-path from $d$ to $e$
	  satisfying \ref{it:p2}.

      \end{itemize}

  \end{itemize}
  %
  %Thus, in all cases, we obtain an $r$-path from $d$ to $e$ satisfying either \ref{it:p1} or \ref{it:p2}.
  In case $m>2$, the resulting path satisfies~\ref{it:p2} because
  of~$(\ast)$, in particular, (d) holds for $0$ and (b) holds for all
  $0<j<m$.
\end{proof}

We restate Theorem~\ref{thm:sq-canmod} and give the missing details
from the proof.

\medskip\noindent\textbf{Theorem~\ref{thm:sq-canmod}.}\textit{
  Let $\Kmc=(\Tmc, \Amc)$ be an \SQ KB and $\vp$ a PRPQ with
  $\Kmc\not\models\vp$. There is a model $\Jmc$ of $\Kmc$ and a
  canonical tree decomposition $(T,\Imf ,\bago)$ of $\Jmc$ with
  (i)~$\Jmc\not\models\vp$, (ii) $\Imf(\varepsilon) \models \Amc$, and
  (iii) width and outdegree of $(T,\Imf)$ are bounded by
  $O(|\Amc|\cdot 2^{p(|\Tmc|)})$, for some polynomial $p$.
}

\begin{proof}
  Let $\Jmc$ and $(T,\bagz,\bago)$ be the interpretation and the tree
  decomposition obtained by the unraveling procedure in the main part.

  We first verify that $(T,\bagz,\bago)$ is indeed a tree
  decomposition of $\Jmc$. Items~\ref{it:tree1} and~\ref{it:tree2} of
  $(T,\bagz)$ being a tree decomposition of \Jmc are an immediate
  consequence of the definition of $\Jmc$ and $\bagz$.
  Item~\ref{it:tree3} is a consequence of the nature of the rules. In
  particular, each rule makes sure that the domain elements in world
  $v$ are either freshly introduced, or appear in the predecessor.

  We argue next that $(T,\bagz,\bago)$ is canonical.  Let
  $v\in T$ be a successor of $w\in T$ and assume that $r=\bago(w)$ and
  $s=\bago(v)$. A general property of the construction of \Jmc is that
  Condition~\ref{it:can1} is satisfied throughout. More precisely, the
  application of a rule does not change the interpretation of elements
  that were already present, and it implies~\ref{it:can1} in the
  created interpretation.  For the remaining conditions, we
  distinguish cases which rule has been applied to obtain $v$ from
  $w$.
  \begin{itemize}

    \item[--] In case of \Rone, it is clear from the definition of~\Rone,
      that~\ref{it:can2} is satisfied.

    \item[--] If \Rtwo has been applied, it is clear that
      $\Delta_v\cap\Delta_w$ is the singleton $\{\delta_0\}$ and that
      $r\neq s$. By the premise of the rule, we know that $\delta_0\in
      F(w)$. By definition of~$\mn{Wit}_{\Imc,r}$, we know that there
      is an $r$-root cluster \abf in $\Imf(v)$ with $\delta_0\in\abf$.
      Finally, observe that~\Rtwo is applied only once to every~$d\in
      F(w)$, and $r'\neq r$. Thus, \Rtwo satisfies~\ref{it:can3}.

    \item[--] Suppose \Rthree has been applied to some $r$-cluster \abf in
      $\bagz(w)$ with $\abf\subseteq\roots(w)$ and a direct $r$-successor
      $\hat e$ of $\tau(\delta)$ in \Imc, for some $\delta\in\abf$
      such that $(\delta,\delta')\notin r^\Jmc$, for any $\delta'$
      with $\tau(\delta')=\hat e$. We show that \abf witnesses~\ref{it:can4}. 

      By definition of $\Delta'$ and~\eqref{eq:ij}, \abf is an $r$-root
      cluster in $\Imf(v)$. Items~\ref{it:can4}(a) and~\ref{it:can4}(b) are satisfied by assumption. 
      Item~\ref{it:can4}(c) follows from the definition of $\Delta_v$.
      For~\ref{it:can4}(d), assume $(d,e)\in r^(\Imf(v))$ and suppose that
      $e\in F(v)$. By definition of \Rthree, we know that $(d,e)\in
      r^\Jmc$. Because $e\in F(v)$, the tuple $(d,e)$ has been added
      to $r^\Jmc$ in this step via the application of~\eqref{eq:ij}.
      Thus, we obtain $d\in\abf\cup\Delta=\abf\cup F(v)$, as required.

   \end{itemize}

   \smallskip We next verify that $(T,\bagz,\bago)$ and
   $\Jmc$ satisfy Conditions~\textit{(i)}--\textit{(iii)} from the
   statement.

   Condition~\textit{(i)} is a consequence of the fact that $\tau$ is
   a homomorphism from \Jmc to \Imc and that PRPQs are preserved under
   homomorphisms. Condition~\textit{(ii)} is ensured by the
   initialization phase. For Condition~\textit{(iii)}, we start with
   the bounding the width. We distinguish
   cases according to which rule was applied. 
\begin{itemize}

  \item for the root $\varepsilon$ of $T$, we know that
    $|\Delta_\varepsilon|$ is bounded as required by construction and
    Lemmas~\ref{lem:outdegree-width} and~\ref{lem:wit}.

  \item If $w$ was created by~\ref{it:r1}, then $|\Delta_w|=2$.

  \item If $w$ was created by~\ref{it:r2}, then $|\Delta_w|$ is
    bounded as required by Lemma~\ref{lem:wit}.

   \item If $w$ was created by~\ref{it:r3}, let $\Cmc$ denote the set of
    all concepts $\qnrleq m r  D$ appearing in $\Tmc$. We make the
    following observations.
    \begin{enumerate}[label=(\alph*)]

      \item For all $(d,e)\in r^\Imc$, $C\in \Cmc$: if $d\in
	C^\Imc$, then $e\in C^\Imc$;
      
      \item Let $d_1,\ldots,d_n$ be such that $(d_i,d_{i+1})\in
	r^\Imc$, for all $1\leq i<n$ and $n>|\Tmc|$. Then
	$\mn{Wit}_{\Imc,r}(d_i)=\mn{Wit}_{\Imc,r}(d_j)$ for some
	$i\neq j$.

    \end{enumerate}
    Point~(a) follows from the semantic of $\qnrleq m r D$ and
    transitivity. For Point~(b) observe that $|\Cmc|<n$, thus
    there are $i\neq j$ such that $d_i\in C^\Imc$ iff $d_j\in C^\Imc$,
    for all $C\in \Cmc$. By definition of $\rightsquigarrow_{\Imc,r}$ and
    $\mn{Wit}_{\Imc,r}$, we also have
    $\mn{Wit}_{\Imc,r}(d_i)=\mn{Wit}_{\Imc,r}(d_j)$. 

    Now consider some branch of applications of  \ref{it:r3}. Each
    application adds (copies of) elements which are new witnesses,
    that is, they are in
    $\mn{Wit}_{\Imc,r}(e)\setminus \mn{Wit}_{\Imc,r}(d)$, for some
    $(d,e)\in r^\Imc$. By Point~(b), along such a branch, elements
    are added at most $|\Tmc|$ times. Each time, at most
    $|\mn{Wit}_{\Imc,r}(d)|$ elements are added. Overall, the size is
    bounded by
    $|\Tmc|\cdot(|\Amc|\cdot2^{p(|\Tmc|)})=O(|\Amc|\cdot
    2^{p(|\Tmc|)})$.

\end{itemize}
Finally, the outdegree is bounded by $k_1\cdot k_2\cdot k_3$, where
$k_1$ is the number of elements in a bag, $k_2$ is the number of role
names, and $k_3$ is the maximal outdegree in \Imc. We have seen bounds
for $k_1$ and $k_2$. So it remains to note that the outdegree in \Imc
is bounded by $|\Amc|+2^{\mn{poly}(|\Tmc|)}$, by
Lemma~\ref{lem:outdegree-width}.

\medskip It remains to prove that $\Jmc\models\Kmc$, which is a
consequence of the following claim.

\smallskip\noindent{\it Claim.}
For all $\delta \in \Delta^\Jmc$ and all $C \in \mn{cl}(\Tmc)$, we have
$$\delta \in C^\Jmc \text{ iff } \tau(\delta) \in C^\Imc.$$

\smallskip\noindent{\it Proof of the Claim.} The proof is by induction
on the structure of concepts. The case $C=A$  for $A \in \mn{N_C}$ follows from $\tau$ being
a homomorphism and rules \ref{it:r1}--\ref{it:r3}. The Boolean cases $C=\neg
D$ and $C=C_1\sqcap C_2$ are  consequences of the induction
hypothesis. It thus remains to consider concepts of the form $C=\qnrsim
n  r D$.  If $r \in \mn{N}_\mn{R}^{nt}$, the claim is  a
straightforward consequence of the induction hypothesis and
construction rule~\ref{it:r1}.  Now,  assume that $r \in \mn{N}_\mn{R}^{t}$.
It suffices to  show that: 
\begin{enumerate} \item[(a)] If
      $\tau(\delta) \not \in {\qnrleq n  r  D}^\Imc $, then  $\delta \not \in
    {\qnrleq n  r  D}^\Jmc$, and 
    
    \item[(b)] if $\tau(\delta) \in {\qnrleq
      n  r  D}^\Imc$, then $\delta \in {\qnrleq n  r  D}^\Jmc$.
  \end{enumerate}

\noindent {\bf For Point~(a)} 
Let now be $\tau(\delta)=d$ and 
$d\not \in {\qnrleq n r D}^\Imc$, that is $d\in {\qnrgeq {(n+1)} r
D}^\Imc$. It suffices to show that $\delta$ has $n+1$ $r$-successors
satisfying $D$. Let $d'\in\Delta^\Imc$ be any domain element such that $(d,d')\in
r^\Imc$ and $d'\in D^\Imc$.  Thus, either $d'\in Q_{\Imc,r}(d)$ or
there is a sequence $d_0,\ldots,d_m$ with $d_0=d$, $d_m=d'$, and
$d_{i+1}$ is a direct $r$-successor of $d_i$ in \Imc, for all $0\leq
i<m$. 

In the first case, by construction, there is $\delta'\in
Q_{\Jmc,r}(\delta)$ with $\tau(\delta')=d'$. By induction hypothesis, $\delta'\in
D^\Jmc$. 

In the second case, we  show that for every such sequence, every 
%First we show that for every 
$\delta\in\Delta^\Jmc$ with $\tau(\delta)=d_0$ and every $0\leq i\leq
m$ there is an $r$-path from $\delta$ to some $\delta_{i}$ with
$\tau(\delta_{i})=d_{i}$.
The base case $i=0$ is immediate. 
So suppose there is an $r$-path  $\pi=\delta_0, w_0,\ldots,w_{i-1},\delta_{i}$ %be an $r$-path 
from
$\delta_0$ to $\delta_{i}$ with $\tau(\delta_{i})=d_{i}$, for $i>0$. 
Let $w\in T$ be such that $w=w_{d_i}$, and 
$\abf\subseteq\roots(w)$ with the $r$-cluster in $\bagz(w)$ such that $d_i \in \abf$.
%$\delta_{i}\in\Delta^r$, then set $w=\varepsilon$ and $\abf$ the
%$r$-cluster containing $\delta_{i-1}$. Otherwise, by construction, there is some
%$w\neq \varepsilon$ such that $\bago(w)=r$ and $\delta_{i-1}$ belongs
%to some $r$-cluster $\abf\subseteq \roots(w)$. 
We distinguish two cases:
\begin{itemize}

  \item If there is some $\delta' \in \Delta^\Jmc$, such that $(\delta,\delta')\in r^\Jmc$, for some  with
    $\tau(\delta')=d_{i+1}$. Then, there is some $r$-path
    $\delta,v_0,\ldots,v_k, \delta'$ from $\delta$ to $\delta'$ in
    $(T,\bagz,\bago)$. Then $\pi,v_0,\ldots,v_k,\delta'$ is the required
    $r$-path.

  \item If $(\delta,\delta')\notin r^\Jmc$, for all $\delta'$ with
    $\tau(\delta')=d_{i+1}$,
    then \ref{it:r3} applies to $w$, $\abf$,
    and $d_{i+1}$. In particular, it adds a successor $v$ of $w$ to $T$
    and adds a domain element $\delta'=(d_{i+1})_v\in \Delta$ to $\Jmc$
    such that $(\delta_i,\delta')\in r^\Jmc$. By construction,
     $(\delta,\delta') \in r^{\bagz(v)}$ and $\pi,v,\delta'$ is the
    required $r$-path.
\end{itemize}

Thus, we can conclude that  there is an $r$-path from $\delta$ to some
$\delta'$ with $\tau(\delta')=d'$. By induction, we know that
$\delta'\in D^\Jmc$. Since distinct $d'$ with $(d,d')\in r^\Imc$ and
$d'\in D^\Imc$ yield distinct $\delta'$, this finishes the proof
of~(a).

\medskip \noindent{\bf For Point (b)} Assume $\tau(\delta) \in {\qnrleq n
 r  D}^\Imc$ with $\tau(\delta)=d$.  It
clearly suffices to show that for every $e\in\mn{Wit}_{\Imc,r}(d)$,
there is at most one $\delta'\in\Delta^\Jmc$ with $(\delta,\delta')\in
r^\Jmc$. To do so, let $w$ be the (unique) world where $\delta$ is
$r$-fresh in $(T,\bagz,\bago)$. We show first that 
\begin{itemize}

  \item[\textit{(x)}] for every $e\in\mn{Wit}_{\Imc,r}(d)$, there is
    precisely one $\delta'\in\Delta_w$ with $(\delta,\delta')\in
    r^{\bagz(w)}$ and $\tau(\delta')=e$.

  \item[\textit{(xx)}] for every $\delta' \in \Delta^\Jmc$ such that
    $(\delta, \delta') \in r^\Jmc$ and $\tau(\delta') \in
    \mn{Wit}_{\Imc,r}(\tau(\delta))$, we have $\delta' \in \Delta_{w}$

\end{itemize}
For showing~\textit{(x)}, observe that either $w = \varepsilon$ or
$\delta$ was added either by an application of \ref{it:r2} or
\ref{it:r3}. We distinguish cases:
\begin{itemize}

  \item Suppose first $w=\varepsilon$. If $\delta=a\in\mn{ind}(\Amc)$
    then~\textit{(x)} is clear due to the initialization of $\Jmc$ and
    $(T,\bagz,\bago)$.  If $\delta= d_r \in \Delta^r$, we have that
    $d\in \mn{Wit}_{\Imc,r }(a)$ for some individual $a$, and by
    Lemma~\ref{lem:wit2} thus $\mn{Wit}_{\Imc,r}(d) \subseteq
    \mn{Wit}_{\Imc,r }(a)$. Together with the initialization of $\Jmc$
    and of $(T,\bagz,\bago)$, this yields that $\Delta_\varepsilon$
    contains exactly one element $\delta_e$ such that
    $(\delta,\delta_e)\in r^{\bagz(\varepsilon)}$ and
    $\tau(\delta_e)=e$, for each $e \in \mn{Wit}_{\Imc,r }(d)$. 

  \item If $w$ was created by \ref{it:r2}, then  there is some $\hat
    \delta \in \Delta_w \cap \Delta_{w'}$  such that for every
    $\delta' \in \Delta_w$, with $\hat \delta \neq \delta'$,
    $\tau(\delta') \in \mn{Wit}_{\Imc,r}(\tau(\hat \delta))$. The
    claim~\textit{(x)} follows now by Lemma~\ref{lem:wit2} and the
    definition of~\ref{it:r2}.
  
  \item If $w$ was created by \ref{it:r3}, let $\abf \subseteq
    F_r(w\cdot -1)$ be the $r$-cluster in $\Delta_w$ that witnesses
    this.  Then there is some $\hat \delta \in \abf$  and $\delta' \in
    \Delta_w$ such that  $\tau(\delta')$ is a direct successor of
    $\tau(\hat \delta)$.  Further, by the choice of $w$, we have
    $\delta \in F(w)$.  Then, by the definition of \Rthree,
    $\tau(\delta) \in \mn{Wit}_{\Imc, r}(\tau(\delta')) \setminus
    \mn{Wit}_{\Imc, r}(\tau(\hat \delta))$.  By  Lemma~\ref{lem:wit2},
    the definition of \ref{it:r3} ensures that there is exactly one
    fresh element in $\Delta_w$ for every $e \in \mn{Wit}_{\Imc,
    r}(\tau(\delta)) \setminus  \mn{Wit}_{\Imc, r}(\tau(\hat
    \delta))$.  It remains to show that there is exactly one element
    in $\Delta_w$ for every  $e \in \mn{Wit}_{\Imc, r}(\hat \delta)
    \cap \mn{Wit}_{\Imc, r}(\delta')$.  Indeed, we can (inductively)
    assume that \textit{(x)} holds for $w'=w\cdot -1$ and
    $\hat\delta$, that is, for every such $e$,  there is some $d'' \in
    \Delta_{w'}$ such that $\tau(\delta'') = e$.  Moreover $(\hat
    \delta, \delta'') \in r^{\bagz(w')}$, and  by the definition of
    \ref{it:r3}, we have  $\delta'' \in \Delta_w$.  

\end{itemize}

For showing~\textit{(xx)}, assume $(\delta, \delta') \in r^\Jmc$. 
 By Lemma~\ref{lem:can-paths}, there is an 
 $r$-path $\pi=\delta_0,w_0,\ldots,w_{k-1},\delta_k$ from $\delta$ to
 $\delta'$, satisfying either \ref{it:p1} or \ref{it:p2}. 
%Then there is an $r$-path $\pi = \delta_0,w_),\dots, w_{k-1},\delta_k$,  with $k \geq 1$, 
%from $\delta$ to $\delta'$
We use the following auxiliary claims.
\begin{claim} \label{claim:step}
For every $v$, and every $\delta \in \Delta_w$. If $v$ was created by an application of rule \ref{it:r3} and 
$\delta \notin F_r(v)$ then $W\subseteq \Delta_{v\cdot -1}$ where $$W=
\{\delta' \in \Delta_{v} \mid (\delta,\delta') \in r^{\bagz(v)} \land
\tau(\delta') \in \mn{Wit}_{\Imc, r}(\tau(\delta))\}.$$
\end{claim}

\begin{proof}
Let $\abf$ in $\bagz(w)$ be the $r$-cluster used in the  application of \ref{it:r3}, $w = v \cdot-1$,  and  $\delta' \in W$.  
We know that $v$ satisfies~\ref{it:can4} and that $\abf$ witnesses this.
Thus, $\abf$ is an $r$-root cluster in $\bagz(w)$ such that 
$\abf \subseteq F_r(w)$.
%
%Let $\abf \subset \Delta_{w}$ be the root $r$-cluster witnessing this. We know that $\delta\notin F(w)$. 
%Then there is some $\hat \delta \in \abf$ and $\delta$
Since $(\delta, \delta')\in r^{\bagz(w)}$,  we have that $\delta \in \Delta_{v}$, which by~\ref{it:r3} means
that either $\delta \in \abf$, or  there is some $\hat \delta \in \abf$ such that $(\hat\delta, \delta) \in r^{\bagz(w)}$, and thus 
$\delta \not \in \abf \cup F(v)$. 

If $\delta \not \in \abf \cup F(v)$, then $(\delta,\delta') \in r^{\bagz(v)}$ implies that $\delta' \not \in F(v)$, 
by condition \ref{it:can4}(d), that is $\delta' \in \Delta_{w}$. Now, assume $\delta \in \abf$. % we have that $\tau(\delta), \tau(\delta'') \in r^\Imc$. 
In that case, since  $\tau(\delta') \in
\mn{Wit}_{\Imc,r}(\tau(\delta))$ we know that 
$\delta'$ is not one of the fresh elements of the form $f_v$ added by~\ref{it:r3}. 
Therefore,  $\delta' \not \in F(v)$ and $\delta' \in \Delta_w$.  
\end{proof}
\begin{claim}\label{claim:path}
Let $\pi = \delta_0, w_0,\dots, w_{k-1}, \delta_k$ be a downward
$r$-path with $\tau(\delta_k) \in
\mn{Wit}_{\Imc_r}(\tau(\delta_0))$.
Then, for every  $0\leq  i \leq  k-1 $, we have 
$\delta' \in \Delta_{w_i}$, and $(\delta_i, \delta') \in r^{\bagz(w_i)}$.
\end{claim}
\begin{proof}
We show this using an inductive argument. 
This holds by assumption for  $i = k-1$ and the definition of $r$-path, and since $(\delta_0, \delta_{k-1}) \in r^\Jmc$ and 
$\tau(\delta') \in \mn{Wit}_{\Imc, r}(\tau(\delta_0))$ implies $\tau(\delta') \in \mn{Wit}_{\Imc, r}(\tau(\delta_{k-1}))$.  
For the inductive step, 
assume that this holds for $0 <  i+1 \leq k-1$.  Then $w_{i+1}$ was created by an application of \ref{it:r3}. Since
$(\delta_i, \delta_{i+1}) \in r^{\bagz{w_i}}$, then $w_{i+1} \notin F_r(w_{i+1})$. Then, as $(\delta_0, \delta_{i+1}) \in r^\Jmc$ 
and $\tau(\delta') \in \mn{Wit}_{\Imc, r}(\tau(\delta_0))$ we can conclude $\tau(\delta') \in \mn{Wit}_{\Imc, r}(\tau(\delta_{k-1}))$.
Thus by Claim~\ref{claim:step}
$\delta' \in \delta_{i}$. %Further, by definition of tree decomposition $(\delta_i, \delta'') \in r^{\bagz(w_i)}$.  $\delta$, and since 
\end{proof}

We now do a final case distinction according to which
case~\ref{it:p1} or~\ref{it:p2} applies to $\pi$. 

\begin{itemize}
\item Assume that $\pi$ satisfies \ref{it:p1}. 
From Claim~\ref{claim:path} we can conclude that $\delta' \in \Delta_0$. By~\ref{it:p1}, we have that either $w_0 = w$, and then $\delta' \in \Delta_w$ as required; 
or $w = w \cdot -1$, and then it must be the case that $w_0$ was created by an application or rule~\ref{it:r3}. 
Thus the statement follows from Claim~\ref{claim:path}. 

\item Assume that $\pi$ satisfies \ref{it:p2}. The statement clearly holds if $k=0$. For $k > 1$, 
using  Claim~\ref{claim:path} we can conclude $\delta' \in \Delta_{w_1}$. 
We know from~\ref{it:p2} that $\delta_1 \in F_r(w_1 \cdot -1)$ and thus $\bago(w_1 \cdot -1) = r$. 
Since $(\delta_1, \delta_2)\in r^{\bagz(w_1)}$, by \ref{it:can1}, we also have that $\bago(w_1)=r$, which means 
that $w_1$ was created by rule~\ref{it:r3} and, by Claim~\ref{claim:step}, we get $\delta' \in w_1\cdot-1$. 

By~\ref{it:p2}, we know that $w_1\cdot -1$ is an ancestor of $w_0$. Let  
%From~\ref{it:p2} we also have that for every node 
$v_0, \dots, v_n=w_0$ be the path   
from $w_1 \cdot -1 = v_0 \cdot -1$ to $w_0$, for $0 \leq n $. From the construction of $\Jmc$, we know that every 
node $v_i$ was added  by an application of either~\ref{it:r2} or~\ref{it:r3}.  
We claim that every $v_i$ was created by \ref{it:r3}.
Indeed, we have by definition of tree decomposition (item~\ref{it:tree3}) that 
$\delta_1 \in \Delta_{v_i}$ for every $0\leq i \leq n$, since $\delta_1 \in \Delta_{w_0}$.  
For $n=0$, the claim follows since $\Delta_{v_0} \cap \Delta_{v_0 \cdot -1}= \delta_1$,  $\bago(v_0)= r$, 
and $(\delta_1,  \delta') \in r^{\bagz(v_0 \cdot -1)}$. This means that $v_0$ was not created by~\ref{it:r2}.
Now assume $v_n$ was created by \ref{it:r3} for $0 \leq n$.
The claim follows for $v_{n+1}$ since  $\Delta_{v_{n}} \cap \Delta_{v_{n+1}}= \delta_1$ and $\delta_1 \notin F(v_n)$ 
because $\delta_1 \in \Delta_{v_0 \cdot -1}$.  Hence, $v_{n+1}$ was not created by~\ref{it:r2}.

Finally, we show that $\delta' \in \Delta_{v_i}$ for every $0\leq i \leq n$. From 
$\delta_1 \in \Delta_{v_i \cdot -1}\cap \Delta_{v_i}$ 
and $(\delta, \delta') \in r^{\bagz(v_i \cdot -1)}$ we get  
$\delta' \in \Delta_{v_1}$. Further, as $v_i$ was introduced via an application of \ref{it:r3} to some $r$-cluster $\abf$ in $\bagz(w)$,  
we also have that $v_i$ satisfies~\ref{it:can4} and that $\abf$ witnesses this.
%Thus, $\abf$ is an $r$-root cluster in $\bagz(w)$ such that 
%$\abf \subseteq F_r(w)$. 
Then, $\delta_1 \notin F(v)$ implies that either $\delta_1 \in \abf$, or  there is some $\hat \delta \in \abf$ 
such that $(\hat\delta, \delta_1) \in r^{\bagz(v_i \cdot -1)}$. In either case, by~\ref{it:can4}(c), we can conclude that
$d' \in \Delta_{v_i}$. Therefore, the previous inductive argument together with the fact that 
$d_1 \in \Delta_{v_0 \cdot -1}$ and $(\delta_1,\delta') \in r^{\bagz{v_0\cdot -1}}$ imply that $\delta' \in \Delta_{v_i}$ for every $i$. 
This in particular implies $\delta' \in \Delta_{w_0}$. 
\end{itemize}

\end{proof}

\section{Proof of Theorem~\ref{thm:main}}

\noindent\textbf{Theorem~\ref{thm:main}.} \textit{
PRPQ entailment over \SQ-knowledge bases is \TwoExpTime-complete.}

\begin{proof}
  The lower bound is inherited from \TwoExpTime-hardness of positive
  existential query entailment in \ALC.

  For the upper bound, we first show correctness of the given procedure,
  that is, we show that $\Kmc\models\varphi$ iff
  $L(\Amf_{\mn{can}}\wedge\Amf_\Kmc\wedge\neg \Amf_{\vp})\neq
  \emptyset$.
  
  $(\Leftarrow)$ Assume that $\Kmc\not\models\varphi$. By
  Theorem~\ref{thm:sq-canmod}, we know that there is a model
  $\Jmc$ of $\Kmc$ and a canonical tree decomposition
  $(T,\bagz,\bago)$ of $\Jmc$
  satisfying Conditions~\textit{(i)}--\textit{(iii)}.

  The key observation is that, by the size $2K$ of $\Delta$, it is
  possible to select a mapping $\pi:\Delta^\Jmc\to \Delta$
  such that for each $w\in T\setminus\{\varepsilon\}$ and each $d\in
  \Delta_w\setminus\Delta_{w\cdot -1}$, we have $\pi(d)\notin
  \{\pi(e)\mid e\in \Delta_{w\cdot -1}\}$.  Define a $\Sigma$-labeled
  tree $(T,\tau)$ by setting, for all $w\in T$, $\Imf_w$ to the image
  of $\bagz(w)$ under $\pi$ and $r_w$ to $\bago(w)$. Clearly,
  $(T,\tau)$ is consistent, and $\Imc_{(T,\tau)}$ is isomorphic to
  $\Jmc$. It is not hard to see that $(T,\tau)\in L(\Amf_\mn{can})$.
  Now, by Lemma~\ref{lem:kba-correct} and $\Jmc\models\Kmc$, we have $(T,\tau)\in
  L(\Amf_\Kmc)$, and, by Lemma~\ref{lem:query-automaton} and $\Jmc\not\models\vp$, we
  have $(T,\tau)\notin L(\Amf_{\varphi})$. Thus
  $L(\Amf_{\mn{can}}\wedge \Amf_{\Kmc}\wedge\neg \Amf_{\varphi})$ is
  not empty. 

  For the direction $(\Rightarrow)$, let $(T,\tau)\in
  L(\Amf_{\mn{can}}\wedge \Amf_{\Kmc}\wedge\neg \Amf_{\varphi})$.
  Since $(T,\tau)\in L(\Amf_{\mn{can}})$, we know that $(T,\tau)$ is
  consistent, that the represented $(T,\bagz,\bago)$ is a canonical
  decomposition of $\Imc_{(T,\tau)}$. It remains to note that, by
  Lemmas~\ref{lem:kba-correct} and~\ref{lem:query-automaton}, we have
  that $\Imc_{(T,\tau)}\models\Kmc$ and
  $\Imc_{(T,\tau)}\not\models\vp$, respectively.
  
  The \TwoExpTime-upper bound follows now from the following facts.
  By Lemma~\ref{lem:kba-correct} and~\ref{lem:query-automaton}, the
  construction of the respective automata can be done in (worst case)
  double exponential time; moreover, the automata have exponentially
  many states. Since intersection and complement of 2ATAs can be done
  in polynomial time, we know that $\Amf_\mn{can}\wedge
  \Amf_{\Kmc}\wedge\neg \Amf_\vp$ has exponentially many states, and
  can be constructed in exponential time. It remains to note
  that emptiness of that automaton can be checked in double
  exponential time.
\end{proof}

\section{The Automaton $\Amf_\mn{can}$}\label{sec:autcan}

We refrain from giving the automaton explicitely, but rather describe
its functioning. Let $(T,\tau)$ be a $\Sigma$-labeled tree.  Consistency of $(T,\tau)$
can be checked by verifying that:
\begin{itemize}

  \item $r_w$ is a role name from \Kmc for every $w\neq
    \emptyset$, and $r_\varepsilon=\bot$, and 

  \item for every $w\in T$, every successor $v$ of $w$, and any two
    elements $d,e\in \Delta^{\Imc_w}\cap \Delta^{\Imc_v}$, we have
    that $d\in A^{\Imc_w}$ iff $e\in A^{\Imc_v}$, for all concept
    names $A$ appearing in \Kmc, and $(d,e)\in r^{\Imc_w}$ iff
    $(d,e)\in r^{\Imc_v}$, for all role names $r$ appearing in $\Kmc$.

\end{itemize}
Both can be easily done with a 2ATA.

For verifying that the encoded structure $\Imc_{(T,\tau)}$ is
canonical, we formulate the following
variants~\ref{it:can1'}--\ref{it:can4'} which talk about $(T,\tau)$.
We call $(T, \tau)$ \emph{canonical} iff for every $w \in T$ with
$\tau(w) = (\Imc_w, r)$ and every successor $v$ of $w$ with
$\tau(v)=(\Imc_v, s)$, the following conditions are satisfied:
\begin{enumerate}[label=(C$_\arabic*'$),leftmargin=*,align=left]

  \item\label{it:can1'} if $(d,e)\in s_1^{\Imc_v}$, then $s_1=s$, or
    $d=e$ and $s_1\!\in\!\mn{N}_\mn{R}^{t}$;

  \item\label{it:can2'} if $s\in \mn{N}_\mn{R}^{nt}$, then
    $\Delta^{\Imc_v}=\{d,e\}$, for some $d\in F(w)$, $e\in F(v)$, and
    $s^{\Imc_v}=\{(d,e)\}$;

  \item \label{it:can3'} if $s\in\mn{N}_\mn{R}^t$ and $r\notin\{\bot,
    s\}$, then there are $d\in F(w)$ and an $r$-root cluster \abf in
    $\Imc_v$ such that $\Delta^{\Imc_w} \cap \Delta^{\Imc_v}=\{d\}$
    and $d\in\abf$; moreover, there is no successor $v'\neq v$ of $w$
    satisfying this for $d$ and $r_v' = s$;
    
  \item \label{it:can4'}if $s\in\mn{N}_\mn{R}^{t}$ and
    $r\in\{\bot,s\}$, then there is an $s$-root cluster $\abf$ in
    $\Imc_v$ with: 

    \begin{enumerate}

      \item $\abf\subseteq F_{s}(w)$;

      \item \abf is an $s$-cluster in $\Imc_w$;

      \item for all $d\in \abf$ and $(d,e)\in s^{\Imc_w}$, we have
	$e\in\Delta^{\Imc_v}$; 

      \item for all $(d,e)\in s^{\Imc_v}$, $d\in \abf\cup F(v)$ or
	$e\notin F(v)$.

    \end{enumerate}

\end{enumerate}
It is not difficult to verify that $(T,\tau)$ is canonical iff the
represented extended tree decomposition $(T,\bagz,\bago)$ is
canonical.  Moreover, Conditions~\ref{it:can1'}--\ref{it:can4'} can be
implemented in a 2ATA in a straightforward way.

\section{Proof of Lemma~\ref{lem:can-paths-encoding}}

Note that Lemma~\ref{lem:can-paths-encoding} from the main part is
just the reformulation of Lemma~\ref{lem:can-paths} adapted to the
encoding. We refer the reader to Appendix~\ref{app:thm1} for a proof of
that Lemma.

\section{Knowledge Base Automaton $\Amf_\Kmc$}

%\nb{J: update lemma}

%\medskip\noindent\textbf{Lemma~\ref{lem:kb-automaton}.} \textit{
  %
\begin{lemma}
 %There are 2ATAs %$\Amf_1,\Amf_2,\Amf_3$ 
  %$\Amf_\Amc$ and $\Amf_\Tmc$ of size $O(|\Amc|\cdot
 % 2^{\mn{poly}(|\Tmc|})$ such that: %$(T,\tau)\in
  %L(\Amf_1)$ iff $(T,\tau')$ is a canonical decomposition (of some
  %interpretation); 
   %For every $(T,\tau)\in L(\Amf_{\mn{can}})$, 
    There is a 2ATA  ${\Amf_\Kmc}$ such that for every $(T,\tau)\in L(\Amf_{\mn{can}})$, we have that $(T,\tau)\in L(\Amf_\Kmc)$  iff
  $\Imc_{T,\tau}\!\models\!\Kmc$. %and $(T,\tau)\in L(\Amf_\Tmc)$ iff
  %$\Imc_{T,\tau}\!\models\!\Tmc$.
   It  can be constructed in time double exponential in $|\Kmc|$, and has
  exponentially many states in $|\Kmc|$.
  \end{lemma}
%}

\begin{proof}
  %
%  For devising $\Amf_1$, observe that the properties of canonical tree
%  compositions can easily be checked using a 2ATA; we leave the
%  details to the reader. 
 $\Amf_\Kmc$ is the intersection of two automata $\Amf_\Amc$
and $\Amf_\Tmc$ verifying that the input satisfies the ABox and the
TBox, respectively. For devising $\Amf_\Amc$, recall that
  $\mn{ind}(\Amc)\subseteq \Delta$ and we identify, in the translation to
  $(T,\tau)$, $[\varepsilon]_a$ with $a$, for each
  $a\in\mn{ind}(\Amc)$, to reflect the SNA. 
Moreover, recall that, by 
Theorem~\ref{thm:sq-canmod}, we can assume that the ABox is actually
satisfied in the root. Thus, $\Amf_\Amc$ only needs to check whether \Amc
is satisfied  at the root. $\Amf_\Amc$ is a single state automaton
$\Amf_\Amc=(\{q_0\},\Sigma,q_0,\delta,F)$ with 
\begin{align*}
  \delta(q_0,\bullet) & = \mn{false} & \\
  \delta(q_0,(\Imc,x))& = \text{if $\Imc\models\Amc$, then
    $\mn{true}$ else $\mn{false}$} 
\end{align*}
%  Checking whether
%  $\Imc_{T,\tau}\models\Amc$ then amounts to checking that $\alpha\in
%  M$ for each $\alpha\in \Amc$ and for $\tau(\varepsilon)=\langle
%  M,x\rangle$, which can easily be done using a 2ATA.

The next lemma concentrates on the construction of   $\Amf_\Tmc$.

\medskip \noindent {\bf Lemma~\ref{lem:kba-correct}} {\it
  For every $(T,\tau)\in L(\Amf_{\mn{can}})$, we have $(T,\tau)\in
  L(\Amf_\Tmc)$ iff $\Imc_{(T,\tau)}\models \Tmc$. It
  can be constructed in time double exponential in $|\Kmc|$, and has
  exponentially many states in $|\Kmc|$.
  }

\begin{proof}
For $\Amf_\Tmc$, we give the missing transitions and states (see Section~\ref{sec:KBA}).
First, we show how the announced ``navigation'' works. First, for
non-transitive roles, we use the following transitions to navigate the
automaton to the unique world where $d$ is fresh, for all states
$(\sim n\,r\,D)(d)$:
\begin{align*}
  \delta((\sim n\,r\,D)(d), (\Imc,x)) & = \mn{false}
  \hspace{2.6cm}\text{if $d\notin\Delta^\Imc$} \\
  \delta((\sim n\,r\,D)(d), (\Imc,x)) & = \big( (0,F_{d}) \wedge
  (0,q^*_{(\sim n\,r\,D),d})\big) \vee{} \\ &\quad\quad (-1,(\sim
  n\,r\,D)(d))\quad\text{if $d\in\Delta^\Imc$}\\
  \delta(F_d,(\Imc,x)) & = \begin{cases} \mn{true} & \text{if
    $x=\varepsilon$} \\ (-1,\overline F_d) & \text{otherwise}
  \end{cases}\\
  \delta(\overline F_d,(\Imc,x)) & = \begin{cases} \mn{true} &
    \text{if $d\notin\Delta^\Imc$} \\ \mn{false} & \text{otherwise}
  \end{cases}
\end{align*}
For transitive roles, we use the following transitions to navigate the
automaton to the unique world where $d$ is $r$-fresh, for all states
$(\sim n\,r\,D)(d)$:
\begin{align*}
  %
%   \delta((\sim n\,r\,D)(d), \bullet) & = \mn{false}
%   \hspace{2.6cm}\text{if $d\notin\Delta^\Imc$} \\
  %
   \delta((\sim n\,r\,D)(d), (\Imc,x)) & = \mn{false}
  \hspace{2.6cm}\text{if $d\notin\Delta^\Imc$} \\
  \delta((\sim n\,r\,D)(d), (\Imc,x)) & = \big( (0,F_{r,d}) \wedge
  (0,q^*_{(\sim n\,r\,D),d})\big) \vee{} \\ &\quad
  \bigvee_{i\in [k]}(i,(\sim
  n\,r\,D)(d))\quad\text{if $d\in\Delta^\Imc$},
  % 
%   \delta(F_d,(\Imc,x)) & = \begin{cases} \mn{true} & \text{if
%     $x=\varepsilon$} \\ (-1,\overline F_d) & \text{otherwise}
% %   
%   \end{cases}\\
%   %
%   \delta(\overline F_d,(\Imc,x)) & = \begin{cases} \mn{true} &
%     \text{if $d\notin\Delta^\Imc$} \\ \mn{false} & \text{otherwise}
% %   
%   \end{cases}
  %
\end{align*}
where the transitions for $F_{r,d}$ are given in the main part. 
% For this purpose, it locates first the (unique) $w$ such that $d\in
% F_r(w)$, and changes to state $q_{(\sim n\ r\ D),d}^*$:
% % 
% \begin{align*}
%   %
%   \delta((\sim n\,r\,D)(d), \bullet) & = \mn{false}, \\
%   %
%   \delta((\sim n\,r\,D)(d), (\Imc,x)) & = \mn{false}, \quad
%   \text{if $d\notin \Delta^\Imc$} \\[2mm]
%   %
%   \delta((\sim n\,r\,D)(d), (\Imc,x)) & = ((0,F_{r,d}) \wedge (0,q_{(\sim
%   n\,r\,D),d}^*))\vee{} \\ &\hspace{-.5cm} \textstyle\bigvee_{i\in[k]} (i,(\sim
%   n\,r\,D)(d)),\quad \text{if $d\in\Delta^\Imc$}\\[2mm]
%   %
% %   \delta(q_d,(\Imc,x)) & = \text{if $d\in\Delta^\Imc$, then \mn{true}
% %   else \mn{false}} \\
%   %
%   \delta(F_{r,d}, (\Imc,\bot)) & = \mn{true} \\[2mm] 
%   % 
%   \delta(F_{r,d}, (\Imc,x)) & = \mn{false} \quad\quad\quad \text{if
%     $x\notin \{r,\bot\}$} \\
%     %
%   \delta(F_{r,d}, (\Imc,r)) & = (-1, F'_{r,d}) \\[2mm]
%   % 
%   \delta(F'_{r,d}, (\Imc,x)) & = \text{if $x\notin \{r,\bot\}$, then
%   \mn{true} else \mn{false}}
%     %
% \end{align*}
%

\subsection*{Counting for non-transitive Roles}

The automaton implements the strategy suggested by
Lemma~\ref{lem:can-paths-encoding} via the following transitions. We
first concentrate on at-least restrictions, so let us fix a state
$q_{\qnrgeq{n}{r}{D},d}^*$. Recall that $k$ is the bound on the
outdegree, and let $N_{m,k}$ be the set of $m$-element subsets of
$[1,k]$. The automaton then has the following transitions for symbols
$(\Imc,s)\in\Sigma$ with $s\neq \bot$, that is, for non-root worlds:
\begin{align*}
  \delta(q_{\qnrgeq{n}{r}{D},d}^*,(\Imc,s)) & = \bigvee_{X\in N_{n,k}}
  \bigwedge_{i\in X} (i,q_{d,r,D}) \\
  \delta(q_{d,r,D},(\Imc,x)) & =\begin{cases} D(e) & \text{if $x=r$
    and $(d,e)\in r^\Imc$} \\ 
    \mn{false} &\text{ otherwise}  \end{cases}
\end{align*}
For symbols $(\Imc,\bot)\in\Sigma$, we have to additionally take
successors in \Imc into account, which is implemented as follows. Let
$S_{\Imc,r}(d)$ denote the set of all $e$ with $(d,e)\in r^\Imc$. We then
define the transition for $\delta(q_{\qnrgeq{n}{r}{D},d}^*,(\Imc,\bot))$ as 
\begin{align*}
  %
%   \delta(q_{(\geq n\,r\,D),d}^*,(\Imc,\bot)) & = 
  \bigvee_{S\subseteq
    S_{\Imc,r}(d)} \Big(\bigwedge_{e\in S} (0,D(e))\wedge \bigvee_{X\in
      N_{n-|S|,k}} \bigwedge_{i\in X} (i,q_{d,r,D})\Big)
\end{align*}
For states corresponding to at-most restrictions, $q_{\qnrleq{n}{r}{D},d}^*$, we include the complementary transitions, that is, for
$(\Imc,s)\in\Sigma$ with $s\neq \bot$:
\begin{align*}
  \delta(q_{\qnrleq{n}{r}{D},d}^*,(\Imc,s)) & = \bigwedge_{X\in N_{n,k}}
  \bigvee_{i\in X} (i,\overline q_{d,r,D}) \\
  \delta(\overline q_{d,r,D},(\Imc,x)) & =\begin{cases} ({\approx}D)(e) & \text{if $x=r$
    and $(d,e)\in r^\Imc$} \\ 
    \mn{true} &\text{ otherwise}  \end{cases}
\end{align*}
where ${\approx}D$ denotes the negation normal form of $\neg D$.
Moreover, we define the transition for $\delta(q_{\qnrleq{n}{r}{D},d}^*,(\Imc,\bot))$ as 
\begin{align*}
  %
%   \delta(q_{(\geq n\,r\,D),d}^*,(\Imc,\bot)) & = 
  \bigwedge_{S\subseteq
    S_{\Imc,r}(d)} \Big(\bigvee_{e\in S} (0,{\approx}D(e))\vee \bigwedge_{X\in
      N_{n-|S|,k}} \bigvee_{i\in X} (i,\overline q_{d,r,D})\Big)
\end{align*}

\subsection*{At-most Restrictions (Transitive Roles)}

Finally, the following are the transitions for the at-most
restrictions (for transitive roles). The strategy there is to try to
find $n+1$ $r$-successors satisfying $D$ and accept if this fails,
thus ``complementing'' the strategy for the at-least restrictions. Let
$N$ be the set of all tuples $\nbf=(n_1,\ldots,n_\ell)$ such that
$\sum_{i}n_i=n+1$. Then, $\delta(q^*_{\qnrleq{n}{r}{D},d},(\Imc,x))$
is defined as follows:
$$ \bigwedge_{\nbf\in N}\bigwedge_{X\subseteq [1,\ell]}\bigvee_{i\in
X} (0,q^{\ref{it:p1}}_{\qnrleq{n_i}{r}{D},a_i})\wedge\bigwedge_{i\in
[1,\ell]\setminus X} (0,q^{\ref{it:p2}}_{\qnrleq{n_i}{r}{D},a_i})$$
\begin{align*}
  \delta(q^{\ref{it:p1}}_{\qnrleq{n}{r}{D}, d}, (\Imc,x)) & =
  (0,\overline{F}_{r,d})\vee (0, q_{\qnrleq{n}{r}{D},d}^\downarrow)
  \\[2mm]
  \delta(q^{\ref{it:p2}}_{\qnrleq{n}{r}{D}, d}, (\Imc,x)) & =
  (0,F_{r,d})\vee (-1,q_{\qnrleq{n}{r}{D},d}^\uparrow) \\
\end{align*} 
Further, $\delta(q^\downarrow_{\qnrleq{n}{r}{D},d},(\Imc,x))$ is
defined as (where $\mbf$ and $M$ are as in Section~\ref{sec:KBA}): 
$$\bigwedge_{\mbf \in M} (0,p^\mn{loc'}_{n_0,r,D,d}) \vee \bigvee_{i=1}^k
(i,p^\mn{succ}_{\qnrleq{n_i}{r}{D},d}).$$
For states of the form $p^\mn{loc'}_{n,r,D,d}$, the transition
function is defined as
$$\delta(p^\mn{loc'}_{n,r,D,d},(\Imc,x))=\bigwedge_{Y\subseteq
  Q_{\Imc,r}(d), |Y|=n} \bigvee_{e\in Y} \approx D(e).$$
  For states of the form $p^\mn{succ}_{\qnrleq{n}{r}{D}}$, the transition
function on input $\bullet$ is defined as 
\begin{align*} \delta(p_{(\leq n\,r\,D),d},\bullet) & = \begin{cases}
    \mn{true} & \text{if $n>0$} \\ \mn{false} & \text{otherwise}
  \end{cases} \ \end{align*}
On inputs of the form $(\Imc,x)$, we set
$\delta(p^\mn{succ}_{\qnrleq{n}{r}{D},d},(\Imc,x))=\mn{true}$ if
$x\neq r$ or $d$ is not in a root cluster; otherwise, we define
\begin{align*}
  \delta(p^\mn{succ}_{\qnrleq{n}{r}{D},d},(\Imc,x)) = \bigwedge_{\nbf
  \in N} \bigvee_{i=1}^{\ell} (0, \overline q_{\qnrleq{n_i}{r}{D},a_i})
\end{align*}
Finally, we set $\delta(\overline q_{\qnrleq{n}{r}{D},d},
(\Imc,x))=\mn{true}$ in case $d \not \in \Delta^\Imc$, and otherwise
$$\delta(\overline q_{\qnrleq{n}{r}{D},d},
(\Imc,x))(0,q^{\ref{it:p1}}_{\qnrleq{n}{r}{D},d})\wedge
(0,q^{\ref{it:p2}}_{\qnrleq{n}{r}{D},d}).$$

\bigskip 

It remains to define the acceptance condition $F$. We set
$F=G_1,G_2,G_3$ where $G_1=\emptyset$, $G_2$ contains all states
of the form $q^*_{\qnrsim{n}{r}{D},d}$ and 
$p^\mn{loc}_{n,r,D,d}$ with $n\geq 1$, and $G_3=Q$. Note that, as mentioned in Section~\ref{sec:KBA}, 
 the parity condition enforces that states
$q^\downarrow_{(\geq n\,r\,D),d}$ with $n\geq 1$ are not suspended
forever, that is, eventualities are finally satisfied.
\end{proof}

Finally, it is not hard to see that  the number of states of the
automaton $\Amf_\Kmc$ (in particular that of $\Amf_\Tmc$) is bounded
exponentially in $|\Kmc|$. Moreover, $\Amf_\Tmc$ can be constructed in
double exponential time in $|\Kmc|$ since the size of the alphabet and
the number of states are exponentially bounded in $|\Kmc|$. 

Having Lemma~\ref{lem:can-paths-encoding} above at hand, it is routine
to show the correctness of the constructed automaton. Indeed,
$\Amf_\Tmc$ basically implements the `strategy' provided by this lemma. 
\end{proof}

%\newpage
\section{Query Automaton $\Amf_\vp$}

We first prove the characterisation lemma.

\medskip\noindent\textbf{Lemma~\ref{lem:witness-sequence}.} {\it
A function $\pi:\xbf\cup
I_\varphi\to \Delta^{\Imc_{(T,\tau)}}$ with $\pi(a)=[\varepsilon]_a$, for every
$a\in I_\varphi$, is a match for $\varphi$ in $\Imc_{(T,\tau)}$ iff
there is a $q_i$ such that for every
$\Bmf(t,t')$ in $\hat q_i$, there is a \emph{witness sequence}
% \begin{lemma} \label{lem:query-char}
%   %
%   Let $(T,\tau)$ be a consistent $\Sigma$-labeled tree and
%   $\varphi=\exists \xbf\, \psi(\xbf)$ a CRPQ. Then, a function
%   $\pi:\xbf\cup I_\varphi\to \Delta^{\Imc_{(T,\tau)}}$ with
%   $\pi(a)=[\varepsilon]_a$ for every $a\in I_\varphi$ is a match for
%   $q$ in $\Imc_{(T,\tau)}$ iff there is, for every $\Bmf(t,t')$ in
%   $\hat q$, a \emph{witness sequence}
  %
  \begin{align*}
    (d_0,s_0),w_1,(d_1,s_1),w_2,\ldots,w_n,(d_n,s_n),
  \end{align*}
  where $(d_i,s_i)\in \Delta \times S_\Bmf$ and $w_i\in T$ and such
  that:
  \begin{enumerate}[label=\textit{(\alph*)},leftmargin=*,align=left]

    \item $s_0=s_{0\Bmf}$, $s_n\in F_\Bmf$,

    \item $\pi(t)=[w_1]_{d_0}$, $\pi(t')=[w_n]_{d_n}$, and

    \item for every $i\in[1,n]$, we have $d_{i-1},d_{i}\in
      \Delta^{\Imc_{w_i}}$, $w_{i}\in [w_{i-1}]_{d_{i-1}}$ if $i\neq
      1$, and $\Imc_{w_i}\models \Bmf_{s_{i-1},s_i}(d_{i-1},d_i)$.

  \end{enumerate}
    %
%   Moreover, in the direction `only if', we can ensure that:
% % 
%   \begin{itemize}[leftmargin=*,align=left]
% 
%     \item[\textit{(d)}] for every $i,j\in[1,n]$ with $i+1<j$ and $w_i=w_j$, there is
%       some $i<m<j$ with $d_{i-1},d_m,d_j\in \Delta_{w_m}$ and
%       $w_m\in [w_i]_{d_{i-1}}\cap [w_i]_{d_j}$.
% 
%   \end{itemize} 
}

\begin{proof}
  $(\Rightarrow)$ Let $\pi$ be a match for $q$ in $\Imc_{(T,\tau)}$, and 
  let $\Bmf(t,t')\in \hat q$. We construct a sequence as required. 

  By definition of a match, we know that
  $\Imc_{(T,\tau)},\pi\models\Bmf(t,t')$, that is, there is a word $\nu_1
  \cdots \nu_n \in L(\Emc)$ and a sequence $[w_0]_{d_0}, \ldots,
  [w_n]_{d_n} \in
  \Delta^{\Imc_{(T,\tau)}}$ such that $[w_0]_{d_0} = \pi(t),
  [w_n]_{d_n} =\pi(t')$, and for all $i \in [1,n]$ we have that 
  \begin{enumerate}[label=\textit{(\roman*)},leftmargin=*,align=left]

    \item if $\nu_i = A?$, then $[w_{i-1}]_{d_{i-1}}=[w_i]_{d_i} \in
      A^{\Imc_{(T,\tau)}}$, 

    \item if $\nu_i = r$, then $([w_{i-1}]_{d_{i-1}}, [w_i]_{d_i})
      \in r^{\Imc_{(T,\tau)}}$, and

    \item if $\nu_i = r^-$, then $([w_{i}]_{d_{i}},
      [w_{i-1}]_{d_{i-1}}) \in r^{\Imc_{(T,\tau)}}$.

  \end{enumerate}
  Observe now that the replacement of $r$ and $r^-$ by $r\cdot r^*$ and
  $r^-\cdot (r^-)^*$, respectively, for all transitive roles together
  with the definition of encoding implies that we can assume
  without loss of generality that for all $i\in[1,n]$, we have:
%   there is an additional sequence $w_1,\ldots,w_n$ such that for all $i\in[1,n]$ we have
  %
  \begin{enumerate}[label=\textit{(\roman*)'},leftmargin=*,align=left]

    \item if $\nu_i=A?$, then $d_{i-1}=d_i$, and $d_i\in A^{\Imc_{w_i}}$, 

    \item if $\nu_i=r$, then $(d_{i-1},d_i)\in r^{\Imc_{w_i}}$, and

    \item if $\nu_i=r^-$, then $(d_{i},d_{i-1})\in r^{\Imc_{w_i}}$.

  \end{enumerate}
  Moreover, there is a sequence of states $s_0,\ldots,s_n\in Q_\Bmf$
  such that $s_0=s_{0\Bmf}$, $s_n\in F_\Bmf$ and
  $(s_i,\nu_i,s_{i+1})\in\Delta_\Bmf$, for all $i\in [0,n-1]$.  Thus,
  the sequence
  $(d_0,s_0),w_1,\ldots,(d_{n-1},s_{n-1}),w_n,(d_n,s_n)$ satisfies
  Items~\textit{(a)}--\textit{(c)} of the Lemma. 
  
%   To see how to establish also~\textit{(d)}, assume $i,j\in[1,n]$ with
%   $i+1<j$ and $w_i=w_j$.
% 
%   \smallskip\noindent\textit{Claim.} There is some $i<m<j$ and
%   $w\in [w_i]_{d_{i-1}}\cap [w_i]_{d_j}$ with $d_{i-1},d_m,d_j\in
%   \Delta^{\Imc_w}$.
% 
%   \smallskip\noindent\textit{Proof of the Claim.} For every
%   $k\in[0,n]$, define the set $$W_k=\{w\in
%     [w_i]_{d_{i-1}}\cap [w_i]_{d_j}
%     \mid d_k\in \Delta^{\Imc_w}\}.$$ Define
%     $W=W_{i-1}\cap W_j$, $V=\bigcup_{i<k<j} W_k$,
%     $W_{i-1}'=W_{i-1}\setminus W_j$, and $W_j'=W_j\setminus W_i$.  
%     Towards showing a contradiction assume that there are no such
%     $m,w$ as in the claim, thus $V$ and $W$ are disjoint.
%     Since additionally $V$ is connected and $W_{i-1}'$ and
%     $W_j'$ are disjoint, we obtain that $V$ is disjoint from $W_{i-1}'$
%     or from $W_j'$ and thus $V$ is disjoint from $W_{i-1}$ or $W_j$.
%     However,  applying Item~\textit{(c)} to $i$ and $j$ yields
%     $W_{i-1}\cap V\neq \emptyset$ and $W_j\cap V\neq \emptyset$, a
%     contradiction. This finishes the proof of the Claim.
%   
%   \smallskip It remains to observe that the sequence
%   %
%   \begin{align*}
%     %
%     & (d_0,s_0),w_1,(d_1,s_1),\ldots,  w_m,(d_m,s_m),\\ &
%     w,(d_m,s_m),\\ & w_{m+1},(d_{m+1},s_{m+1}),\ldots,w_n,(d_n,s_n)
%     %
%   \end{align*}
% %  
%   still satisfies~\textit{(a)}-\textit{(c)} and
%   additionally~\textit{(d)} for $i,j$. 
  
  \smallskip $(\Leftarrow)$ Assume that the sequences exist for every
  $\Bmf(t,t')\in \hat q$. We show that $\pi$ is a match. Let
  $(d_0,s_0),w_1,\ldots,w_n,(d_n,s_n)$ be the
  sequence for some $\Bmf(t,t')\in\hat q$. By Item~\textit{(c)}, we obtain 
  \begin{align*}
    \Imf(w_1) \models \Bmf_{s_0,s_1}(d_0,d_1), \ldots, \Imf(w_n)
    \models \Bmf_{s_{n-1},s_{n}}(d_{n-1},d_n).
  \end{align*}
  Since $(T,\Imf,\bago)$ is a tree decomposition of $\Imc$, we also
  have
  $$\Imc\models \Bmf_{s_0,s_1}(d_0,d_1),
  \ldots,\Imc\models\Bmf_{s_{n-1},s_n}(d_{n-1},d_n).$$
  This implies $\Imc\models \Bmf_{s_0,s_n}(d_0,d_n)$ and, by
  Item~\textit{(a)}, $\Imc\models \Bmf(d_0,d_n)$. Finally, using
  Item~\textit{(b)}, we obtain $\Imc,\pi\models\Bmf(t,t')$.
\end{proof}

The following lemma provides a crucial observation underlying the
design (and correctness) of the automaton.

\begin{lemma} \label{lem:correct-qa}
  Let $(d_0,s_0),w_1,\ldots,w_n,(d_n,s_n)$ be a witness sequence
  satisfying~(a)--(c) from Lemma~\ref{lem:witness-sequence}, and $i<j$. If
  $[w_i]_{d_i}\cap [w_{j}]_{d_j}\neq \emptyset$, then either $j=i+1$
  or there is an $i<m<j$ such that $[w_m]_{d_m}\cap [w_i]_{d_i}\cap
  [w_{j}]_{d_j}\neq \emptyset$.
\end{lemma}
\begin{proof}
  If $j=i+1$, we are done. So assume that $j>i+1$, and define sets
  $W_{i}=[w_i]_{d_i}$ and $W_j=[w_j]_{d_j}$, and  
  $$V=\bigcup_{i<k<j}[w_k]_{d_k}.$$
  Note that each of $W_i$ and $W_j$ are connected subsets of $T$. By
  Condition~\textit{(c)}, also $V$ is connected.  Moreover, by
  assumption $W_i\cap W_j\neq \emptyset$, and, again by
  Condition~\textit{(c)}, both $W_i\cap V\neq\emptyset$ and $W_j\cap
  V\neq \emptyset$. Since they are subsets of a tree, their joint
  intersection $W_i\cap W_j\cap V$ cannot be empty. Hence, 
  there is an $m$ as required. 
\end{proof}

\noindent\textbf{Lemma~\ref{lem:query-automaton}.}\textit{
  There is 2ATA $\Amf_\vp$ such that for every
  $(T,\tau)\in L(\Amf_{\mn{can}})$, we have $(T,\tau)\in
  L(\Amf_\vp)$ iff $\Imc_{(T,\tau)}\models \vp$. 
  It can be constructed in exponential time in $|\vp|+|\Kmc|$ and has
  exponentially in $|\vp|+|\Kmc|$ many states.
}

\medskip For the construction of the automaton $\Amf_\varphi$, recall
that $k$ is the outdegree underlying the input trees $(T,\tau)$ and
that $\Delta$ is the finite domain (of size $2K$). We construct the
2ATA $\Amf_\vp=(Q,\Sigma,q_0,\delta,F)$ as follows. Recall that 
$\vp$ can be equivalently rewritten into a disjunction
$q_1\vee\ldots\vee q_m$ of CRPQs.
Slightly abusing notation, we sometimes treat the $q_i$ as sets of
atoms.

\smallskip

\textbf{Set of states.} States in $Q$ take four forms. The basic
states are $q_0$, and for each $1\leq i\leq m$, the CRPQ $\hat q_i$.
States of the third form are all tuples $\langle p,V_l,V_r\rangle$
such that there is a $i$ such that
\begin{itemize}

  \item[--] $p\subseteq \hat q_i$ and $I_p=\emptyset$,

  \item[--] $V_l$ is a set of expressions $(d,s)\to_\Bmf x$ such
    that $\Bmf$ is the automaton of some atom $\Bmf(t,t')$ in
    $\hat q_i$, $s\in Q_{\Bmf}$, $d\in\Delta$, 
    $x\in \mn{var}(p)$, and for each $\Bmf(t,t')$ in $\vp$,
    there is at most one such expression,

  \item[--] $V_r$ is a set of expressions $x\to_\Bmf (d,s)$ such that
    $\Bmf$ is the automaton of some atom $\Bmf(t,t')$ in $\hat q_i$,
    $s\in Q_{\Bmf}$, $d\in\Delta$, $x\in \mn{var}(p)$, and for each
    $\Bmf(t,t')$ in $\vp$, there is at most one such expression.

\end{itemize}

States of the fourth form are tuples $\langle d,s,\Bmf,d',s'\rangle$
with $d,d'\in \Delta$, $s,s'\in Q_\Bmf$, and $\Bmf$ in $\vp$.

Intuitively, a state of form $\langle p,V_l,V_r\rangle$ expresses
the following obligations: 
\begin{itemize}

  \item[--] each atom $\Bmf(t,t')$ in $p$ still has to be `processed',

  \item[--] each $(d,s)\to_\Bmf x$ means that we need to find a path
    from $s$ to a final state in $\Bmf$ which is also a path from
    $d$ to the image of variable $x$,

  \item[--] each $x\to_\Bmf(d,s)$ means that we need to find a
    path from $q_{0\Bmf}$ to $s$ in $\Bmf$ which is
    also a path from the image of variable $x$ to $d$.

\end{itemize}
A state of the second form $\langle d,s,\Bmf,d',s'\rangle$ expresses the
obligation that there is a path along which we can reach both $d'$ from
$d$ in $\Imc_{(T,\tau)}$ and $s'$ from $s$ in $\Bmf$.

\smallskip \textbf{Transition function.} As a general proviso, we set
$\delta(q,\bullet)=\mn{false}$, for all states $q\in Q$; in what
follows, we define the transitions only for symbols of the form
$\sigma=(\Imc,x)\in\Sigma$. As the transition function does not depend
on $x$, we generally write only $\Imc$.

The automaton starts off in state $q_0$ by choosing
non-deterministically a disjunct $q_i$:
$$\delta(q_0,\Imc)=\bigvee_{1\leq i\leq m} (0,\hat q_i).$$
For every state $\hat q_i$, we define a transition as follows.
Let $\Theta(\hat q_i)$ be the set of all triples $(Q_0,V_l,V_r)$ which can be the result of
the following procedure:
\begin{enumerate}

  \item initialize $V_l=V_r=Q_0:=\emptyset$;

  \item for each $\Bmf(a,b)\in \hat q_i$, choose some
    $s_f\in F_\Bmf$ and add $\langle a,s_{\Bmf
    0},\Bmf,b,s_f\rangle\in Q_0$;

  \item for all $\Bmf(x,a)\in \hat q_i$, choose some $s_f\in F_\Bmf$ and
    add $x\to_\Bmf (a,s_f)\in V_r$;

  \item for all $\Bmf(a,x)\in \hat q_i$,
    add $(a,s_{0\Bmf})\to_\Bmf x\in V_l$.

\end{enumerate}
Intuitively, we choose an accepting state in $F_\Bmf$ for every
occurrence of an individual name as the second argument in some atom
$\Bmf(t,t')$.  Moreover, obtain $p$ from $\hat q_i$ by dropping all
atoms mentioning an individual name. The transition for $\hat q_i$
is then
\begin{align*}
  \delta(\hat q_i,\Imc) = \bigvee_{(Q_0,V_l,V_r)\in \Theta(\hat
  q_i)}\big(
  (0,\langle p,V_l,V_r\rangle) \wedge \bigwedge_{q\in Q_0}
  (0,q)\big).
\end{align*}

For transitions for states of the form $\langle
d,s,\Bmf,d',s'\rangle$, we take inspiration from
Lemma~\ref{lem:correct-qa}. We start with setting $\delta(\langle
d,s,\Bmf,d',s'\rangle,\Imc)=\mn{false}$ whenever
$\{d,d'\}\not\subseteq \Delta^\Imc$, and assume from now on that
$\{d,d'\}\subseteq \Delta^\Imc$. The base case is the following: 
\begin{align*}
  \delta(\langle d,s,\Bmf,d',s'\rangle,\Imc)) = \mn{true}\quad
  \text{if $\Imc\models \Bmf_{s,s'}(d,d')$. }
\end{align*}
For the case when $\Imc\not\models \Bmf_{s,s'}(d,d')$, we include the
following transitions:
\begin{align*}
  & \delta(\langle d,s,\Bmf,d',s'\rangle,\Imc) = \bigvee_{i\in[k]}
  (i,\langle d,s,\Bmf,d',s'\rangle) \vee{} \\
  & \ \ \bigvee_{\substack{d''\in \Delta^\Imc, \\ s''\in Q_\Bmf}}
  \big( (0,\langle d,s,\Bmf,d'',s''\rangle)\wedge (0,\langle
  d'',s'',\Bmf, d',s'\rangle)\big)
  \end{align*}
  Intuitively, the automaton looks for a node to continue (first line)
  and then intersects the path non-deterministically (second line,
  c.f., Lemma~\ref{lem:correct-qa}); it
  is successful if it finds a node where the required path exists
  inside the associated interpretation (base case). 

  For states of the form $\langle p,V_l,V_r\rangle$, we start with
  including the transitions $\delta(\langle p,
  V_l,V_r\rangle,\bullet)=\mn{false}$ whenever $p=V_l=V_r=\emptyset$, and
  \begin{align*}
    \delta(\langle p,V_l,V_r\rangle,\Imc) = \mn{false} 
  \end{align*}
  whenever there is a $(d,s)\to_\Amc x\in V_l$
  or a $x\to_\Emc (d,s)\in V_r$ with
  $d\notin\Delta^\Imc$. 
  
  So assume now that $\langle
  p,V_l,V_r\rangle$ and $\Imc$ are compatible in this sense, and let 
  $S=\mn{var}(p)\cup \{x\mid (d,s)\to_\Bmf x\in V_l\}\cup \{x\mid
  x\to_\Bmf (d,s)\in V_r\}$. We denote 
  with $\Pmc(S,k)$ the set of all partitions of $S$ into $k+1$ pairwise
  disjoint, possibly empty sets $S_0,\ldots,S_k$. For each $\Sbf=(S_0,\ldots,S_k)\in \Pmc(S,k)$, define
  $\Theta(p,V_l,V_r,\Sbf)$ as the set of all tuples
  $(Q^0,p^1,V_l^1,V_r^1,\ldots,p^k,V_l^k,V_r^k)$ that
  can be obtained as the result of the following procedure.
  \begin{enumerate}

    \item for every $x\in S_0$, choose a value $d_x\in \Delta^\Imc$;

    \item for every atom $\Bmf(x,y)\in p$ with $\{x,y\}\subseteq S$
      choose $s_f\in F_\Emc$ and add $\langle d_x, s_{0\Bmf}, \Bmf,
      d_y, s_f\rangle$ to $Q^0$;

    \item for every atom $\Bmf(x,y)\in p$ with $x\in S_0$, $y\in S_i$
      for $i>0$, choose a value $d_{\Bmf y}\in \Delta^\Imc$ and a state
      $s_{y}\in Q_\Bmf$, and add $\langle
      d_x,s_{0\Bmf},\Bmf,d_{\Bmf y},s_y\rangle$ to $Q^0$ and $(d_{\Bmf
      y},s_y)\to_\Bmf y$ to $V_l^i$;

   \item for every atom $\Bmf(x,y)\in p$ with $y\in S_0$, $x\in S_i$
     for $i>0$, choose a value $d_{\Bmf x}\in \Delta^\Imc$ and states
     $s_{x}\in Q_\Bmf$, $s_f\in F_\Bmf$, and add $\langle d_{\Bmf
     x},s_x,\Bmf,d_y,s_f\rangle\in Q^0$ and
     $x\to_\Bmf (d_{\Bmf x},s_x)\in V_r^i$;

   \item for every atom $\Bmf(x,y)\in p$ with $x,y\in S_i$ for
     $i>0$, add $\Bmf(x,y)$ to $p^i$;

   \item for every atom $\Bmf(x,y)\in p$ with $x\in S_i,y\in S_j$ for
     $i\neq j$ and $i,j>0$, choose values $d_{\Bmf x},d_{\Emc y}\in \Delta^\Imc$ and
     states $s_{x},s_y$, and add $\langle d_{\Bmf
     x},s_x,\Bmf,d_{\Bmf y},s_y\rangle\in Q^0$, $(d_{\Bmf
      y},s_y)\to_\Bmf y\in V_l^j$, and 
     $x\to_\Bmf (d_{\Bmf x},s_x)\in V_r^i$;

   \item for every $(d,s)\to_\Bmf x\in V_l$:

     \begin{itemize}

       \item[--] if $x\in S_0$, choose some $s_f\in F_\Amf$ and add $\langle
	 d,s,\Amf,d_x,s_f\rangle\in Q^0$;

       \item[--] if $x\in S_i$ for $i>0$, then choose
	 $d'\in\Delta^\Imc$ and $s'\in Q_\Bmf$, and add
	 $\langle d,s,\Bmf,d',s'\rangle\in Q^0$ and
	 $(d',s')\to_\Bmf x\in V_l^i$;

     \end{itemize}

   \item for every $x\to_\Bmf (d,s)\in V_r$:

     \begin{itemize}

       \item[--] if $x\in S_0$, add $\langle
	 d_x,s_{0\Bmf},\Bmf,d,s\rangle\in Q^0$;

       \item[--] if $x\in S_i$ for $i>0$, then choose
	 $d'\in\Delta^\Imc$ and $s'\in Q_\Bmf$, and add
	 $\langle d',s',\Bmf,d,s\rangle\in Q^0$ and
	 $x\to_\Bmf(d',s')\in V_r^i$.

     \end{itemize}

 \end{enumerate}
 We then include the following transition: $\delta(\langle
 p,V_l,V_r\rangle,\Imc)$ as the following expression:
 \begin{align*}
    %
%     \delta(\langle p, V_l, V_r\rangle,\sigma) & =
   \delta(\langle p,V_l,V_r\rangle,\Imc) = \bigvee_{\substack{\Sbf\in \Pmc(S,k),\\
   (Q^0,p^1,V_l^1,V_r^1,\ldots,p^k,V_l^k,V_r^k)\in\Theta(p,V_l,V_r,\Sbf)}}\delta^*
%     \bigwedge_{q\in Q^0} (0,q) \wedge \bigwedge_{i=1}^k (i,\langle
%     p_i,V_l^i,V_r^i\rangle),
    %
  \end{align*}
  where $\delta^*$ abbreviates
  \begin{align*}
    \bigwedge_{q\in Q^0} (0,q) \wedge \bigwedge_{i=1}^k (i,\langle
    p^i,V_l^i,V_r^i\rangle).
  \end{align*}
  Finally, we define the parity acceptance condition as $F=Q$ to
  enforce that no state appears infinitely often.

  It should be clear that the number of states of the automaton is
  bounded by an exponential in $|\vp|$ and polynomially in $\Delta$,
  that is, exponentially in $|\Kmc|$. Moreover, it is easy to verify
  that $\delta$ can also be computed in exponential time.  To finish
  the proof of Lemma~\ref{lem:query-automaton}, it remains to show
  correctness of the constructed automaton.

\begin{lemma} \label{lem:query-correct}
  For every $(T,\tau)\in L(\Amf_{\mn{can}})$, we have that $(T,\tau)\in
  L(\Amf_\vp)$ iff $\Imc_{(T,\tau)}\models \vp$.  
\end{lemma}

\begin{proof}
  $(\Rightarrow)$
  Assume some accepting run of $\Amf_\vp$ on $(T,\tau)$. Let
  $\hat q_i$ be the successor state of $q_0$ in the accepting run.
  Moreover, define a mapping $\pi$ by taking: 
  \begin{itemize}

    \item[--] $\pi(a)=[\varepsilon]_{a}$ for all $a\in
      \mn{ind}(q_i)$;

    \item[--] $\pi(x)=[w]_{d}$, if the automaton visits $w$ in
      some state $\langle p, V_l,V_r\rangle$ and selects $S_0$ with $x\in
      S_0$ and $d_x=d$.

  \end{itemize}
  Note that $w$ and $d$ are uniquely defined by the construction of
  $\Amf_\vp$. In particular, the definition of the transitions for
  states of the form $\langle p, V_l, V_r \rangle$ makes sure that
  each variable $x$ is instantiated precisely once, and thus in a
  unique world $w_x$ to a unique value $d_x$. We show how to read off
  from the accepting run witnessing sequences for every $\Bmf(t,t')\in
  \hat q_i$. By Lemma~\ref{lem:witness-sequence}, this implies that $\pi$ is
  a match for $q_i$ (and thus for $\varphi$) in $\Imc_{(T,\tau)}$.

  Throughout the construction we maintain the following invariant: 
  \begin{itemize}

    \item[$(\ast)$] if $(d_0,s_0),w_1,\ldots,(d_{n},s_{n})$ is the
      currently constructed sequence for $\Bmf(t,t')$, then it
      satisfies~\textit{(a)} and~\textit{(b)}. Moreover, the automaton
      visits $w_i$ in state $(d_{i-1},s_{i-1},\Bmf,d_i,s_i)$, for all
      $i\in[1,n]$.

  \end{itemize}
  Fix some $\Bmf(t,t')\in \hat q_i$. We first distinguish cases on
  whether or not $t,t'$ are constant names.
  \begin{itemize}

    \item[--] If both $t,t'$ are constant symbols, then the automaton
      visits $\varepsilon$ in state $\langle t,s_{0\Bmf},\Bmf,t',s_f\rangle$ for some
      $s_f\in F_\Bmf$. We initialize the sequence for $\Bmf(t,t')$ with
      $(t,s_{0\Bmf}),\varepsilon,(t',s_f)$. Obviously, $(\ast)$ is
      satisfied. 

    \item[--] If $t$ is a constant name and $t'$ is not, then by construction of $\Amf_\vp$, in
      particular the treatment of $V_l$ in the definition of
      $\Theta$, there is a sequence
      $(d_0,s_0),w_1,(d_1,s_1),\ldots,(d_n,s_n)$ such that $d_0=t,
      s_0=s_{0\Bmf}$, $d_n=d_x$, $w_n=w_x$, $s_n\in F_\Bmf$, and which
      additionally satisfies~$(\ast)$.

    \item[--] The case that $t'$ is a constant name and $t$ is not is
      analogous (using $V_r$).

  \end{itemize}
  Now, take an atom $\Bmf(x,y)\in \hat q_i$ with $x,y\in
  \mn{var}(\hat q_i)$. By definition of $\Amf_\vp$, there is a unique world
  $w$ and a state $\langle p,V_l,V_r\rangle$ with $\Bmf(x,y)$ such that the automaton
  visits $w$ in state $\langle p,V_l,V_r\rangle$ and selects
  $\Smf$ with $x\in S_i$ and $y\in S_j$ for $i\neq j$. Thus, one of
  the cases 3, 4, or 6 applies.  We distinguish cases:
  \begin{itemize}

    \item[--] In case of~3, by the treatment
      of $V_l$, there is a sequence
      $(d_0,s_0),w_1,(d_1,s_1),\ldots,(d_n,s_n)$ such that $d_0=t,
      s_0=s_{0\Bmf}$, $d_n=d_x$, $w_n=w_x$, $s_n\in F_\Bmf$, and which
      additionally satisfies~$(\ast)$.

    \item[--] The case of~4 is analogous.

    \item[--] In case of~6, we know that $x\in S_i$, $y\in S_j$ for
      some $i,j>0$ and $i\neq j$. By construction, there are $d_{\Bmf x},d_{\Bmf y}\in
      \Delta^{\Imc_w}$ and states $s_{x},s_y$ and:
      \begin{enumerate}

	\item by the treatment of $V_r$ in $\Theta$: a sequence
	  $(d_0,s_0),w_1,\ldots,w_n,(d_n,s_n)$ with $d_0=d_x$,
	  $s_0=s_{0\Bmf}$, $w_1=w_x$, $d_n=d_{\Bmf x}$,
	  $s_n=s_x$, and such that, for each $i$, the automaton visits
	  $w_i$ in state $\langle
	  d_{i-1},s_{i-1},\Bmf,d_i,s_i\rangle$;

	\item by the treatment of $V_l$ in $\Theta$: a sequence
	  $(d_0',s_0'),w_1',\ldots,w_m',(d_m',s_m')$ with
	  $d_0'=d_{\Bmf y}$, $s_0'=s_y$, $d_m'=d_y$, and $s_m'\in
	  F_{\Bmf}$, and and such that, for each $i$, the automaton
	  visits $w_i'$ in state $\langle
	  d_{i-1}',s_{i-1}',\Bmf,d_i',s_i'\rangle$.

      \end{enumerate}
      We then start with the sequence
      \begin{align*}
	(d_0,s_0),\ldots, (d_n,s_n), w,
	(d_0',s_0'),\ldots,(d_m',s_m').
      \end{align*}
      This sequence satisfies~$(\ast)$ because of~1.\ and~2.\ above and
      because the automaton visits $w$ in state $\langle
      d_n,s_n,\Bmf,d_{0}',s_0'\rangle=\langle d_{\Bmf
      x},s_x,\Bmf,d_{\Bmf y},s_y\rangle$.

  \end{itemize}
  Thus, for each $\Bmf(t,t')\in \hat q_i$, we have constructed a sequence 
  satisfying~$(\ast)$. Next, we refine these sequences such that they
  also satisfy~\textit{(c)}. Let
  $(d_{i-1},s_{i-1}),w_i,(d_i,s_i)$ be an infix of the sequence
  constructed so far for some $\Bmf(t,t')\in\hat q_i$. By~$(\ast)$, we know
  that the automaton visits $w_i$ in $\langle
  d_{i-1},s_{i-1},\Bmf,d_i,s_i\rangle$. We distinguish cases:  
  \begin{itemize}

    \item[--] If the automaton accepts at this point, the sequence
      satisfies~\textit{(c)} for this $i$, and we are done.

    \item[--] If the automaton moves to some neighbor $w_i\cdot j$ with
      $j\in[k]$, then we replace $w_i$ in the sequence by $w_i\cdot
      j$. Obviously, invariant~$(\ast)$ is preserved.

    \item[--] If the automaton applies the intersection transition to
      $d''\in \Delta^{\Imc_{w_i}}$ and $s''\in Q_\Bmf$, then replace
      $(d_{i-1},s_{i-1}),w_i,(d_i,s_i)$ with 
      $(d_{i-1},s_{i-1}),w_i,(d'',s''),w_i(d_i,s_i)$. Obviously, the
      invariant~$(\ast)$ remains preserved. 

  \end{itemize}
  Because of the acceptance condition, the latter two cases apply only
  finitely often, so the process terminates with a sequence that
  satisfies~\textit{(c)} for all $i$. By
  Lemma~\ref{lem:witness-sequence}, we know that
  $\Imc_{(T,\tau)},\pi\models \Bmf(t,t')$.

  \smallskip $(\Leftarrow)$ As $\Imc_{(T,\tau)}\models \vp$, there is a match	
  $\pi$ for $\vp$ in \Imc. Thus, there is some $i$ such that
  $\Imc_{(T,\tau)},\pi\models \Bmf(t,t')$, for all $\Bmf(t,t')\in\hat
  q_i$. By Lemma~\ref{lem:witness-sequence},
  there are witnessing sequences for each $\Bmf(t,t')\in \hat q_i$
  satisfying Conditions~\textit{(a)}--\textit{(c)}. Guided by these
  sequences, we construct an accepting run of $\Amc_\vp$.
  Throughout the construction of this run, some invariants are preserved. 
  First, whenever the automaton visits a node $w\in T$ in a state $\langle
  p,V_l,V_r\rangle$, then 
  \begin{enumerate}[label=(I\arabic*),align=left,leftmargin=*]

    \item $\mn{var}(p)\cup \{x\mid (d,s)\to_\Bmf x\in V_l\}\cup
      \{x\mid x\to_\Bmf (d,s)\in V_r\}$ is the set of all variables
      $x$ such that the image of $x$ under $\pi$ is below or in $w$;

    \item $(d,s)\to_\Bmf x\in V_l$ implies $d\in
      \Delta^{\Imc_w}$ and,
      in the sequence for $\Bmf$, there is some $i$ such that
      $d_i=d$, $s_i=s$ and for all $j\geq i$, $w_j$ is below or equal
      $w$;

    \item $x\to_{\Bmf} (d,s)\in V_r$ implies $d\in\Delta^{\Imc_w}$ and
      in the sequence for $\Bmf$, there is some $i$ such that
      $d_i=d$, $s_i=s$ and for all $j\leq i$, $w_j$ is below or equal
      $w$.

  \end{enumerate}
  Moreover, if the automaton visits a node $w\in T$ in state
  $\langle d,s,\Bmf,d',s'\rangle$ then
  \begin{itemize}[align=left,leftmargin=*]

    \item[(I4)] $d,d'\in\Delta^{\Imc_w}$ and there are $i<j$ such
      that, in the witnessing sequence for $\Bmf$, we have
      $d_{i}=d$, $s_{i}=s$, $d_j=d'$, $s_j=j$, and
      $w\in [w_i]_{d}\cap [w_j]_{d'}$.

  \end{itemize}

  Throughout the definition of the run, we use $s_{\Bmf f}$ to refer
  to the state $s_n$ in the witnessing sequence for each
  $\Bmf(t,t')\in \hat q_i$.

  The automaton starts in state $q_0$ and chooses to proceed in $\hat
  q_i$.  Define $Q_0,V_l,V_r$ by taking
  \begin{align*}
    Q_0 & = \{ \langle a,s_{0\Bmf},\Bmf,b,s_{\Bmf f}\rangle \mid
  \Bmf(a,b)\in \hat q_i \} \\
    V_l & = \{ (a,s_{0\Bmf })\to_\Bmf x \mid \Bmf(a,x)\in \hat q_i\}
    \\
    V_r & = \{ x\to_\Bmf (a,s_{\Bmf f}) \mid \Bmf(x,a)\in \hat q_i\}
  \end{align*}
  and extend the run according to this (possible) choice of
  $Q_0,V_l,V_r$. The invariants are obviously true after these first
  transitions. Assume now that $\Amf_\vp$ visits $w$ in state $\langle
  p,V_l,V_r\rangle$ and let $S=\mn{var}(p)\cup \{x\mid (d,s)\to_\Bmf
  x\in V_l\}\cup \{x\mid x\to_\Bmf (d,s)\in V_r\}$. First, define a partition
  $S_0,\ldots,S_k$ as follows: $S_0$ contains all $x\in S$ such that
  $\pi(x)=[w]_d$ for some $d\in \Delta^{\Imc_w}$; denote this witness
  $d$ with $d_x$. $S_i$ contains all
  $x\in S$ such that $\pi(x)=[v]_d\in \Delta^{\Imc_v}$ for some
  $v$ in the subtree rooted at $w\cdot i$.  Then, define a tuple
  $(Q^0,p^1,V_l^1,V_r^1\ldots,\ldots,V_l^k,V_r^k)$ as follows:
  \begin{itemize}

    \item[--] For $\Bmf(x,y)\in p$ with $\{x,y\}\subseteq S_0$, add
      $\langle d_x,s_{0\Bmf },\Bmf,d_y,s_{\Bmf f}\rangle$ to
      $Q^0$.

    \item[--] Let $\Bmf(x,y)\in p$ with $x\in S_0, y\in S_i$ for
      $i>0$. Read off
      from the witnessing sequence for $\Bmf(x,y)$ the maximal $\ell$ such that
      $w_j=w$ for all $j\leq \ell$, and add $\langle d_x,s_{\Bmf
      0},\Bmf,d_{\ell+1},s_{\ell+1}\rangle\in Q^0$ and
      $(d_{\ell+1},s_{\ell+1})\to_{\Bmf} y\in V_l^i$.

    \item[--] Let $\Bmf(x,y)\in p$ with $y\in S_0$, $x\in S_i$ for
      $i>0$. Read off
      from the witnessing sequence for $\Bmf(x,y)$ the minimal $\ell$ such that
      $w_j=w$ for all $j\geq \ell$, and add $\langle
      d_\ell,s_\ell,\Bmf,d_{y},s_{n}\rangle\in Q^0$ and $x\to_\Bmf
      (d_\ell,s_\ell)\in V_r^i$.

    \item[--] Let $\Bmf(x,y)\in p$ with $x,y\in S_i$ for $i>0$. Then
      add $\Bmf(x,y)\in p^i$.

    \item[--] Let $\Bmf(x,y)\in p$ with $x\in S_i,y\in S_j$ for
      $i,j>0$ and $i\neq j$. By definition of $S_i$ and $S_j$, there
      have to be indices $l<u\in[0,n]$ such that $w_m$ is below $w$,
      for all $m<l$ and all $m\geq l$.  Add $\langle
      d_{l},s_l,\Bmf,d_{u},s_{u}\rangle$ to $Q^0$,
      $(d_{d_{u},s_{u}})\to_\Bmf y$ to $V_l^j$, and
      $x\to_\Bmf(d_l,s_l)$ to $V_r^i$.

   \item[--] Let $(d,s)\to_\Bmf x\in V_l$. We distinguish two cases: 
      \begin{itemize}

	\item if $x\in S_0$, then add $\langle d,s,\Bmf,d_x,d_{\Bmf
	  f}\rangle$ to $Q^0$;

	\item if $x\in S_i$ for $i>0$, then by invariant (I2), there
	  is some $\ell$ such that $d_\ell=d$, $s_\ell=s$, and for all
	  $j> \ell$,
	  $w_j$ is below or equal $w$. By Condition~\textit{(c)},
	  there have to be $j>\ell$ and a successor $v$ of $w$ such that
	  $d_j\in\Delta^{\Imc_w}\cap\Delta^{\Imc_v}$, and for all
	  $j'>j$, we have that $w_j$ is below or equal $v$. Add
	  $\langle d,s,\Bmf,d_j,s_j\rangle$ to $Q^0$ and
	  $(d_j,s_j)\to_\Bmf x$ to $V_l^i$.

      \end{itemize}

    \item[--] Let $x\to_\Bmf (d,s)\in V_r$. We distinguish two cases: 
      \begin{itemize}

	\item if $x\in S_0$, then add $\langle d_x,d_{\Bmf
	  0},\Bmf,d,s\rangle$ to $Q^0$;

	\item if $x\in S_i$ for $i>0$, then by invariant (I3), there
	  is some $\ell$ such that $d_\ell=d$, $s_\ell=s$, and for all
	  $j<\ell$, $w_j$ is below or equal $w$. By
	  Condition~\textit{(c)}, there have to be $j<\ell$ and a
	  successor $v$ of $w$ such that
	  $d_j\in\Delta^{\Imc_w}\cap\Delta^{\Imc_v}$, and for all
	  $j'<j$, we have that $w_j$ is below or equal $v$. Add
	  $\langle d_j,s_j,\Bmf,d,s\rangle$ to $Q^0$ and
	  $(d_j,s_j)\to_\Bmf x$ to $V_r^i$.

      \end{itemize}
 
  \end{itemize}

  We then extend the run by the constructed tuple in the
  non-deterministic choice in the definition of $\delta(\langle
  p,V_l,V_r\rangle,(\Imc,x))$. It is routine to verify that the
  invariants remain true.

  It thus remains to show how to complete the run when the automaton visits
  a node $w$ in state $\langle d,s,\Bmf,d',s'\rangle$. By~(I4), we know
  that there are $i<j$ such that, in the witnessing sequence for
  $\Bmf$, we have $d_{i}=d$, $s_{i}=s$, $d_j=d'$, $s_j=j$,
  and $w\in [w_i]_{d_{i}}\cap [w_j]_{d_j}$.

  By Lemma~\ref{lem:correct-qa}, either $j=i+1$ or there is an $i<m<j$
  such that there is some $\hat w\in [w_m]_{d_m}\cap [w_i]_{d_i}\cap
  [w_j]_{d_j}$. In the first case, we extend the run such that the
  automaton visits $w_j$ in $\langle d,s,\Bmf,d',s'\rangle$ and
  accepts, because of~(c). Otherwise, we extend the run by navigating
  the automaton in state $\langle d,s,\Bmf,d',s'\rangle$ to node $\hat
  w$, and apply the intersection transition to $d''=d_m$ and
  $s''=s_m$. It should be clear that~(I4) remains preserved. 
  
  Since the witness sequences are finite, this process terminates
  after a finite number of steps and the constructed run thus
  satisfies the parity condition.  
\end{proof}

%\section{Putting all Together}
%
%\textcolor{blue}{
%Small paragraph repeating  how the designed automata are used (general picture) plus complexity bounds}

\end{document}